%% file: main.tex
\documentclass[acmsmall,screen,nonacm]{acmart}

\acmJournal{PACMPL}
\acmVolume{11}
\acmNumber{POPL}              
\acmArticle{1}
\acmYear{2027}
\acmMonth{1}
\acmDOI{}                     
     \usepackage{centernot}
     \usepackage[most]{tcolorbox}      
     \usepackage{xcolor}

\makeatletter
\def\@acmplainindent{0pt}
\def\@acmdefinitionindent{0pt}
\def\@proofindent{\noindent}
\makeatother

\usepackage{xspace}
\usepackage{nicefrac}
\usepackage{aliascnt}
\usepackage{cleveref}\usepackage{enumitem}
\usepackage{subcaption}
\usepackage{csquotes}
\AtEndPreamble{%

  \theoremstyle{acmplain}%
  \newaliascnt{lemma}{theorem}%
  \newtheorem{lemma}[lemma]{Lemma}\aliascntresetthe{lemma}%
  \newaliascnt{corollary}{theorem}%
  \newtheorem{corollary}[corollary]{Corollary}\aliascntresetthe{corollary}%
  \newaliascnt{proposition}{theorem}%
  \newtheorem{proposition}[proposition]{Proposition}\aliascntresetthe{proposition}%
  \theoremstyle{acmdefinition}%
  \newaliascnt{definition}{theorem}%
  \newtheorem{definition}[definition]{Definition}\aliascntresetthe{definition}%
  \newaliascnt{example}{theorem}%
  \newtheorem{example}[example]{Example}\aliascntresetthe{example}%
}
\usepackage[]{todonotes} 
\usepackage{bbold}
\usepackage{array}
\usepackage{placeins}
\usepackage{tikz}
\usepackage{mathpartir}
\usepackage{mathtools}
\usepackage[ruled,vlined,linesnumbered]{algorithm2e}
\usepackage{xcolor}
\usepackage{scalerel,stackengine}
\stackMath
\newcommand\reallywidehat[1]{%
\savestack{\tmpbox}{\stretchto{%
  \scaleto{%
    \scalerel*[\widthof{\ensuremath{#1}}]{\kern-.6pt\bigwedge\kern-.6pt}%
    {\rule[-\textheight/2]{1ex}{\textheight}}
  }{\textheight}%
}{0.5ex}}%
\stackon[1pt]{#1}{\tmpbox}%
}
\usetikzlibrary{arrows.meta,backgrounds,fit,positioning,shapes.geometric}

\input{marcos}
\usepackage{subcaption}
\allowdisplaybreaks
\begin{document}

\title{A Fast Quantitative Analyzer for NetKAT}

\author{Thomas Lu}
\orcid{0009-0000-3474-6263}
\affiliation{%
  \institution{Cornell University}
  \city{Ithaca}
  \state{New York}
  \country{USA}
}
\email{cl2625@cornell.edu}
\author{Qiancheng Fu}
\orcid{0000-0002-5234-8565}
\affiliation{%
  \institution{Cornell University}
  \city{Ithaca}
  \state{New York}
  \country{USA}
}
\email{qf59@cornell.edu}
\author{Kevin Batz}
\orcid{0000-0001-8705-2564}
\affiliation{%
  \institution{Cornell University}
  \city{Ithaca}
  \state{New York}
  \country{USA}
}
\email{ksb239@cornell.edu}
\author{Oliver Bøving}
\orcid{0009-0009-4702-2876}
\affiliation{%
  \institution{Technical University of Denmark}
  \city{Copenhagen}
  \country{Denmark}
}
\email{oembo@dtu.dk}
\author{Tiago Ferreira}
\orcid{0000-0002-6942-0228}
\affiliation{%
  \institution{University College London}
  \city{London}
  \country{United Kingdom}
}
\email{t.ferreira@ucl.ac.uk}
\author{Mark Moeller}
\orcid{0009-0002-9512-565X}
\affiliation{%
  \institution{Cornell University}
  \city{Ithaca}
  \state{New York}
  \country{USA}
}
\email{moeller@cs.cornell.edu}
\author{Nate Foster}
\orcid{0000-0002-6557-684X}
\affiliation{%
  \institution{EPFL}
  \city{Lausanne}
  \country{Switzerland}
}
\affiliation{%
  \institution{Jane Street}
  \country{USA}
}
\email{nate.foster@epfl.ch}
\author{Alexandra Silva}
\orcid{0000-0001-5014-9784}
\affiliation{%
  \institution{Cornell University}
  \city{Ithaca}
  \state{New York}
  \country{USA}
}
\email{alexandra.silva@cornell.edu}

\renewcommand{\shortauthors}{Lu et al.}

\begin{abstract}
When designing a network, engineers must navigate trade-offs (e.g., one topology offers more aggregate bandwidth, another lower latency or better resilience) that demand reasoning about {\em quantitative properties}. We present a fast analyzer for quantitative network properties based on weighted \netkat (\wnetkat), a domain-specific language that provides a semantic foundation for quantitative reasoning by modeling network behavior using weights drawn from a semiring. At the core of our development is the design of a symbolic data structure---weighted symbolic packet programs (wSPPs)---that compactly represent the semantics of weighted policies, for which a direct implementation would be intractable. We show how to compute all policy constructs symbolically; unsurprisingly, the crux is Kleene star, for which we design a tailored algorithm. We further develop trace-carrying Pareto semirings, which compute multi-objective frontiers together with the network paths that realize them. We formalize the development in Lean and provide an optimized Rust implementation. Being parametric on a  semiring, our implementation covers both classical and quantitative analyses: we show that it is competitive with KATch, a heavily optimized Boolean-reachability verifier, and orders of magnitude faster than McNetKAT and Storm on probabilistic analyses. A case study comparing Fat-tree and Jellyfish data-center topologies shows the framework supports multi-objective design-time analysis.
\end{abstract}


\maketitle


\section{Introduction}\label{sec:intro}
\input{introduction}

\section{Weighted NetKAT, wSPPs, and Pareto Optimality: A Bird's Eye View}\label{sec:overview}
\input{overview}

\section{Weighted NetKAT}\label{sec:preliminaries}
\input{preliminaries}

\section{Weighted Symbolic Packet Programs (wSPPs)}\label{sec:syntax-semantics}
\input{syntax-semantics}

\section{Pareto Weights and Witnessing Traces}\label{sec:pareto}
\input{pareto}


\section{Implementation and Evaluation}\label{sec:implementation-evaluation}
\input{implementation}

\section{Related Work}\label{sec:related-work}
\input{related_work}

\section{Conclusion and Future Work}\label{sec:conclusion}
\input{conclusion}


\bibliographystyle{ACM-Reference-Format}
\bibliography{references}

\pagebreak
\appendix
\crefalias{section}{appendix}
\crefname{appendix}{Appendix}{Appendices}
\Crefname{appendix}{Appendix}{Appendices}

\section{Some Background on Bounded (Frontier) Semirings}

If $=$ and $\preceq$ (in \Cref{def:frontier-semiring}) are decidable and $\cdot$ computable on $T$, the
operations of $\Fr(T)$ are computable merge-and-discard-dominated
manipulations of finite sets. Antichains are moreover
\emph{canonical}, i.e., distinct antichains generate distinct lower sets, so
$\dcl{(-)} \colon \Fr(T) \to \Par(T)$ is injective and will serve as the
embedding $\emb$ of \Cref{def:embedding}; conditions (i) and (ii) are
routine. Condition (iii) is, however, delicate: the iteration defining $\cstar{F}$
must stabilize equal the infinite sum $\emb(F)^*$. This can, however, fail:

\begin{example}\label{example:pareto-arctic}
  Over the \emph{arctic} semiring $(\N\cup\{\pm\infty\},\max,+,-\infty,0)$
  of worst-case accumulated latency, the frontier $F = \{1\}$ has
  $L^k = \dcl\{k\}$ for $L \triangleq \dcl F$, so
  $L^* = \dcl\{0,1,2,\ldots\}$: the closure of an unbounded chain, which
  no finite set generates. No star operation on $\Fr(T)$ whatsoever can
  decode to $L^*$. In fact, this is the case whenever some element of $T$ has strictly
  increasing powers. \hfill $\triangle$
\end{example}

What rules this out is \emph{boundedness}, i.e., when \emph{iterating can always be
skipped}:

\begin{definition}\label{def:bounded}
  A partially ordered monoid (resp.\ an $\omega$-continuous semiring), is
  \emph{bounded} if its unit is the greatest element of its order.
\end{definition}

\begin{lemma}[\cite{mohri2002semiring, matrix-star-algo}]\label{lemma:bounded-star}
  In a bounded $\omega$-continuous semiring, $\wta^* = \semione$ for all
  $\wta$.
\end{lemma}

Boundedness holds for the semirings of \emph{best-case} reasoning
(bottleneck, Viterbi, tropical, security levels) and is preserved by
products. For bounded $T$ (e.g., our running $\Trop\times\Btl$) the unit
$\semione = \dcl\{\semione_T\} = T$ is the greatest element of $\Par(T)$,
so $\Par(T)$ is bounded as well and the frontier iteration of
\Cref{def:frontier-semiring} stabilizes immediately at
$\cstar{F} = \{\semione_T\}$. 

\section{Complete Algorithms for the wSPP Operations}\label[appendix]{sec:appendix-algorithms}
\input{appendix-algorithms}

\section{A Portfolio of Embeddings}\label[appendix]{sec:appendix-portfolio}\label[appendix]{sec:embedding}\label[appendix]{sec:embedding-portfolio}
\input{embedding}

\section{Synthetic Adversarial Benchmarks}\label[appendix]{sec:appendix-synthetic}
\input{synthetic}

\end{document}

%% file: marcos.tex
\newcommand{\tl}[1]{\todo[noinline,color=orange]{\footnotesize TL\@: #1}}

\newcommand{\katch}{\textsf{KATch}\xspace}

\newcommand{\kat}{\textsf{KAT}\xspace}
\newcommand{\netkat}{\textsf{NetKAT}\xspace}
\newcommand{\probnetkat}{\textsf{ProbNetKAT}\xspace}
\newcommand{\mcnetkat}{\textsf{McNetKAT}\xspace}
\newcommand{\wnetkat}{\textsf{wNetKAT}\xspace}

\newcommand{\wspptool}{\textsf{Neko}\xspace}

\newcommand{\semi}{\ensuremath{\mathcal{S}}}          
\newcommand{\semidom}{\ensuremath{S}}                 
\newcommand{\semizero}{\ensuremath{\mathbb{0}}}
\newcommand{\semione}{\ensuremath{\mathbb{1}}}
\newcommand{\semiord}{\ensuremath{\mathrel{\preceq}}}

\newcommand{\semisup}{\ensuremath{\bigsqcup}}
\newcommand{\N}{\mathbb{N}}
\newcommand{\wta}{\ensuremath{w}}                     
\newcommand{\wtb}{\ensuremath{w'}}

\newcommand{\csr}{\ensuremath{\mathcal{A}}}           
\newcommand{\csrset}{\ensuremath{A}}                  
\newcommand{\emb}{\iota}                              
\newcommand{\cstar}[1]{#1^{\circledast}}              
\newcommand{\Rext}{\mathbb{R}^{\infty}_{\geq 0}}      
\newcommand{\Qext}{\mathbb{Q}^{\infty}_{\geq 0}}      

\newcommand{\packets}{\ensuremath{\mathsf{Pk}}}
\newcommand{\Pk}{\packets}
\newcommand{\fields}{\ensuremath{\mathsf{F}}}
\newcommand{\values}{\ensuremath{\mathsf{Val}}}
\newcommand{\fielda}{\ensuremath{f}}
\newcommand{\pkta}{\ensuremath{\alpha}}
\newcommand{\pktb}{\ensuremath{\beta}}
\newcommand{\pktc}{\ensuremath{\gamma}}
\newcommand{\updatepkt}[3]{\ensuremath{#1 [ #2 \coloneq #3 ]}}
\newcommand{\dom}{\operatorname{dom}}

\newcommand{\pols}{\ensuremath{\mathsf{Pol}}}
\newcommand{\Pol}{\pols}                              
\newcommand{\dup}{\ensuremath{\mathsf{dup}}}
\newcommand{\tests}{\ensuremath{\mathsf{Pred}}}
\newcommand{\TRUE}{\ensuremath{\mathsf{true}}}
\newcommand{\FALSE}{\ensuremath{\mathsf{false}}}
\newcommand{\SKIP}{\ensuremath{\mathsf{skip}}}
\newcommand{\DROP}{\ensuremath{\mathsf{drop}}}
\newcommand{\EQ}[2]{\ensuremath{#1 = #2}}
\newcommand{\NOTEQ}[2]{\ensuremath{#1 \neq #2}}
\newcommand{\AND}[2]{\ensuremath{#1 \wedge #2}}
\newcommand{\OR}[2]{\ensuremath{#1 \vee #2}}
\newcommand{\NOT}[1]{\ensuremath{\neg #1}}
\newcommand{\ASSN}[2]{\ensuremath{#1 \leftarrow #2}}
\newcommand{\SEQN}{\ensuremath{\,;}}
\newcommand{\SEQ}[2]{\ensuremath{#1 \SEQN #2}}
\newcommand{\ADDN}{\ensuremath{\oplus}}
\newcommand{\ADD}[2]{\ensuremath{#1 \ADDN #2}}
\newcommand{\WEIGH}[2]{\ensuremath{#1 \odot #2}}
\newcommand{\ITER}[1]{\ensuremath{#1^{*}}}
\newcommand{\NFOLD}[2][n]{\ensuremath{{#2}^{(#1)}}}

\newcommand{\sem}[1]{\llbracket #1 \rrbracket}
\newcommand{\semdom}{\ensuremath{\mathcal{M}}}        
\newcommand{\norm}[1]{\lVert#1\rVert}                 

\newcommand{\WSPP}{\ensuremath{\textnormal{\textsf{wSPP}}}}
\newcommand{\wspp}{\ensuremath{\mathsf{wspp}}}        
\newcommand{\wsppadd}{\mathbin{\hat{\oplus}}}
\newcommand{\wsppmul}{\mathbin{\hat{;}}}
\newcommand{\wsppscale}{\mathbin{\hat{\odot}}}
\newcommand{\wsppmapadd}{\mathbin{\hat{\oplus}_{\!m}}}
\newcommand{\wsppmapmul}{\mathbin{\hat{;}_{m}}}
\newcommand{\wsppmapscale}{\mathbin{\hat{\odot}_{\!m}}}
\newcommand{\wsppstar}{\hat{*}}
\newcommand{\bigwsppmapadd}{\mathop{\hat{\bigoplus}_{\!m}}}
\newcommand{\rowmul}{\mathbin{\triangleright}}
\newcommand{\gbox}[1]{\begingroup
  \setlength{\fboxsep}{1pt}\colorbox{lightgray}{$\displaystyle #1$}\endgroup}

\newcommand{\Par}{\mathsf{Par}}                    
\newcommand{\Fr}{\mathsf{Fr}}                      
\newcommand{\dcl}[1]{{\downarrow}#1}               
\newcommand{\fmax}{\operatorname{max}}             
\newcommand{\trleq}{\preceq_{\mathsf{tr}}}         
\newcommand{\Trop}{\mathbb{T}}                     
\newcommand{\Btl}{\mathbb{B}}                      
\newcommand{\lett}[1]{\mathsf{#1}}                 

\newcommand{\ruleboxstyle}[2]{\begingroup
  \setlength{\fboxsep}{0.7pt}\setlength{\fboxrule}{0.4pt}%
  \fbox{$#1 #2$}\endgroup}
\newcommand{\rulebox}[1]{%
  \mathchoice
    {\ruleboxstyle{\displaystyle}{#1}}%
    {\ruleboxstyle{\textstyle}{#1}}%
    {\ruleboxstyle{\scriptstyle}{#1}}%
    {\ruleboxstyle{\scriptscriptstyle}{#1}}}

\newcommand{\strhat}[1]{\reallywidehat{#1}}

\newcommand{\dflt}[1]{\lfloor #1 \rfloor}
\newcommand{\fprod}{\textstyle\overset{\raisebox{-0.3ex}[0pt][0pt]{$\scriptscriptstyle\rightarrow$}}{\prod}}
\newcommand{\bprod}{\textstyle\overset{\raisebox{-0.3ex}[0pt][0pt]{$\scriptscriptstyle\leftarrow$}}{\prod}}
\newcommand{\listfprod}[1]{\textstyle\fprod #1}
\newcommand{\listbprod}[1]{\textstyle\bprod #1}

\newcommand{\ruleprod}{\mathop{\vcenter{\hbox{%
  \scalebox{0.76}{$\displaystyle\prod$}}}}}
\newcommand{\guardprod}[2]{%
  \bigl(\strhat{\mathop{\ruleprod}\limits_{\mathclap{\scriptscriptstyle #1}}\,#2}\bigr)}
\newcommand{\fieldsof}[1]{\ensuremath{\mathsf{fields}(#1)}}

\newcommand{\AlgPoint}[2]{%
  {\normalfont\footnotesize\bfseries\textcolor{#1}{\textsf{(#2)}}}}
\newcommand{\AlgPointStart}[2]{%
  \leavevmode\rule{0pt}{3.25ex}\makebox[0pt][l]{%
    \hspace*{-1.05em}\raisebox{1.9ex}[0pt][0pt]{\AlgPoint{#1}{#2}}}}
\newcommand{\AlgPointEnd}[2]{%
  \leavevmode\rule[-2.25ex]{0pt}{3.9ex}\makebox[0pt][l]{%
    \hspace*{-1.05em}\raisebox{-1.85ex}[0pt][0pt]{\AlgPoint{#1}{#2}}}}
\newcommand{\AlgPointEndLow}[2]{%
  \leavevmode\rule[-2.75ex]{0pt}{4.45ex}\makebox[0pt][l]{%
    \hspace*{-1.05em}\raisebox{-2.35ex}[0pt][0pt]{\AlgPoint{#1}{#2}}}}
\SetKwProg{Fn}{function}{}{}
\SetKwFunction{Star}{Star}
\SetKwFunction{FacDfltId}{FactorDefault}
\SetKwFunction{FacTestAsgn}{FactorTest}
\SetKwFunction{FacDfltAsgn}{FactorWrite}
\SetKw{KwAnd}{and}
\SetKw{KwwithKw}{with}
\SetKw{KwReturn}{return}

\definecolor{wsppfield}{HTML}{5B4BAA}
\definecolor{wsppbranch}{HTML}{1F77B4}
\definecolor{wsppmissing}{HTML}{C46A10}
\definecolor{wsppdefault}{HTML}{2E7D32}
\definecolor{lightgray}{RGB}{245,245,245}
\definecolor{headergray}{RGB}{220,220,220}

\usepackage{listings}
\usepackage{wrapfig}

\definecolor{leanKeyword}{RGB}{90,75,170}   
\definecolor{leanSort}{RGB}{31,119,180}     
\definecolor{leanComment}{RGB}{110,110,120}
\definecolor{leanString}{RGB}{160,60,60}

\lstdefinelanguage{lean}{
  morekeywords={inductive,structure,class,instance,def,abbrev,theorem,lemma,
    example,mutual,partial,noncomputable,where,with,match,fun,let,in,do,by,if,
    then,else,namespace,end,open,variable,variables,deriving,return,from,have,
    show,calc,termination_by,decreasing_by,extends,attribute,section,import},
  morekeywords=[2]{Type,Prop,Sort,Nat,Bool,Field,Value},
  sensitive=true,
  morecomment=[l]{--},
  morecomment=[s]{/-}{-/},
  morestring=[b]",
  literate=
    {→}{{$\rightarrow$}}1 {↦}{{$\mapsto$}}1 {⇒}{{$\Rightarrow$}}1
    {∀}{{$\forall$}}1 {∃}{{$\exists$}}1 {λ}{{$\lambda$}}1
    {∈}{{$\in$}}1 {∉}{{$\notin$}}1 {∅}{{$\emptyset$}}1
    {∧}{{$\wedge$}}1 {∨}{{$\vee$}}1 {¬}{{$\neg$}}1
    {≤}{{$\leq$}}1 {≥}{{$\geq$}}1 {≠}{{$\neq$}}1 {≡}{{$\equiv$}}1 {≈}{{$\approx$}}1
    {⊕}{{$\oplus$}}1 {⊙}{{$\odot$}}1 {⊗}{{$\otimes$}}1 {×}{{$\times$}}1
    {∪}{{$\cup$}}1 {∩}{{$\cap$}}1 {⊆}{{$\subseteq$}}1 {⊂}{{$\subset$}}1
    {∘}{{$\circ$}}1 {·}{{$\cdot$}}1 {∣}{{$\mid$}}1
    {⟨}{{$\langle$}}1 {⟩}{{$\rangle$}}1 {‖}{{$\Vert$}}1
    {⊤}{{$\top$}}1 {⊥}{{$\bot$}}1 {⊔}{{$\sqcup$}}1 {⊓}{{$\sqcap$}}1
    {⨆}{{$\bigsqcup$}}1 {∑}{{$\sum$}}1
    {ℕ}{{$\mathbb{N}$}}1 {ℤ}{{$\mathbb{Z}$}}1 {ℚ}{{$\mathbb{Q}$}}1 {ℝ}{{$\mathbb{R}$}}1
    {α}{{$\alpha$}}1 {β}{{$\beta$}}1 {γ}{{$\gamma$}}1
    {σ}{{$\sigma$}}1 {ρ}{{$\rho$}}1 {φ}{{$\varphi$}}1
}

\lstdefinestyle{leanstyle}{
  language=lean,
  basicstyle=\ttfamily\small,
  keywordstyle=\color{leanKeyword}\bfseries,
  keywordstyle=[2]\color{leanSort},
  commentstyle=\color{leanComment}\itshape,
  stringstyle=\color{leanString},
  showstringspaces=false,
  columns=fullflexible,
  keepspaces=true,
  inputencoding=utf8,
  extendedchars=true,
  upquote=true,
  breaklines=true,
  aboveskip=0.6em, belowskip=0.4em,
}

\DeclareMathOperator{\bighoplus}{\widehat{\bigoplus}}  

\newcommand{\keys}{\operatorname{keys}}                
\newcommand{\bighoplusmap}{\mathop{\widehat{\bigoplus}_{\!m}}}  

%% file: introduction.tex

Network verification is a success story for the programming languages community. The idea is to view the {\em network as a program} in a domain-specific language, which can be analyzed to establish properties of interest. Tools built on this idea are now routinely used in industry \cite{tiros2019cav,rcdc2019,batfish,batfish2023}, checking functional properties (e.g., reachability or isolation) to catch configuration errors early.

In many cases, however, the properties of interest to network operators are inherently quantitative. Operators care not only whether a packet can reach its destination, but how {\em quickly} it will arrive, how {\em much bandwidth} will be available, and the {\em probability} that it will be delivered in the presence of failures. Hence, a natural next step for network verification is to develop quantitative frameworks that can reason about these properties and others.

Quantitative reasoning is especially valuable early in the design process, when operators must compare alternatives and make tradeoffs (e.g., between topologies, routing policies, failure-recovery schemes, etc.). For instance, one design might provide more aggregate bandwidth but require packets to take longer paths, while another might minimize latency but concentrate traffic at bottlenecks and sacrifice resilience. So, the answer to ``which design is better?'' is rarely simple. Instead, network engineers must navigate a {\em frontier}, considering different tradeoffs in each design. 

These problems are not hypothetical: Amazon recently announced that they are using quasi-
\begin{wrapfigure}{r}{0.3\textwidth}  
\vspace{-10pt}
\includegraphics[width=.9\linewidth]{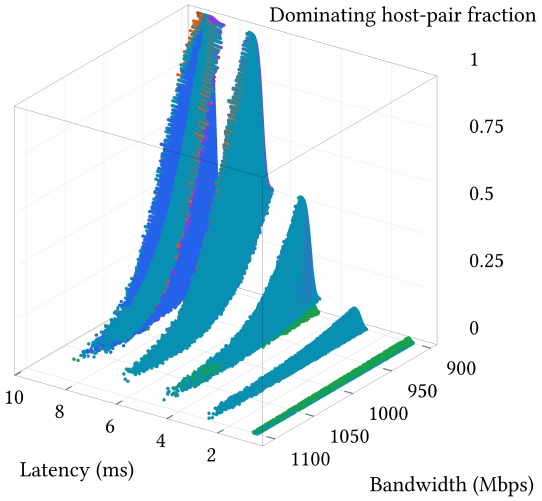}
\vspace{-8pt}
\end{wrapfigure}
random graphs to design new data centers~\cite{bernardi2026rng}, reviving academic ideas based on so-called Jellyfish topologies~\cite{singla2012jellyfish}. The increased path diversity in quasi-random graphs buys throughput and graceful degradation under failures, at the price of longer paths. An engineer comparing a quasi-random graph against a fabric based on a traditional fat-tree~\cite{singh2015clos} needs to reason across multiple objectives to determine which design suits their needs. Having a formalism to model and analyze the trade-offs of these  choices would provide assurances before large scale deployment. For instance, the graph on the right above (produced by our analyzer) provides validation that Jellyfish networks indeed enjoy improved per-server bandwidth and lower path latencies over comparable Fat-tree designs.  Additionally, we will show how one can also uncover other properties besides isolated latency and bandwidth by exploiting the fact that our framework is parametric in the choice of semiring.

Our analyzer is built using as core modeling language Weighted \netkat (\wnetkat)~\cite{wNetKAT}, which provides a semantic foundation for the style of quantitative reasoning we describe above. \wnetkat extends \netkat~\cite{netkat} with weights drawn from a semiring: every policy denotes a \emph{weighted} packet relation, and the choice of semiring determines the quantity being tracked---e.g., the tropical semiring for latency, the bottleneck semiring for bandwidth, the non-negative reals for probability, and so on. However, a direct implementation of \wnetkat's semantics would require representing a matrix with $\packets \times \packets$ entries over the semiring, which is impractical.

Our first core contribution is the development of \emph{weighted symbolic packet programs} (\WSPP{}s), a symbolic data structure that compactly represents the weighted relation denoted by a \wnetkat policy. This generalizes the data structure that powers \netkat's fastest verifier, \katch~\cite{KATch}, from Boolean to semiring weights. Policy combinators are computed analytically, using operations on the data structure itself. The most challenging aspect of our technical development is the treatment of the Kleene star. Whereas in \katch the star is computed as a straightforward fixpoint iteration over a finite lattice, in the weighted setting the corresponding chain of approximations may climb forever. Computing the star exactly requires factoring the iteration symbolically in a multi-pass algorithm that uses applications of algebraic identities in a strict order to ensure termination.  

Our second core contribution is the development of an extension to Pareto semirings~\cite{samborskii1992fourier} to keep track of both packet traces and Pareto-optimal weights. Intuitively, the weights track non-dominated tradeoffs made across multiple objectives while the trace witnesses the network path that achieves that weight. Notably, this construction recovers the per-path information that \netkat's $\dup$ primitive provides through its history-based semantics. Hence, working in the $\dup$-free fragment, with its substantially lighter semantics, does not sacrifice expressiveness.

Our third core contribution is a mechanization and implementation of our approach, in Lean and Rust. The Lean mechanization runs to ~53k lines of code and formalizes the \WSPP{} data structure, the star algorithm, the Pareto and trace-carrying constructions, and the algorithmic optimizations. The Rust implementation runs to ~9k lines of code and adds the sharing and caching optimizations that further improve performance. Our evaluation measures the performance of both implementations and compares with other state-of-the-art solvers, including  McNetKAT~\cite{smolka_scalable_2019} and Storm~\cite{hensel_probabilistic_2020}.

In summary, this paper makes the following contributions:
\begin{itemize}[leftmargin=*]
  \item We develop \WSPP{}s, a symbolic, representation for $\dup$-free \wnetkat policies, with lifted analytic constructions for all \wnetkat primitives (\Cref{sec:syntax-semantics}), including Kleene Star.
  \item We develop trace-carrying Pareto semirings, which compute multi-objective frontiers an absorption property that keeps the representation of Kleene star finite (\Cref{sec:pareto}).
  \item We present an implementation of \WSPP{}s and extensively evaluate its performance (\Cref{sec:implementation-evaluation}). Because \WSPP{}s are generic in the semiring, a single engine covers both and is competitive with \katch{}2~\cite{KATch2}, and orders of magnitude faster than McNetKAT and Storm.
  \item We present a case study on Fat-tree and Jellyfish data-center topologies in which we compare multiple objectives in a design-time analysis~\Cref{sec:implementation-evaluation}. We show that Jellyfish networks enjoy improved per-server bandwidth and lower path latencies over Fat-tree. More interestingly, our trace-carrying pareto semirings allow us to study tradeoffs between these two quantities, while highlighting how different routing methods affect the distribution of Pareto-optimal paths.
\end{itemize}

\noindent Our entire Lean development is available online in anonymized form,\footnote{\url{https://anonymous.4open.science/r/wnetkat-50F0}} and a checkpoint has been submitted as supplementary material. Some proofs were implemented with the assistance of Codex and Claude Code, and then manually scrutinized. 

Next, \Cref{sec:overview} provides a high-level overview to our approach, which we then detail in the core technical sections that follow.

%% file: overview.tex
Weighted NetKAT~\cite{wNetKAT} (\wnetkat) is a domain-specific programming language for reasoning about the quantitative behavior of networks and their forwarding policies. The syntax is concise\footnote{The reader familiar with \netkat will note the omission of \textsf{dup}. \Cref{sec:pareto-traces} shows how we do not loose reasoning power.}:
\[
\begin{array}{r@{~}c@{~}l}
  t,t' & ::= & \FALSE \mid \TRUE \mid \EQ{\fielda}{v} \mid \NOT{t} \mid \OR{t}{t'} \mid \AND{t}{t'} \\
  p,q & ::= & t \mid \ASSN{\fielda}{v} \mid \SEQ{p}{q} \mid \WEIGH{w}{p} \mid \ADD{p}{q} \mid \ITER{p}
\end{array}
\]
Here $t,t'$ are \emph{tests}, and $p,q$ are \emph{policies}. \emph{Atomic} policies are tests $t$ or assignments of \emph{values} $v$ to a \emph{field} $f$ (taken from some finite set). Policies can be composed sequentially ($\SEQ{p}{q}$), \emph{weighted} ($\WEIGH{w}{p}$, see below), or through a \emph{choice} operator ($\ADD{p}{q}$). $\ITER{p}$ is the \emph{Kleene star}, denoting, intuitively, an unbounded choice between doing nothing or executing $p$ an arbitrary number of times.
%
%

The expressivity of \wnetkat stems from its parametricity in an $\omega$-continuous \emph{semiring} $\semi = (\semidom, +, \cdot\,, \semizero, \semione)$ (cf.\ \Cref{sec:prelim-wnetkat}): Weights $w$ are taken from that semiring $\semi$, and $\semi$ determines the quantitative behavior that can be modeled. Semantically, every policy denotes a \emph{weighted packet relation}
$
  \sem{p} \colon \packets \to \packets \to \semidom,
$
where $\packets$ is the set of \emph{packets}---i.e., finite maps from fields to values. $\sem{p}(\pkta)(\pktb)$ is the \emph{weight} of executions starting in $\pkta$ and ending in $\pktb$, where:
\[
  \sem{\WEIGH{w}{p}}(\pkta)(\pktb) = w \cdot \sem{p}(\pkta)(\pktb)
  \qquad\text{and}\qquad
  \sem{\ADD{p}{q}}(\pkta)(\pktb) = \sem{p}(\pkta)(\pktb) + \sem{q}(\pkta)(\pktb)~
\]
That is, a weighting $\WEIGH{w}{p}$ means \enquote{weigh the execution of $p$ by $w$} and $\ADD{p}{q}$ means \enquote{add the result of the weights obtained from running $p$ and $q$, respectively}. For the Kleene star, we have
\[
    \sem{\ITER{p}}(\pkta)(\pktb) = \semione + \sem{p}(\pkta)(\pktb) + \sem{\ITER{\SEQ{p}{p}}}(\pkta)(\pktb) + \ldots ~,
\]
which reflects the unbounded choice between doing nothing ($\semione$) or executing $p$ arbitrarily often.
\begin{example}
    \label{ex:overview:wnetkat}
    Let $\semi = (\N\cup\{\pm\infty\},\max,+,-\infty,0)$ be the \emph{arctic semiring} and consider the policy
    \[
        p \quad{}={}\quad
        \ADD{
        \underbrace{
        \WEIGH{2}{(\EQ{\fielda}{1})}
        }_{\text{incur cost of $2$ if $\fielda=1$}}
        \qquad
        }{
        \qquad
          \underbrace{
            \ADD
          {\WEIGH{4}{(\ASSN{\fielda}{2})}~~~}
          {~~~\WEIGH{10}{(\ASSN{\fielda}{2})}}
          }_{
            \text{choose between incurring a cost of $4$ or $10$ when setting $\fielda$ to $2$}
          }
        }~.
    \]
This semiring models costs (think, e.g., latency), i.e., sequencing \emph{accumulates costs} and choice is resolved in a \emph{maximizing manner}. For instance, $\sem{p}(\{f \mapsto 1\})(\{f \mapsto 1\}) = 2$ due to the leftmost branch (the test $\EQ{f}{1}$ succeeds), $\sem{p}(\{f \mapsto 1\})(\{f \mapsto 3\}) = -\infty$ since there is no way to set $f$ to $3$, and the test on the left fails too which yields the input packet to be \enquote{dropped}. Lastly, $\sem{p}(\{f \mapsto 1\})(\{f \mapsto 2\}) = 10$ since the choice on the right is resolved in a cost \emph{maximizing} manner. If, on the other hand, we were to instantiate $\semi$ with the \emph{tropical semiring} $(\N\cup\{\infty\},\min,+,\infty,0)$, then $\sem{p}(\{f \mapsto 1\})(\{f \mapsto 2\}) = 4$ because choice would be resolved in a \emph{minimizing} manner. \hfill $\triangle$
\end{example}
\subsection{Weighted Symbolic Packet Programs: Efficient Symbolic Inference}
\input{wspps_overview}

\definecolor{ovevtcol}{RGB}{122, 58, 150}
\newcommand{\evt}[1]{\textcolor{ovevtcol}{\ensuremath{\mathsf{#1}}}}

\subsection{Weighted NetKAT for Trace Carrying, Pareto-Optimal Trade-Offs}
\label{sec:overview-pareto}

\input{fig-overview-pareto}

For a given network, the choice of semiring depends on the quantities that encode the relevant metrics and objectives. Over the
\emph{tropical semiring} $\Trop = (\N\cup\{\infty\},\min,+,\infty,0)$  from \Cref{ex:overview:wnetkat}, weights are latencies: choice keeps the faster
branch, sequencing adds delays. Over the \emph{bottleneck semiring} $\Btl = (\N\cup\{\pm\infty\},\max,\min,-\infty,\infty)$,
weights are bandwidths: choice keeps the wider branch and sequencing throttles
to the thinnest link. However, operators rarely optimize for just one objectives---they often need to reason about \emph{trade-offs} across multiple objectives.

Our first key insight is to leverage \emph{frontier semirings} \cite{geilen2007algebra, gondran2008graphs}: A general construction that lets us combine semirings, such as tropical and bottleneck, to reason about trade-offs. The weights are finite sets, called \emph{frontiers}, of vectors of the form $(\ell, w)$ (latency $\ell$, bandwidth $w$). All vectors in a frontier are \emph{pareto-optimal} in the sense that, for any given vector in the frontier, no other vector in the frontier is \emph{strictly better} in \emph{both} components (e.g. no lower latency \emph{and} wider bandwidth).

Concretely, consider the policy $p$ of \Cref{fig:overview-pareto}(a);
ignore the \textcolor{ovevtcol}{colored} event letters inside its weights
for the moment---they become relevant in \Cref{sec:overview-traces}, and
without them the leading factor of $p$ is simply $\SKIP$ (a no-op). A
backbone path crosses four aging spans: span $i$, entered at switch
$i{-}1$, still runs at $10 \cdot i$\,Gb/s, or traffic takes a leased
detour around it, over the modern $80$\,Gb/s express fabric, at the price
of $i$ extra milliseconds. The detours, however, are reserved for the
\emph{gold} service class: a second field
$\mathsf{cls} \in \{\mathsf{gold}, \mathsf{be}\}$ carries a packet's
class, the detour branches test it, and best-effort traffic may \emph{buy}
gold status at the ingress---an assignment $\ASSN{\mathsf{cls}}{\mathsf{gold}}$
whose weight charges one millisecond of authorization latency and confines
the flow to a $25$\,Gb/s premium tunnel. Which latency--bandwidth
trade-offs can traffic crossing the chain realize?

The semantics provides answers \emph{per packet pair}, with different pairs receiving
different frontiers. A gold route leases some set $S$ of detours: it
accumulates latency $\sum_{i \in S} i$ and is throttled by
$\min_{i \notin S} 10 \cdot i$---the thinnest span or to the
fabric's $80$\,Gb/s once every span is bypassed. Of these
sixteen routes, only five are pareto-optimal. Leasing a detour when a thinner span exists
essentially buys latency for nothing, while upgraded best-effort traffic receives the same
treatment below the tunnel cap:
\[
\begin{array}{@{}l@{\;\,}l@{}}
  \sem{p}(\{\mathsf{sw} \mapsto 0, \mathsf{cls} \mapsto \mathsf{gold}\})
         (\{\mathsf{sw} \mapsto 4, \mathsf{cls} \mapsto \mathsf{gold}\})
  &=\;
  \{(0,10),\, (1,20),\, (3,30),\, (6,40),\, (10,80)\}\,,\\[0.3em]
  \sem{p}(\{\mathsf{sw} \mapsto 0, \mathsf{cls} \mapsto \mathsf{be}\})
         (\{\mathsf{sw} \mapsto 4, \mathsf{cls} \mapsto \mathsf{gold}\})
  &=\;
  \{(1,10),\, (2,20),\, (4,25)\}\,.
\end{array}
\]
\Cref{fig:overview-pareto}(b)
plots both frontiers.

\subsection{Trace-Carrying Weights: Every Trade-Off with Its Witness}
\label{sec:overview-traces}

A frontier tells the operator \emph{what} is achievable, not \emph{how}:
which detours realize $(6, 40)$? We develop \emph{Trace-Carrying Frontier Semirings}, a novel structure that answers such questions \emph{inside the
weights}. Formally, a weight pairs each cost vector with a \emph{log}, and the
weighting operator doubles as a logging primitive, and, since our wSPP compiler is
semiring-agnostic, this enables the \emph{efficient, trace-carrying trade-off analysis}.

Now consider the colored event letters of \Cref{fig:overview-pareto}(a). They log the billable events: admitting a packet logs
$\evt{adm}$, buying gold logs $\evt{u}$, and leasing detour $i$ logs
$\evt{d_i}$; crossing a legacy span logs nothing. Read with its letters, the same policy
answers the same two queries with the same two staircases---but every
trade-off now ships with its witness, annotated in
\Cref{fig:overview-pareto}(b): the gold entry
$((6,40),\, \evt{adm}\,\evt{d_1 d_2 d_3})$
reads \emph{``lease the detours around spans $1$ to $3$,''} and the
best-effort entry $((4,25),\, \evt{adm}\,\evt{u}\,\evt{d_1 d_2})$ reads
\emph{``buy gold, then lease detours $1$ and $2$.''}

Trace-carrying pareto semirings have to be constructed with care: In the presence of unbounded iteration, frontiers might become infinite and we lose computably tractable inference. For this reason, we identify (i) algebraic properties of the underlying semirings and (ii) impose suitable order-theoretic notions so that frontiers are \emph{guaranteed to remain finite}, even in the presence of unbounded iteration. Intuitively, this is possible for all semirings where longer traces can only make things worse, so that ever-increasing traces can be omitted when tracking ways to reach pareto-optimality. Crucially, when there are \emph{multiple} (in an appropriate sense) \enquote{incomparable} traces giving raise to the \emph{same} pareto optimal vector, our construction \emph{tracks} all of them. We will demonstrate in \Cref{sec:case-study} how this fine-grained perspective can yield highly informative analysis results on the fault-tolerant achievability of pareto-optimal points.

\paragraph*{Fat Tree vs. Jellyfish Topologies}

As an example of how \wnetkat might be useful as a design aid, we designed a case study investigating the tradeoffs of two popular datacenter topology schemas.
Data center topology design has been relatively stable for a number of years, but is recently open for discussion again: Amazon has recently reported their preference for randomized topologies (based on ``Jellyfish'' networks~\cite{singla2012jellyfish}) over traditional Fat-trees~\cite{bernardi2026rng}). There is also a great deal of attention being paid towards optimizing topologies for machine learning loads~\cite{qian2024llm}.

To give an idea of the type of analysis we might want to do, \Cref{fig:topo-compare} shows an example of a Fat-tree and a Jellyfish network.
In order to keep the focus on the topology, each has the same host count, switch count, and per-switch port count.
By assigning appropriate bandwidth and latency costs to individual links, these models allow us to reason about the tradeoff space for the various host-to-host paths available under each topology and routing scheme.
In the figure, we show latencies and bandwidths for the two paths, and the example illustrates that even between the same two hosts, the situation might not be straightforward to compare. \wnetkat allows us to efficiently analyze the different quantitative tradeoff behavior of these networks, scaling our full case study in \Cref{sec:case-study} to larger networks with 250 hosts.

\begin{figure*}
    \centering
    \resizebox{.9\textwidth}{!}{\input{topology_comparison.tex}}
    \caption{Two paths between the same two hosts in two different topologies: A favors bandwidth, B favors latency; neither dominates the other.}
    \label{fig:topo-compare}
\end{figure*}

%% file: wspps_overview.tex
\newcommand{\srv}{\ensuremath{\mathsf{srv}}}
\newcommand{\rdy}{\ensuremath{\mathsf{rdy}}}
\newcommand{\bkf}{\ensuremath{\mathsf{bkf}}}
\newcommand{\pwait}{\ensuremath{\mathsf{wait}}}
\tikzset{
  ovstate/.style={draw, rounded corners=2.5pt, inner xsep=3pt,
    inner ysep=2.5pt, font=\scriptsize},
  ovlab/.style={font=\scriptsize, inner sep=1.5pt, fill=white},
  ovplate/.style={rounded corners=4pt, draw=black!30, fill=black!4},
  ovfield/.style={draw=wsppfield!80!black, circle, fill=wsppfield!15,
    minimum size=13pt, inner sep=0pt, font=\scriptsize},
  ovgate/.style={draw, diamond, minimum size=7pt, inner sep=0pt,
    fill=black!8},
  ovel/.style={font=\tiny, inner sep=1pt, fill=white},
  ovbbox/.style={draw=wsppbranch, fill=wsppbranch!5, fill opacity=0.55,
    rounded corners=2pt, inner sep=3pt},
  ovmbox/.style={draw=wsppmissing, fill=wsppmissing!5, fill opacity=0.55,
    rounded corners=2pt, inner sep=3pt},
  ovdbox/.style={draw=wsppdefault, fill=wsppdefault!5, fill opacity=0.55,
    rounded corners=2pt, inner sep=3pt},
}
\newsavebox{\ovpolbox}
\newcommand{\ovfitpol}[1]{%
  \sbox{\ovpolbox}{#1}%
  \ifdim\wd\ovpolbox>\linewidth
    \resizebox{\linewidth}{!}{\usebox{\ovpolbox}}%
  \else
    \usebox{\ovpolbox}%
  \fi}
\newcommand{\ovvp}{\vphantom{\displaystyle\sum_{j \geq 0} \bigl(\tfrac{1}{2}\bigr)^{j}}}

Weighted NetKAT (\wnetkat) is an expressive semiring-agnostic language, capable of modeling quantitative aspects of networks, ranging from single-objective reasoning to multi-objective trade-offs between, e.g., latencies and bottlenecks. Trace-carrying semiring constructions will additionally provide network operators with information on how pareto-optimal points are achieved.

The power of this expressiveness is only unleashed, however, when \emph{automatic inference} of these quantities, pareto frontiers, and witnesses is computationally tractable.  The core contribution of our paper is the design of   %
  \emph{a symbolic, semiring-agnostic data structure for weighted NetKAT,}  which we call weighted symbolic packet programs (wSPPs),   %
  \emph{alongside an efficient compiler from weighted NetKAT policies to wSPPs}.
wSPPs are a decision-tree-like data structure that enables a compact representation of the astronomical packet space of a \wnetkat policy. Using explicit enumeration to represent this packet space via, for instance, matrix representations of the weighted packet relation, would be computationally intractable. Further, our efficient compiler from \wnetkat policies to wSPPs enables an end-to-end framework for the semiring-agnostic automatic inference on \wnetkat policies. As a first test of scalability, we show that wSPPs, even though \emph{not} tailored to probabilistic networks, significantly outperform McNetKAT, a tool \emph{tailored} for inference on probabilistic networks, as well as state-of-the-art probabilistic model checkers.

While we took inspiration from NetKAT's efficient \katch~\cite{KATch} engine, the semiring setting poses several challenges that we highlight in the next sections, after introducing wSPPs by example.

\subsection{wSPPs by Example} In this section, we illustrate on the \emph{probability semiring} $([0,\infty], +, \cdot\,, 0, 1)$, suitable for modeling, e.g., faulty links or probabilistic routing strategies. Given a policy $p$ and a pair of input/output packets $\pkta,\pktb$, $\sem{p}(\pkta)(\pktb)$ is the probability of reaching packet $\pktb$ when executing $p$ on $\pkta$, and performing inference means determining these probabilities for \emph{all} such packet-pairs.

Concretely, consider one time slot of a retry protocol running at switch
$\EQ{\mathsf{sw}}{1}$, shown as the policy $p$ in
\Cref{fig:overview-policy}: a field $\mathsf{sw}$ holds a packet's
location and a field $\mathsf{st} \in \{\rdy, \pwait\}$ its transmission
state, and in each slot a ready packet is either transmitted (handed to
switch $2$), put on hold, or dropped, while a waiting packet keeps
waiting, gets ready again, or is likewise dropped. The guards on
$\mathsf{st}$ are disjoint, so the outer $\ADDN$ is a case split, and
within each case the branch weights form a convex combination, i.e., a
probabilistic choice. The central question is: Given input values/output values for $\mathsf{sw},\mathsf{st}$, in the form of input/output packets $\pkta,\pktb$, what is the probability $\sem{p}(\pkta)(\pktb)$ of ending up in $\pktb$ when starting in $\pkta$? We answer this question by \emph{compiling $p$ to a wSPP}.

\Cref{fig:overview-policy} also shows the wSPP $\hat{p}$ that our compiler produces for $p$. Reading a wSPP is a walk from the root towards a leaf, one field at a time. A circle labeled with a field $\fielda$ branches first on the value the \emph{input} packet $\pkta$ carries in field $\fielda$ (the dashed arrow being a \enquote{default edge} standing for all values \emph{not listed explicitly}) and, at the diamond
thus reached, then on the value the \emph{output} packet carries $\pktb$ (this time the \enquote{default edge} says that the input and the output value of the respective field agree). The weight of a
concrete input-output pair of packets is read off in one such walk: if it
ends in a leaf $\rulebox{w}$, the weight is $w$. In \Cref{fig:overview-policy}(b), for
instance, input $\{\mathsf{sw} \mapsto 1,\, \mathsf{st} \mapsto \rdy\}$ and
output $\{\mathsf{sw} \mapsto 2,\, \mathsf{st} \mapsto \rdy\}$ descend
along $1, 2, \rdy, \rdy$ to the leaf $\rulebox{\nicefrac{1}{4}}$, whereas every input at a switch other than $1$ falls into the dashed branch and meets the leaf $\rulebox{0}$ since $p$ drops such packets.
\input{fig-overview-policy}

\subsection{The Challenge: Dealing with Unbounded Weighted Iteration}
\label{sec:wspp:challenge}
The main challenge for wSPPs is dealing with \emph{unbouded iteration}, i.e., weighted NetKAT's Kleene star. To illustrate,
consider the policy $q$ of \Cref{fig:overview-iterated}(a), which runs $p$ from \Cref{fig:overview-policy}(a) for \emph{any
number of times}, then discards everything still at switch $1$. With
which probability does $q$ deliver a packet---that is, produce the output
$\pktb \triangleq \{\mathsf{sw} \mapsto 2,\, \mathsf{st} \mapsto \rdy\}$?
For inputs at switch $1$, the answer is an \emph{infinite and nested geometric sum over all delivery
runs}, parametric in the initial transmission state $\mathsf{st}$. A ready packet is
delivered by $k \in \N$ failed rounds---wait, keep waiting for some $j$
slots, get ready again---followed by one successful transmission,
\[
  \sem{q}(\{\mathsf{sw} \mapsto 1,\, \mathsf{st} \mapsto \rdy\})(\pktb)
  \;=\;
  \sum_{k \geq 0}\,\Bigl(\;
    \underbrace{\tfrac{1}{2}\ovvp}_{\textnormal{\footnotesize wait}} \cdot
    \underbrace{\displaystyle\sum_{j \geq 0} \bigl(\tfrac{1}{2}\bigr)^{j}}_{\textnormal{\footnotesize keep waiting}} \cdot\,
    \underbrace{\tfrac{1}{4}\ovvp}_{\textnormal{\footnotesize get ready}}
  \;\Bigr)^{\!k} \cdot
  \underbrace{\tfrac{1}{4}\ovvp}_{\textnormal{\footnotesize transmit}}
  \;=\;
  \sum_{k \geq 0} \bigl(\tfrac{1}{4}\bigr)^{k} \cdot \tfrac{1}{4}
  \;=\;
  \tfrac{1}{3}\,,
\]
while a waiting packet must first re-enter the ready queue and is then
delivered as above:
\[
  \sem{q}(\{\mathsf{sw} \mapsto 1,\, \mathsf{st} \mapsto \pwait\})(\pktb)
  \;=\;
  {\displaystyle\sum_{j \geq 0} \bigl(\tfrac{1}{2}\bigr)^{j}}\cdot\,
  {\tfrac{1}{4}\ovvp} \cdot\,
  {\tfrac{1}{3}\ovvp}
  \;=\;
  \tfrac{1}{6}\,.
\]

%
Dealing with this phenomenon, in a manner that keeps wSPPs concise and their compilation efficient, is non-trivial. First, observe that \katch's ~\cite{KATch} way of dealing with iteration is not an option: In \katch, the Kleene star is simply unfolded \emph{finitely often} until a fixed point is reached. In the purely Boolean setting, this terminates since the underlying lattice is \emph{finite}. For us, the infinite sums depicted above already show that finite unfoldings of the star only ever correspond to \emph{partial sums}, and these partial sums converge \emph{only in the limit} of sought-after probabilities.

Therefore, we have to take a different, significantly more complex approach (\Cref{sec:wspp-star}). In a nutshell, we exploit the \emph{equational properties} of \wnetkat to obtain a \emph{symbolic state elimination algorithm}, \enquote{denesting} the Kleene star in a modular way, similarly to how one computes regular expressions from finite automata. For efficiency, we have to carefully preserve \emph{canonicity properties} of wSPPs. These techniques pay off: our experiments show that our compiler significantly outperforms techniques tailored to probabilistic inference, despite being parametric in the choice of semiring. 

\input{fig-overview-iterated}


%% file: fig-overview-policy.tex
\begin{figure}[!tb]
\centering
\begingroup
\footnotesize
\begin{subfigure}[b]{0.52\textwidth}
\centering
\ovfitpol{$
\begin{array}{@{}l@{}}
  p \;\triangleq\; \EQ{\mathsf{sw}}{1} \SEQN \bigl(\\[0.5em]
  \;\;\;\;
    \EQ{\mathsf{st}}{\rdy} \SEQN \bigl(\,
      \overbrace{\WEIGH{\nicefrac{1}{4}}{(\ASSN{\mathsf{sw}}{2})}}^{\textnormal{\scriptsize\color{black!55} transmit}}
      \;\ADDN\;
      \overbrace{\WEIGH{\nicefrac{1}{2}}{(\ASSN{\mathsf{st}}{\pwait})}}^{\textnormal{\scriptsize\color{black!55} wait}}
      \;\ADDN\;
      \overbrace{\WEIGH{\nicefrac{1}{4}}{\DROP}}^{\textnormal{\scriptsize\color{black!55} drop}}
    \,\bigr)
    \;\ADDN\\[0.6em]
  \;\;\;\;
    \EQ{\mathsf{st}}{\pwait} \SEQN \bigl(\,
      \underbrace{\WEIGH{\nicefrac{1}{2}}{\SKIP}}_{\textnormal{\scriptsize\color{black!55} keep waiting}}
      \;\ADDN\;
      \underbrace{\WEIGH{\nicefrac{1}{4}}{(\ASSN{\mathsf{st}}{\rdy})}}_{\textnormal{\scriptsize\color{black!55} get ready}}
      \;\ADDN\;
      \underbrace{\WEIGH{\nicefrac{1}{4}}{\DROP}}_{\textnormal{\scriptsize\color{black!55} drop}}
    \,\bigr)\bigr)
\end{array}
$}
\vspace{1.1em}
\end{subfigure}%
\hfill
\begin{subfigure}[b]{0.45\textwidth}
\centering
\ovfitpol{%
\begin{tikzpicture}[>=Latex, line join=round,scale=.8]
  \node[ovfield] (sw) at (2.7,4.2) {$\mathsf{sw}$};
  \node[ovgate] (gb) at (1.55,3.4) {};
  \node[ovgate] (gd) at (4.85,3.4) {};
  \node[font=\scriptsize, inner sep=1pt] (lz0) at (4.85,2.6)
    {$\rulebox{0}$};
  \node[ovfield] (stm) at (0.35,2.45) {$\mathsf{st}$};
  \node[ovfield] (std) at (2.75,2.45) {$\mathsf{st}$};
  \node[ovgate] (g1) at (0.0,1.6) {};
  \node[ovgate] (gm) at (0.95,1.6) {};
  \node[font=\scriptsize, inner sep=1pt] (l1) at (0.0,0.85)
    {$\rulebox{\nicefrac{1}{4}}$};
  \node[ovgate] (g2) at (1.95,1.6) {};
  \node[ovgate] (g3) at (3.15,1.6) {};
  \node[ovgate] (gn) at (4.15,1.6) {};
  \node[font=\scriptsize, inner sep=1pt] (l2) at (1.95,0.85)
    {$\rulebox{\nicefrac{1}{2}}$};
  \node[font=\scriptsize, inner sep=1pt] (l3) at (2.8,0.85)
    {$\rulebox{\nicefrac{1}{2}}$};
  \node[font=\scriptsize, inner sep=1pt] (l4) at (3.55,0.85)
    {$\rulebox{\nicefrac{1}{4}}$};
  \node[font=\scriptsize, inner sep=1pt] (lz1) at (0.95,0.85)
    {$\rulebox{0}$};
  \node[font=\scriptsize, inner sep=1pt] (lz2) at (4.15,0.85)
    {$\rulebox{0}$};
  \draw[->] (sw) -- node[ovel,left=2pt] {$1$} (gb);
  \draw[->,dashed] (sw) -- node[ovel,right=2.5pt,font=\scriptsize] {$*$} (gd);
  \draw[->,dashed] (gd) -- (lz0);
  \draw[->,dashed] (gm) -- (lz1);
  \draw[->,dashed] (gn) -- (lz2);
  \draw[->] (gb) -- node[ovel,left=2pt] {$2$} (stm);
  \draw[->] (gb) -- node[ovel,right=2pt] {$1$} (std);
  \draw[->] (stm) -- node[ovel,left=1.5pt] (e1) {$\rdy$} (g1);
  \draw[->,dashed] (stm) -- node[ovel,pos=0.55,right=2pt,font=\scriptsize]
    (em) {$*$} (gm);
  \draw[->] (g1) -- node[ovel,left=1.5pt] (e2) {$\rdy$} (l1);
  \draw[->] (std) -- node[ovel,left=2.5pt] (e3) {$\rdy$} (g2);
  \draw[->] (std) -- node[ovel,left=2pt] (e4) {$\pwait$} (g3);
  \draw[->,dashed] (std) -- node[ovel,pos=0.55,above=2pt,font=\scriptsize]
    (en) {$*$} (gn);
  \draw[->] (g2) -- node[ovel,left=1.5pt] (e5) {$\pwait$} (l2);
  \draw[->] (g3) -- node[ovel,left=2pt] (e6) {$\pwait$} (l3);
  \draw[->] (g3) -- node[ovel,right=1.5pt] (e7) {$\rdy$} (l4);
  \begin{scope}[on background layer]
    \node[ovmbox, fit=(stm)(g1)(gm)(l1)(lz1)(e1)(e2)(em)] {};
    \node[ovdbox, fit=(std)(g2)(g3)(gn)(l2)(l3)(l4)(lz2)(e3)(e4)(e5)(e6)(e7)(en)] {};
    \node[ovbbox, fit=(gd)(lz0)] {};
  \end{scope}
\end{tikzpicture}}
\end{subfigure}
\endgroup
\caption{\footnotesize A probabilistic policy $p$ (left) and its compiled wSPP $\hat{p}$ (right).}
\label{fig:overview-policy}
\end{figure}

%% file: fig-overview-iterated.tex
\begin{figure}[tb]
\centering
\begingroup
\footnotesize
\begin{subfigure}[b]{0.55\textwidth}
\centering
\ovfitpol{$
\begin{array}{@{}l@{}}
  q \;\triangleq\; \Bigl(\, \EQ{\mathsf{sw}}{1} \SEQN \bigl(\\[0.5em]
  \;\;\;\;\;\;
    \EQ{\mathsf{st}}{\rdy} \SEQN \bigl(\,
      \overbrace{\WEIGH{\nicefrac{1}{4}}{(\ASSN{\mathsf{sw}}{2})}}^{\textnormal{\scriptsize\color{black!55} transmit}}
      \;\ADDN\;
      \overbrace{\WEIGH{\nicefrac{1}{2}}{(\ASSN{\mathsf{st}}{\pwait})}}^{\textnormal{\scriptsize\color{black!55} wait}}
      \;\ADDN\;
      \overbrace{\WEIGH{\nicefrac{1}{4}}{\DROP}}^{\textnormal{\scriptsize\color{black!55} drop}}
    \,\bigr)
    \;\ADDN\\[0.6em]
  \;\;\;\;\;\;
    \EQ{\mathsf{st}}{\pwait} \SEQN \bigl(\,
      \underbrace{\WEIGH{\nicefrac{1}{2}}{\SKIP}}_{\textnormal{\scriptsize\color{black!55} keep waiting}}
      \;\ADDN\;
      \underbrace{\WEIGH{\nicefrac{1}{4}}{(\ASSN{\mathsf{st}}{\rdy})}}_{\textnormal{\scriptsize\color{black!55} get ready}}
      \;\ADDN\;
      \underbrace{\WEIGH{\nicefrac{1}{4}}{\DROP}}_{\textnormal{\scriptsize\color{black!55} drop}}
    \,\bigr)\bigr)
  \hphantom{q \;\triangleq\;} \Bigr)^{\!*} \SEQN\; \NOTEQ{\mathsf{sw}}{1}
\end{array}
$}
\vspace{0.9em}
\end{subfigure}%
\hfill
\begin{subfigure}[b]{0.42\textwidth}
\centering
\ovfitpol{%
\begin{tikzpicture}[>=Latex, line join=round,scale=.7]
  \node[ovfield] (sw) at (1.7,4.2) {$\mathsf{sw}$};
  \node[ovgate] (gb) at (0.7,3.4) {};
  \node[ovgate] (gd) at (3.2,3.4) {};
  \node[ovfield] (tau) at (0.7,2.45) {$\mathsf{st}$};
  \node[font=\scriptsize, inner sep=1pt] (lone) at (3.2,2.45)
    {$\rulebox{1}$};
  \node[ovgate] (g1) at (0.1,1.6) {};
  \node[ovgate] (g2) at (1.35,1.6) {};
  \node[ovgate] (gm) at (2.3,1.6) {};
  \node[font=\scriptsize, inner sep=1pt] (l13) at (0.1,0.85)
    {$\rulebox{\nicefrac{1}{3}}$};
  \node[font=\scriptsize, inner sep=1pt] (l16) at (1.35,0.85)
    {$\rulebox{\nicefrac{1}{6}}$};
  \node[font=\scriptsize, inner sep=1pt] (lz) at (2.3,0.85)
    {$\rulebox{0}$};
  \draw[->] (sw) -- node[ovel,left=2pt] {$1$} (gb);
  \draw[->,dashed] (sw) -- node[ovel,right=2.5pt,font=\scriptsize] {$*$} (gd);
  \draw[->] (gb) -- node[ovel,left=1.5pt] {$2$} (tau);
  \draw[->,dashed] (gd) -- (lone);
  \draw[->] (tau) -- node[ovel,left=2pt] (e1) {$\rdy$} (g1);
  \draw[->] (tau) -- node[ovel,left=2pt] (e2) {$\pwait$} (g2);
  \draw[->,dashed] (tau) -- node[ovel,pos=0.55,above=2pt,font=\scriptsize]
    (e5) {$*$} (gm);
  \draw[->] (g1) -- node[ovel,left=1.5pt] (e3) {$\rdy$} (l13);
  \draw[->] (g2) -- node[ovel,left=1.5pt] (e4) {$\rdy$} (l16);
  \draw[->,dashed] (gm) -- (lz);
  \begin{scope}[on background layer]
    \node[ovmbox, fit=(tau)(g1)(g2)(gm)(l13)(l16)(lz)(e1)(e2)(e3)(e4)(e5)] {};
    \node[ovdbox, fit=(gd)(lone)] {};
  \end{scope}
\end{tikzpicture}}
\end{subfigure}
\endgroup
\caption{The policy $q$ (left) obtained from iterating the policy $p$ from \Cref{fig:overview-policy} and its wSPP $\hat{q}$ (right).}
\label{fig:overview-iterated}
\end{figure}

%% file: fig-overview-pareto.tex
\newsavebox{\ovparboxA}
\newcommand{\ovparfitA}[1]{%
  \sbox{\ovparboxA}{#1}%
  \ifdim\wd\ovparboxA>\linewidth
    \resizebox{\linewidth}{!}{\usebox{\ovparboxA}}%
  \else
    \usebox{\ovparboxA}%
  \fi}
\begin{figure}[tb]
\centering
\begingroup
\footnotesize
\begin{subfigure}[b]{0.56\textwidth}
\centering
\ovparfitA{$
\begin{array}{@{}l@{}}
  \mathit{upg} \;\triangleq\; \SKIP \,\ADDN\, \bigl(\,
    \EQ{\mathsf{cls}}{\mathsf{be}} \SEQN
    \overbrace{\WEIGH{\{((1,\, 25),\, \evt{u})\}}{(\ASSN{\mathsf{cls}}{\mathsf{gold}})}}^{\textnormal{\scriptsize\color{black!55} buy gold: $1$\,ms, $25$\,Gb/s tunnel}}
  \,\bigr)\\[1.1em]
  \mathit{seg}_i \;\triangleq\; \EQ{\mathsf{sw}}{i{-}1} \SEQN \bigl(\,
    \underbrace{\WEIGH{\{(0,\, 10 i)\}}{(\ASSN{\mathsf{sw}}{i})}}_{\textnormal{\scriptsize\color{black!55} cross legacy span $i$}}
    \,\ADDN\,
    \underbrace{\EQ{\mathsf{cls}}{\mathsf{gold}} \SEQN \WEIGH{\{((i,\, 80),\, \evt{d_i})\}}{(\ASSN{\mathsf{sw}}{i})}}_{\textnormal{\scriptsize\color{black!55} leased detour, gold only}}
  \,\bigr)\\[1.6em]
  p \;\triangleq\;
    \underbrace{\WEIGH{\{((0, \infty),\, \evt{adm})\}}{\SKIP}}_{\textnormal{\scriptsize\color{black!55} admission}}
    \SEQN\; \mathit{upg} \SEQN \mathit{seg}_1 \SEQN \mathit{seg}_2
    \SEQN \mathit{seg}_3 \SEQN \mathit{seg}_4
\end{array}
$}
\vspace{.5em}
\end{subfigure}%
\hfill
\begin{subfigure}[b]{0.42\textwidth}
\centering
\begin{tikzpicture}[x=0.30cm, y=0.034cm, line join=round,scale=.9]
  \fill[wsppbranch!8]
    (0,0) -- (0,10) -- (1,10) -- (1,20) -- (3,20) -- (3,30) -- (6,30)
    -- (6,40) -- (10,40) -- (10,80) -- (11.8,80) -- (11.8,0) -- cycle;
  \fill[wsppmissing!12]
    (1,0) -- (1,10) -- (2,10) -- (2,20) -- (4,20) -- (4,25)
    -- (11.8,25) -- (11.8,0) -- cycle;
  \draw[wsppbranch, semithick]
    (0,10) -- (1,10) -- (1,20) -- (3,20) -- (3,30) -- (6,30)
    -- (6,40) -- (10,40) -- (10,80) -- (11.8,80);
  \draw[wsppmissing, semithick]
    (1,10) -- (2,10) -- (2,20) -- (4,20) -- (4,25) -- (11.8,25);
  \draw[->] (0,0) -- (12.6,0);
  \node[font=\tiny, anchor=north east] at (18.6,0) {latency [ms]};
  \draw[->] (0,0) -- (0,92)
    node[above right=0pt, font=\tiny] {bandwidth [Gb/s]};
  \foreach \x in {2,4,6,8,10}
    \draw (\x,0) -- (\x,-3) node[below, font=\tiny] {$\x$};
  \foreach \y in {20,40,60,80}
    \draw (0,\y) -- (-0.35,\y) node[left, font=\tiny] {$\y$};
  \foreach \p in {(2,10),(3,10),(4,10),(5,10),(6,10),(7,10),(9,10),
                  (4,20),(5,20),(8,20),(7,30)}
    \fill[black!28] \p circle (1.1pt);
  \foreach \p in {(0,10),(1,20),(3,30),(6,40),(10,80)}
    \fill[wsppbranch] \p circle (1.6pt);
  \foreach \p in {(1,10),(2,20),(4,25)}
    \fill[wsppmissing] \p circle (1.6pt);
  \node[font=\tiny, anchor=north west, inner sep=1pt]
    at (0.05,8.6) {$\evt{adm}$};
  \node[font=\tiny, anchor=south east, inner sep=1pt]
    at (2.85,22.0) {$\evt{adm}\,\evt{d_1}$};
  \node[font=\tiny, anchor=south east, inner sep=1pt]
    at (5.85,31.8) {$\evt{adm}\,\evt{d_1 d_2}$};
  \node[font=\tiny, anchor=south west, inner sep=1pt]
    at (6.15,41.8) {$\evt{adm}\,\evt{d_1 d_2 d_3}$};
  \node[font=\tiny, anchor=south east, inner sep=1pt]
    at (9.8,81.8) {$\evt{adm}\,\evt{d_1 d_2 d_3 d_4}$};
  \node[font=\tiny, anchor=north west, inner sep=1pt]
    at (1.85,8.6) {$\evt{adm}\,\evt{u}$};
  \node[font=\tiny, anchor=north west, inner sep=1pt]
    at (2.15,18.6) {$\evt{adm}\,\evt{u}\,\evt{d_1}$};
  \node[font=\tiny, anchor=north west, inner sep=1pt]
    at (5.5,16.2) {$\evt{adm}\,\evt{u}\,\evt{d_1 d_2}$};
  \draw[ovevtcol!50, line width=0.3pt] (7.2,16.6) -- (4.32,23.6);
\end{tikzpicture}
\end{subfigure}
\endgroup
\caption{\footnotesize The two-class bypass chain. ~The policy $p$
($i = 1, \ldots, 4$) over the Pareto semiring of latency--bandwidth
frontiers (left); a literal weight $\{(\ell, w)\}$ charges $\ell$ ms and
caps the bandwidth at $w$\,Gb/s. The \textcolor{ovevtcol}{colored} event
letters log the billable events. The two frontiers of $p$ (right), shaded regions are the down-closed achievable trade-offs, the
staircases their finite frontier representations, and each frontier point
carries its \textcolor{ovevtcol}{witness}.}
\Description{Two panels: the two-class bypass-chain policy whose weights
carry colored event letters for admission, the gold purchase, and the
leased detours; and a plot with latency on the horizontal and bandwidth on
the vertical axis showing two staircases, a blue five-point gold frontier
rising to eighty gigabits and an orange three-point frontier for upgraded
best-effort traffic capped at twenty-five gigabits, a faint cloud of pruned
gold routes, shaded achievable regions, and witness words annotating every
frontier point.}
\label{fig:overview-pareto}
\end{figure}

%% file: topology_comparison.tex
\usetikzlibrary{shapes.geometric,shadows}
\definecolor{pod0}{HTML}{3E64A8}
\definecolor{pod1}{HTML}{DB7E2E}
\definecolor{pod2}{HTML}{3F9A5C}
\definecolor{pod3}{HTML}{C33A4B}
\definecolor{core}{HTML}{3A3A3A}
\definecolor{tor}{HTML}{1E7A7A}
\definecolor{hostfill}{HTML}{F4F4F6}
\definecolor{pathA}{HTML}{C9A227}
\definecolor{pathB}{HTML}{C0369D}
\begin{tikzpicture}[
  every node/.style={font=\sffamily},
  switch/.style={circle, draw=white, line width=0.7pt, fill=#1, minimum size=5.8pt, inner sep=0pt,
    drop shadow={opacity=0.30, shadow xshift=0.35pt, shadow yshift=-0.35pt}},
  hostnode/.style={rectangle, draw=black!55, fill=hostfill, minimum size=4pt, inner sep=0pt},
  hostendpoint/.style={rectangle, draw=black, fill=hostfill, thick, minimum size=6.5pt, inner sep=0pt,
    drop shadow={opacity=0.30, shadow xshift=0.35pt, shadow yshift=-0.35pt}},
  linkline/.style={black!22, line width=0.35pt},
  chordline/.style={tor!45!white, line width=0.3pt, opacity=0.4},
  pathhalo/.style={white, line width=3.2pt, line cap=round, opacity=0.9},
  pathline/.style={#1, line width=1.6pt, line cap=round},
  pathlabel/.style={fill=#1!10!white, draw=#1!65!black, line width=0.4pt, text opacity=1,
    rounded corners=1.5pt, inner sep=1.6pt,
    drop shadow={opacity=0.22, shadow xshift=0.3pt, shadow yshift=-0.3pt}},
]
\begin{scope}[xshift=0.000cm]
  \node[anchor=south] at (3.800, 4.900) {\bfseries\footnotesize Fat-tree ($k=4$): 20 switches, 16 hosts};
  \node[anchor=east, gray, font=\scriptsize] at (-0.32, 4.600) {core};
  \node[anchor=east, gray, font=\scriptsize] at (-0.32, 3.067) {agg};
  \node[anchor=east, gray, font=\scriptsize] at (-0.32, 1.533) {edge};
  \node[anchor=east, gray, font=\scriptsize] at (-0.32, 0.000) {hosts};
  \draw[linkline] (0.475, 1.533) -- (0.475, 3.067);
  \draw[linkline] (0.475, 1.533) -- (1.425, 3.067);
  \draw[linkline] (1.425, 1.533) -- (0.475, 3.067);
  \draw[linkline] (1.425, 1.533) -- (1.425, 3.067);
  \draw[linkline] (0.475, 3.067) -- (0.950, 4.600);
  \draw[linkline] (0.475, 3.067) -- (2.850, 4.600);
  \draw[linkline] (1.425, 3.067) -- (4.750, 4.600);
  \draw[linkline] (1.425, 3.067) -- (6.650, 4.600);
  \draw[linkline] (2.375, 1.533) -- (2.375, 3.067);
  \draw[linkline] (2.375, 1.533) -- (3.325, 3.067);
  \draw[linkline] (3.325, 1.533) -- (2.375, 3.067);
  \draw[linkline] (3.325, 1.533) -- (3.325, 3.067);
  \draw[linkline] (2.375, 3.067) -- (0.950, 4.600);
  \draw[linkline] (2.375, 3.067) -- (2.850, 4.600);
  \draw[linkline] (3.325, 3.067) -- (4.750, 4.600);
  \draw[linkline] (3.325, 3.067) -- (6.650, 4.600);
  \draw[linkline] (4.275, 1.533) -- (4.275, 3.067);
  \draw[linkline] (4.275, 1.533) -- (5.225, 3.067);
  \draw[linkline] (5.225, 1.533) -- (4.275, 3.067);
  \draw[linkline] (5.225, 1.533) -- (5.225, 3.067);
  \draw[linkline] (4.275, 3.067) -- (0.950, 4.600);
  \draw[linkline] (4.275, 3.067) -- (2.850, 4.600);
  \draw[linkline] (5.225, 3.067) -- (4.750, 4.600);
  \draw[linkline] (5.225, 3.067) -- (6.650, 4.600);
  \draw[linkline] (6.175, 1.533) -- (6.175, 3.067);
  \draw[linkline] (6.175, 1.533) -- (7.125, 3.067);
  \draw[linkline] (7.125, 1.533) -- (6.175, 3.067);
  \draw[linkline] (7.125, 1.533) -- (7.125, 3.067);
  \draw[linkline] (6.175, 3.067) -- (0.950, 4.600);
  \draw[linkline] (6.175, 3.067) -- (2.850, 4.600);
  \draw[linkline] (7.125, 3.067) -- (4.750, 4.600);
  \draw[linkline] (7.125, 3.067) -- (6.650, 4.600);
  \draw[linkline] (0.265, 0.000) -- (0.475, 1.533);
  \draw[linkline] (0.685, 0.000) -- (0.475, 1.533);
  \draw[linkline] (1.215, 0.000) -- (1.425, 1.533);
  \draw[linkline] (1.635, 0.000) -- (1.425, 1.533);
  \draw[linkline] (2.165, 0.000) -- (2.375, 1.533);
  \draw[linkline] (2.585, 0.000) -- (2.375, 1.533);
  \draw[linkline] (3.115, 0.000) -- (3.325, 1.533);
  \draw[linkline] (3.535, 0.000) -- (3.325, 1.533);
  \draw[linkline] (4.065, 0.000) -- (4.275, 1.533);
  \draw[linkline] (4.485, 0.000) -- (4.275, 1.533);
  \draw[linkline] (5.015, 0.000) -- (5.225, 1.533);
  \draw[linkline] (5.435, 0.000) -- (5.225, 1.533);
  \draw[linkline] (5.965, 0.000) -- (6.175, 1.533);
  \draw[linkline] (6.385, 0.000) -- (6.175, 1.533);
  \draw[linkline] (6.915, 0.000) -- (7.125, 1.533);
  \draw[linkline] (7.335, 0.000) -- (7.125, 1.533);
  \draw[pathhalo] (2.165, 0.000) -- (2.375, 1.533);
  \draw[pathhalo] (2.375, 1.533) -- (2.375, 3.067);
  \draw[pathhalo] (2.375, 3.067) -- (0.950, 4.600);
  \draw[pathhalo] (0.950, 4.600) -- (6.175, 3.067);
  \draw[pathhalo] (6.175, 3.067) -- (6.175, 1.533);
  \draw[pathhalo] (6.175, 1.533) -- (5.965, 0.000);
  \draw[pathhalo] (2.165, 0.000) -- (2.375, 1.533);
  \draw[pathhalo] (2.375, 1.533) -- (3.325, 3.067);
  \draw[pathhalo] (3.325, 3.067) -- (6.650, 4.600);
  \draw[pathhalo] (6.650, 4.600) -- (7.125, 3.067);
  \draw[pathhalo] (7.125, 3.067) -- (6.175, 1.533);
  \draw[pathhalo] (6.175, 1.533) -- (5.965, 0.000);
  \draw[pathline=pathA] (2.165, 0.000) -- (2.375, 1.533);
  \draw[pathline=pathA] (2.375, 1.533) -- (2.375, 3.067);
  \draw[pathline=pathA] (2.375, 3.067) -- (0.950, 4.600);
  \draw[pathline=pathA] (0.950, 4.600) -- (6.175, 3.067);
  \draw[pathline=pathA] (6.175, 3.067) -- (6.175, 1.533);
  \draw[pathline=pathA] (6.175, 1.533) -- (5.965, 0.000);
  \draw[pathline=pathB] (2.165, 0.000) -- (2.375, 1.533);
  \draw[pathline=pathB] (2.375, 1.533) -- (3.325, 3.067);
  \draw[pathline=pathB] (3.325, 3.067) -- (6.650, 4.600);
  \draw[pathline=pathB] (6.650, 4.600) -- (7.125, 3.067);
  \draw[pathline=pathB] (7.125, 3.067) -- (6.175, 1.533);
  \draw[pathline=pathB] (6.175, 1.533) -- (5.965, 0.000);
  \node[switch=pod0] at (0.475, 1.533) {};
  \node[switch=pod0] at (1.425, 1.533) {};
  \node[switch=pod0] at (0.475, 3.067) {};
  \node[switch=pod0] at (1.425, 3.067) {};
  \node[switch=pod1] at (2.375, 1.533) {};
  \node[switch=pod1] at (3.325, 1.533) {};
  \node[switch=pod1] at (2.375, 3.067) {};
  \node[switch=pod1] at (3.325, 3.067) {};
  \node[switch=pod2] at (4.275, 1.533) {};
  \node[switch=pod2] at (5.225, 1.533) {};
  \node[switch=pod2] at (4.275, 3.067) {};
  \node[switch=pod2] at (5.225, 3.067) {};
  \node[switch=pod3] at (6.175, 1.533) {};
  \node[switch=pod3] at (7.125, 1.533) {};
  \node[switch=pod3] at (6.175, 3.067) {};
  \node[switch=pod3] at (7.125, 3.067) {};
  \node[switch=core] at (0.950, 4.600) {};
  \node[switch=core] at (2.850, 4.600) {};
  \node[switch=core] at (4.750, 4.600) {};
  \node[switch=core] at (6.650, 4.600) {};
  \node[hostnode] at (0.265, 0.000) {};
  \node[hostnode] at (0.685, 0.000) {};
  \node[hostnode] at (1.215, 0.000) {};
  \node[hostnode] at (1.635, 0.000) {};
  \node[hostnode] at (2.165, 0.000) {};
  \node[hostnode] at (2.585, 0.000) {};
  \node[hostnode] at (3.115, 0.000) {};
  \node[hostnode] at (3.535, 0.000) {};
  \node[hostnode] at (4.065, 0.000) {};
  \node[hostnode] at (4.485, 0.000) {};
  \node[hostnode] at (5.015, 0.000) {};
  \node[hostnode] at (5.435, 0.000) {};
  \node[hostnode] at (5.965, 0.000) {};
  \node[hostnode] at (6.385, 0.000) {};
  \node[hostnode] at (6.915, 0.000) {};
  \node[hostnode] at (7.335, 0.000) {};
  \node[hostendpoint] at (2.165, 0.000) {};
  \node[hostendpoint] at (5.965, 0.000) {};
  \node[pathlabel=pathA] at (2.713, 4.656) {\scriptsize \textcolor{pathA}{\textbf{A}} 947\,Mbps $\cdot$ 4.58\,ms};
  \node[pathlabel=pathB] at (5.947, 3.670) {\scriptsize \textcolor{pathB}{\textbf{B}} 905\,Mbps $\cdot$ 4.48\,ms};
\end{scope}
\begin{scope}[xshift=13.776cm]
  \node[anchor=south] at (0.000, 4.900) {\bfseries\footnotesize Jellyfish (random regular, $k=4$): 20 switches, 16 hosts, uplinks=3};
  \draw[chordline] (0.000, 4.042) -- (3.893, 1.762);
  \draw[chordline] (0.000, 4.042) -- (1.265, 3.957);
  \draw[chordline] (0.000, 4.042) -- (3.312, 1.276);
  \draw[chordline] (1.265, 3.957) -- (-3.893, 2.838);
  \draw[chordline] (1.265, 3.957) -- (-2.406, 3.709);
  \draw[chordline] (2.406, 3.709) -- (-1.265, 3.957);
  \draw[chordline] (2.406, 3.709) -- (2.406, 0.891);
  \draw[chordline] (2.406, 3.709) -- (3.893, 1.762);
  \draw[chordline] (3.312, 3.324) -- (-1.265, 3.957);
  \draw[chordline] (3.312, 3.324) -- (-3.312, 3.324);
  \draw[chordline] (3.312, 3.324) -- (-3.312, 1.276);
  \draw[chordline] (3.893, 2.838) -- (0.000, 0.558);
  \draw[chordline] (3.893, 2.838) -- (-1.265, 0.643);
  \draw[chordline] (3.893, 2.838) -- (-2.406, 0.891);
  \draw[chordline] (4.094, 2.300) -- (-4.094, 2.300);
  \draw[chordline] (4.094, 2.300) -- (0.000, 0.558);
  \draw[chordline] (4.094, 2.300) -- (-1.265, 3.957);
  \draw[chordline] (3.893, 1.762) -- (-3.312, 3.324);
  \draw[chordline] (3.312, 1.276) -- (0.000, 0.558);
  \draw[chordline] (3.312, 1.276) -- (-3.312, 1.276);
  \draw[chordline] (2.406, 0.891) -- (-3.312, 3.324);
  \draw[chordline] (2.406, 0.891) -- (-3.893, 2.838);
  \draw[chordline] (1.265, 0.643) -- (-2.406, 3.709);
  \draw[chordline] (1.265, 0.643) -- (-4.094, 2.300);
  \draw[chordline] (1.265, 0.643) -- (-2.406, 0.891);
  \draw[chordline] (-1.265, 0.643) -- (-2.406, 0.891);
  \draw[chordline] (-1.265, 0.643) -- (-3.312, 1.276);
  \draw[chordline] (-3.893, 1.762) -- (-3.893, 2.838);
  \draw[chordline] (-3.893, 1.762) -- (-2.406, 3.709);
  \draw[chordline] (-3.893, 1.762) -- (-4.094, 2.300);
  \draw[linkline] (0.000, 4.042) -- (0.000, 4.460);
  \draw[linkline] (1.265, 3.957) -- (1.569, 4.354);
  \draw[linkline] (2.406, 3.709) -- (2.984, 4.047);
  \draw[linkline] (3.312, 3.324) -- (4.107, 3.570);
  \draw[linkline] (3.893, 2.838) -- (4.828, 2.967);
  \draw[linkline] (4.094, 2.300) -- (5.076, 2.300);
  \draw[linkline] (3.893, 1.762) -- (4.828, 1.633);
  \draw[linkline] (3.312, 1.276) -- (4.107, 1.030);
  \draw[linkline] (2.406, 0.891) -- (2.984, 0.553);
  \draw[linkline] (1.265, 0.643) -- (1.569, 0.246);
  \draw[linkline] (0.000, 0.558) -- (0.000, 0.140);
  \draw[linkline] (-1.265, 0.643) -- (-1.569, 0.246);
  \draw[linkline] (-2.406, 0.891) -- (-2.984, 0.553);
  \draw[linkline] (-3.312, 1.276) -- (-4.107, 1.030);
  \draw[linkline] (-3.893, 1.762) -- (-4.828, 1.633);
  \draw[linkline] (-4.094, 2.300) -- (-5.076, 2.300);
  \draw[pathhalo] (4.107, 3.570) -- (3.312, 3.324);
  \draw[pathhalo] (3.312, 3.324) -- (-3.312, 3.324);
  \draw[pathhalo] (-3.312, 3.324) -- (3.893, 1.762);
  \draw[pathhalo] (3.893, 1.762) -- (2.406, 3.709);
  \draw[pathhalo] (2.406, 3.709) -- (-1.265, 3.957);
  \draw[pathhalo] (-1.265, 3.957) -- (4.094, 2.300);
  \draw[pathhalo] (4.094, 2.300) -- (5.076, 2.300);
  \draw[pathhalo] (4.107, 3.570) -- (3.312, 3.324);
  \draw[pathhalo] (3.312, 3.324) -- (-1.265, 3.957);
  \draw[pathhalo] (-1.265, 3.957) -- (4.094, 2.300);
  \draw[pathhalo] (4.094, 2.300) -- (5.076, 2.300);
  \draw[pathline=pathA] (4.107, 3.570) -- (3.312, 3.324);
  \draw[pathline=pathA] (3.312, 3.324) -- (-3.312, 3.324);
  \draw[pathline=pathA] (-3.312, 3.324) -- (3.893, 1.762);
  \draw[pathline=pathA] (3.893, 1.762) -- (2.406, 3.709);
  \draw[pathline=pathA] (2.406, 3.709) -- (-1.265, 3.957);
  \draw[pathline=pathA] (-1.265, 3.957) -- (4.094, 2.300);
  \draw[pathline=pathA] (4.094, 2.300) -- (5.076, 2.300);
  \draw[pathline=pathB] (4.107, 3.570) -- (3.312, 3.324);
  \draw[pathline=pathB] (3.312, 3.324) -- (-1.265, 3.957);
  \draw[pathline=pathB] (-1.265, 3.957) -- (4.094, 2.300);
  \draw[pathline=pathB] (4.094, 2.300) -- (5.076, 2.300);
  \node[switch=tor] at (0.000, 4.042) {};
  \node[switch=tor] at (1.265, 3.957) {};
  \node[switch=tor] at (2.406, 3.709) {};
  \node[switch=tor] at (3.312, 3.324) {};
  \node[switch=tor] at (3.893, 2.838) {};
  \node[switch=tor] at (4.094, 2.300) {};
  \node[switch=tor] at (3.893, 1.762) {};
  \node[switch=tor] at (3.312, 1.276) {};
  \node[switch=tor] at (2.406, 0.891) {};
  \node[switch=tor] at (1.265, 0.643) {};
  \node[switch=tor] at (0.000, 0.558) {};
  \node[switch=tor] at (-1.265, 0.643) {};
  \node[switch=tor] at (-2.406, 0.891) {};
  \node[switch=tor] at (-3.312, 1.276) {};
  \node[switch=tor] at (-3.893, 1.762) {};
  \node[switch=tor] at (-4.094, 2.300) {};
  \node[switch=tor] at (-3.893, 2.838) {};
  \node[switch=tor] at (-3.312, 3.324) {};
  \node[switch=tor] at (-2.406, 3.709) {};
  \node[switch=tor] at (-1.265, 3.957) {};
  \node[hostnode] at (0.000, 4.460) {};
  \node[hostnode] at (1.569, 4.354) {};
  \node[hostnode] at (2.984, 4.047) {};
  \node[hostnode] at (4.107, 3.570) {};
  \node[hostnode] at (4.828, 2.967) {};
  \node[hostnode] at (5.076, 2.300) {};
  \node[hostnode] at (4.828, 1.633) {};
  \node[hostnode] at (4.107, 1.030) {};
  \node[hostnode] at (2.984, 0.553) {};
  \node[hostnode] at (1.569, 0.246) {};
  \node[hostnode] at (0.000, 0.140) {};
  \node[hostnode] at (-1.569, 0.246) {};
  \node[hostnode] at (-2.984, 0.553) {};
  \node[hostnode] at (-4.107, 1.030) {};
  \node[hostnode] at (-4.828, 1.633) {};
  \node[hostnode] at (-5.076, 2.300) {};
  \node[hostendpoint] at (4.107, 3.570) {};
  \node[hostendpoint] at (5.076, 2.300) {};
  \node[pathlabel=pathA] at (2.747, 1.447) {\scriptsize \textcolor{pathA}{\textbf{A}} 942\,Mbps $\cdot$ 5.54\,ms};
  \node[pathlabel=pathB] at (-0.558, 4.314) {\scriptsize \textcolor{pathB}{\textbf{B}} 902\,Mbps $\cdot$ 2.56\,ms};
\end{scope}
\end{tikzpicture}

%% file: preliminaries.tex


This section recaps the fragment of \wnetkat used in this paper. Throughout, we fix
a finite set $\fields$ of \emph{fields} and a finite set $\values$ of
\emph{values}. A \emph{packet} $\pkta \in \packets \triangleq \fields \to
\values$ assigns a value to every field; we write $\pkta.\fielda$ for
$\pkta(\fielda)$ and $\updatepkt{\pkta}{\fielda}{v}$ for the packet that
agrees with $\pkta$ except that field $\fielda$ maps to $v$.

\subsection{Syntax and Semantics}\label{sec:prelim-wnetkat}

The syntax and semantics of ($\dup$-free) \wnetkat are given in
\Cref{fig:wnetkat-semantics}. The syntax is parametric on a \emph{semiring} $\semi = (\semidom, +, \cdot\,, \semizero, \semione)$, 
  consisting of a carrier $\semidom$ of {\em weights}, where
  $(\semidom, +, \semizero)$ is a commutative monoid, $(\semidom, \cdot\,, \semione)$ is a monoid, and
   $\cdot$ distributes over $+$ and $\semizero$ is annihilating
  ($\semizero \cdot \wta = \wta \cdot \semizero = \semizero$).  
  \emph{Predicates} $t$ are Boolean combinations
of tests $\EQ{\fielda}{v}$ and receive a two-level semantics. First, the
evident Boolean interpretation $\sem{t}_{\tests}(\pkta) \in \{0,1\}$ on
tests, which is then embedded into weights via Iverson brackets. 
Second, \emph{policies} $p$ denote
\emph{weighted packet relations}
\(
  \sem{p} \colon \packets \to \packets \to \semidom,
\)
where $\sem{p}(\pkta)(\pktb)$ is the total weight with which $p$ transforms
the input packet $\pkta$ into the output packet $\pktb$. 

A predicate, used as a policy, either passes the input packet on unchanged
(with weight $\semione$) or drops it (weight $\semizero$). The modification
$\ASSN{\fielda}{v}$ deterministically updates field $\fielda$.
 The quantitative combinators mirror the semiring structure: the \emph{choice}
$\ADD{p}{q}$ adds weights, the \emph{weighting} $\WEIGH{\wta}{p}$ scales them,
and \emph{sequential composition} $\SEQ{p}{q}$ sums, over all intermediate
packets $\pktc$, the product of the weight of $p$ getting from $\pkta$ to
$\pktc$ and the weight of $q$ getting from $\pktc$ to $\pktb$. 

$\ITER{p}$ is an infinite choice
$\ADD{\SKIP}{\ADD{p}{\ADD{(\SEQ{p}{p})}{\cdots}}}$ between running $p$ any
finite number of times. Its semantics is the countable sum
$\sum_{n \in \N} \sem{\NFOLD{p}}(\pkta)(\pktb)$, which exists if we assume $\semi$ to be
$\omega$-continuous.

\begin{definition}
  \label{def:omega-continuous-semiring}
  A semiring $\semi$ is 
  \emph{$\omega$-continuous} w.r.t.\ a partial order
  ${\semiord} \subseteq \semidom^2$, if
 $(\semidom, \semiord)$ is \emph{$\omega$-complete}:
      every ascending chain $\wta_0 \semiord \wta_1 \semiord \ldots$ has a
      least upper bound $\semisup_{i} \wta_i$; $\semiord$ is \emph{positive}: $\semizero$ is the least element; and
 $+$ and $\cdot$ are monotone and preserve suprema of ascending chains in each
      argument.
\end{definition}

$\omega$-continuity provides \emph{countable sums}:
for any $\{\wta_i\}_{i \in \N}$, the partial sums form an ascending chain (by
positivity and monotonicity), so
\(
  \sum_{i \in \N} \wta_i
  \triangleq
    \semisup_{k \in \N}\; \sum_{i < k} \wta_i
\)
is well-defined, and its value is invariant under permutations of the
summands~\cite{domain_theory}. 
 \Cref{fig:embedding-portfolio} in \Cref{sec:embedding} collects the
$\omega$-continuous semirings of common network quantities. A running
example throughout this paper is the \emph{probability semiring}
$\Rext = ([0,\infty], +, \cdot\,, 0, 1)$ of extended non-negative reals
with $\wta^*$ as the geometric series $\sum_i \wta^i$ and the usual order.
This is the semiring underlying probabilistic analyses.

\begin{figure}[t]
\centering
\begin{minipage}[t]{.40\textwidth}
\textbf{Syntax}
\[
\begin{array}{r@{\ \ }r@{~}c@{~}l@{}}
  \textrm{Packets} & \multicolumn{3}{l}{\packets \ni \pkta,\pktb,\pktc}\\[0.2em]
  \textrm{Weights} & \multicolumn{3}{l}{\semidom \ni \wta, \wtb}\\[0.2em]
  \textrm{Predicates} & \tests \ni t, u &
  ::= & \FALSE \mid \TRUE \\
  & & \mid & \EQ{\fielda}{v} \mid \NOT{t} \\
  & & \mid & \OR{t}{t'} \mid \AND{t}{t'} \\[0.2em]
  \textrm{Policies} & \pols \ni p,q &
  ::= & t                       \\
  & & \mid & \ASSN{\fielda}{v}  \\
  & & \mid & \SEQ{p}{q}         \\
  & & \mid & \WEIGH{\wta}{p}    \\
  & & \mid & \ADD{p}{q}         \\
  & & \mid & \ITER{p}           \\
\end{array}
\]
\end{minipage}\hfill\vrule\hfill
\begin{minipage}[t]{.56\textwidth}
\textbf{Semantics}
\[
\def\arraystretch{1.15}
\begin{array}{r@{~~}c@{~~}l}
  \sem{p} &\colon& \packets \to \packets \to \semidom \\
  \sem{t}(\pkta)(\pktb) & = & [\pkta = \pktb \,\wedge\, \sem{t}_{\tests}(\pkta) = 1] \\
  \sem{\ASSN{\fielda}{v}}(\pkta)(\pktb) & = & [\pktb = \updatepkt{\pkta}{\fielda}{v}] \\
  \sem{\SEQ{p}{q}}(\pkta)(\pktb) & = & \displaystyle\sum_{\pktc \in \packets} \sem{p}(\pkta)(\pktc) \cdot \sem{q}(\pktc)(\pktb) \\
  \sem{\WEIGH{\wta}{p}}(\pkta)(\pktb) & = & \wta \cdot \sem{p}(\pkta)(\pktb) \\
  \sem{\ADD{p}{q}}(\pkta)(\pktb) & = & \sem{p}(\pkta)(\pktb) + \sem{q}(\pkta)(\pktb) \\
  \sem{\ITER{p}}(\pkta)(\pktb) & = & \displaystyle\sum_{n \in \N} \sem{\NFOLD{p}}(\pkta)(\pktb) \\
  \multicolumn{3}{l}{\quad\text{where } \NFOLD[0]{p} = \SKIP
    \text{ and } \NFOLD[n+1]{p} = \SEQ{p}{\NFOLD{p}}}
\end{array}
\]
\end{minipage}
\caption{Syntax and semantics of ($\dup$-free) \wnetkat over an
$\omega$-continuous semiring $\semi = (\semidom, +, \cdot\,, \semizero,
\semione)$. Operator precedence, from strongest- to weakest-binding:
$(-)^*$, $\neg$, $\wedge$, $\vee$, $\odot$, $;$, $\oplus$. Here,
$[P]$ denotes the Iverson bracket into $\semi$: $[P] = \semione$ if the
proposition $P$ holds and $[P] = \semizero$ otherwise.}
\Description{Two-panel figure: the grammar of dup-free \wnetkat on the left and
its denotational semantics as weighted packet relations on the right.}
\label{fig:wnetkat-semantics}
\end{figure}


\subsection{Equational Reasoning}\label{sec:prelim-equational}
We write $p \equiv q$ for \emph{semantic equivalence} $\sem{p} = \sem{q}$.
On $\dup$-free policies, $\equiv$ agrees with the history-based semantics of
full \wnetkat~\cite{wNetKAT}, so all results below apply unchanged in that
setting.

The data structure of \Cref{sec:syntax-semantics} is manipulated by
\emph{equational} rewriting of policies, so we collect the laws this paper
uses. The equational theory of \wnetkat is, however,
\emph{not} that of \netkat.

\begin{example}
  \label{example:non-idempotent}
  Choice is not idempotent: distributivity in $\semi$ only yields
  $\ADD{p}{p} \equiv \WEIGH{(\semione + \semione)}{p}$, which differs from
  $p$ in general (e.g.\ over $\Rext$). Hence, the idempotence axiom of
  \kat~\cite{kat} and the set-based reasoning behind \netkat's
  symbolic verifiers~\cite{KATch,smolka_fast_2015} is unsound for \wnetkat.
\end{example}

What survives is the theory of $\omega$-continuous (indeed,
Conway~\cite{iteration_semirings, inductive_semirings}) semirings. Since
the semantic domain $\Pk \rightarrow \Pk \rightarrow S$ of \wnetkat forms such a semiring (\cite{matrix_continuity_result,wNetKAT})
with appropriate semiring operations,
we obtain the two star laws that drive the algorithm of
\Cref{sec:wspp-star}:

\begin{theorem}[star laws]
  \label{theorem:star-eqns}
  For all $p, q \in \pols$:
  \begin{align*}
    \ITER{p}  \equiv \ADD{\SKIP}{\SEQ{p}{\ITER{p}}}
               \equiv \ADD{\SKIP}{\SEQ{\ITER{p}}{p}}
           \;  \textsc{\em (fixpoint)}\; 
    \ITER{(\ADD{p}{q})}  \equiv \SEQ{\ITER{p}}{\ITER{(\SEQ{q}{\ITER{p}})}}
               \equiv\SEQ{\ITER{(\SEQ{\ITER{p}}{q})}}{\ITER{p}}
             \; \textsc{\em (denesting)}
  \end{align*}
\end{theorem}

Beyond the star laws we use two families of easily verified facts, collected
in \Cref{fig:preliminaries-dependency-map}: \emph{linearity} laws (E2), which
hold in any semiring, and \emph{packet algebra} laws (E3) about the
interaction of tests and assignments, familiar from
\netkat~\cite{netkat}. We refer to them by their identifiers (E2a), (E3b),
etc.

\input{preliminaries-dependency-map}

%% file: preliminaries-dependency-map.tex

\begin{figure}[!tbp]
\centering
\fontsize{8}{9}\selectfont
\setlength{\tabcolsep}{1pt}
\renewcommand{\arraystretch}{1.10}
\setlength{\fboxsep}{2pt}
\newlength{\depmapwidth}
\setlength{\depmapwidth}{\linewidth}
\definecolor{equationfacts}{HTML}{FFF3DA}
\definecolor{packetfacts}{HTML}{EAF3FF}
\newcommand{\factgroup}[2]{%
    \colorbox{#1}{\parbox{\dimexpr\depmapwidth-2\fboxsep\relax}{\textbf{#2}}}
    \\[-0.15em]}
\newcommand{\factid}[1]{\textnormal{\tiny\textcolor{black!65}{(#1)}}\!}
\begin{tabular}{@{}c}
\toprule
\factgroup{equationfacts}{E2. Linearity (valid in any semiring)}
\begin{minipage}[t]{\linewidth}
\vspace{0pt}
\[
\begin{gathered}
\factid{E2a}\;\SEQ{p}{(\ADD{q_1}{q_2})}\equiv \ADD{\SEQ{p}{q_1}}{\SEQ{p}{q_2}},\qquad
\SEQ{(\ADD{q_1}{q_2})}{p}\equiv \ADD{\SEQ{q_1}{p}}{\SEQ{q_2}{p}},\\
\factid{E2b}\;\WEIGH{\wta}{(\ADD{q_1}{q_2})}\equiv \ADD{(\WEIGH{\wta}{q_1})}{(\WEIGH{\wta}{q_2})},\qquad
\WEIGH{(\wta + \wtb)}{q}\equiv \ADD{(\WEIGH{\wta}{q})}{(\WEIGH{\wtb}{q})},\\
\factid{E2c}\;\SEQ{(\WEIGH{\wta}{q_1})}{q_2}\equiv \WEIGH{\wta}{(\SEQ{q_1}{q_2})},\qquad
\WEIGH{\wta}{(\WEIGH{\wtb}{q})}\equiv \WEIGH{(\wta \cdot \wtb)}{q},\\
\factid{E2d}\;\SEQ{\SKIP}{p}\equiv \SEQ{p}{\SKIP}\equiv \WEIGH{\semione}{p}\equiv p,\qquad
\ADD{p}{\DROP} \equiv p,\qquad
\factid{E2e}\;\SEQ{\DROP}{p}\equiv \SEQ{p}{\DROP}\equiv\DROP\equiv \WEIGH{\semizero}{p}
\end{gathered}
\]
\end{minipage}\\[0.65em]
\midrule
\factgroup{packetfacts}{E3. Packet algebra (single-field laws, cf.~\cite{netkat})}
\begin{minipage}[t]{\linewidth}
\vspace{0pt}
\[
\begin{gathered}
\factid{E3a}\;\SEQ{\ASSN{\fielda}{v_1}}{\ASSN{\fielda}{v_2}}\equiv \ASSN{\fielda}{v_2},\qquad
\SEQ{\ASSN{\fielda}{v}}{\EQ{\fielda}{v}}\equiv \ASSN{\fielda}{v},\qquad
\factid{E3b}\;\SEQ{\ASSN{\fielda}{v_1}}{\EQ{\fielda}{v_2}}\equiv \DROP
  \;\text{ if } v_1\ne v_2,\\
\factid{E3c}\;\SEQ{\EQ{\fielda}{v_1}}{\EQ{\fielda}{v_2}}\equiv\DROP
  \text{ if } v_1 \neq v_2,\quad
  \SEQ{\NOTEQ{\fielda}{v}}{\EQ{\fielda}{v}}\equiv\DROP,\quad
  \SEQ{\NOTEQ{\fielda}{v_2}}{\EQ{\fielda}{v_1}}\equiv \EQ{\fielda}{v_1}
  \;\text{ if } v_1\ne v_2,\\
\factid{E3d}\;\SEQ{e}{p}\equiv \SEQ{p}{e}
\quad\text{if } e\in\{\EQ{\fielda}{v},\,\NOTEQ{\fielda}{v},\,\ASSN{\fielda}{v}\}
\text{ and $p$ does not mention $\fielda$}
\end{gathered}
\]
\end{minipage}\\
\bottomrule
\end{tabular}
\caption{\footnotesize Linearity laws inherited from
the semiring structure and packet algebra laws governing tests
and assignments.}
\label{fig:preliminaries-dependency-map}
\end{figure}

%% file: syntax-semantics.tex
\newcommand{\Mid}{\ensuremath{\mathsf{mid}}}

This section develops the central data structure of this paper:
\emph{weighted symbolic packet programs} (wSPPs), a symbolic representation
of the weighted packet relations $\sem{p}$ denoted by \wnetkat policies.
In \Cref{sec:embedding-defs}, we carefully define the computable structures our algorithm operates on, followed by an extensive presentation of our compiler in \Cref{sec:wspp}.

\subsection{Star Semirings and Embeddings}\label{sec:embedding-defs}

The denotational semantics needs $\semi$ to be $\omega$-continuous so that
iteration can be defined with infinite sums $\sem{\ITER{p}} = \sum_{n\in\N}\sem{p}^n$
(\Cref{sec:preliminaries}). Any \emph{executable} data structure, however,
requires computing with weights: representing them finitely, adding,
multiplying, and comparing them. These two demands conflict, as illustrated
by the following example.

\begin{example}
  \label{example:embedding-tension}
  The probability semiring $\Rext$ is $\omega$-continuous, but
  uncountable and thus contains weights with no finite representation. The relevant
  computable part, the extended non-negative rationals $\Qext$, is a semiring
  of finite objects with decidable equality and order, but it is not
  $\omega$-continuous: the chain $3,\, 3.1,\, \ldots$ of
  approximations of $\pi$ has no least upper bound in $\Qext$.
  \hfill $\triangle$
\end{example}

Looking ahead, the operations of
\Cref{sec:wspp-easy-ops,sec:wspp-star} evaluate finitely many sums
and products of leaf weights,
test weights for equality with $\semizero$, and, crucially, take only finitely many \emph{semiring-level star operations}
$\wtb^* = \semione + \wtb + \wtb\cdot\wtb + \wtb\cdot\wtb\cdot\wtb + \dots$. This
motivates the following interface:

\begin{definition}
  \label{def:star-semiring}
  A \emph{star semiring}
  $\csr = (\csrset, +, \cdot\,, \semizero, \semione, \cstar{(-)}, \semiord)$
  is a semiring with a unary operation
  $\cstar{(-)} \colon \csrset \to \csrset$ and a partial order
  ${\semiord} \subseteq \csrset^2$, both arbitrary. We call $\csr$
  \emph{computable} if its elements admit finite representations under which
  $=$ and $\semiord$ are decidable and $+$, $\cdot$, and $\cstar{(-)}$ are
  computable.
\end{definition}

Computable star semirings are then appropriately embedded into $\omega$-continuous ones:


\begin{definition}[embedding]
  \label{def:embedding}
  Let $\csr$ be a star semiring and $\semi$ an $\omega$-continuous semiring.
  An \emph{embedding} $\emb \colon \csr \hookrightarrow \semi$ is a map
  $\emb \colon \csrset \to \semidom$ such that  $\emb(\semizero)=\semizero$, $\emb(\semione)=\semione$, and, for all $a, b \in \csrset$:
  \begin{align*}
\textnormal{(i) }    \emb(a+b)=\emb(a)+\emb(b),\quad
    \emb(a\cdot b)&=\emb(a)\cdot\emb(b); \quad \textnormal{(ii) }
    a \semiord b ~\text{iff}~ \emb(a) \semiord \emb(b); \quad \textnormal{(iii) }
    \emb(\cstar{a}) = \emb(a)^{*}. 
  \end{align*}
\end{definition}

Conditions (i) and (ii) make $\emb$ an \emph{injective} semiring homomorphism.
Crucially, in (iii) the left-hand side is
an application of the $\cstar{(-)}$ function we are \emph{free to program} and the right-hand side a
countable semiring sum. For the probability semiring $\cstar{(-)}$ is the
geometric series: $\Qext \hookrightarrow \Rext$ with
\[
  \cstar{a} \;\triangleq\;
  \begin{cases}
    (1-a)^{-1} & \text{if } a < 1,\\
    \infty & \text{otherwise}
  \end{cases}
  \;=\;
  \sum_{i=0}^{\infty} a^i~.
\]
\Cref{sec:appendix-portfolio}
collects such embeddings for the standard semirings of network quantities,
and \Cref{sec:pareto} constructs them for pareto semirings and trace-carrying
pareto semirings.

\smallskip\noindent\textbf{Convention.}
For the remainder of the paper we fix an $\omega$-continuous semiring
$\semi$ and a computable star semiring $\csr$ with an embedding
$\emb \colon \csr \hookrightarrow \semi$. Since $\emb$ is injective
and preserves all structure, we \emph{identify} $\csr$ with its image such that
\[
  \csrset \subseteq \semidom
  \quad\text{is a subsemiring with decidable } =, \semiord
  \text{ and computable } +,\, \cdot\,,\,
  \cstar{(-)} = (-)^{*}|_{\csrset}.
\]
By (iii), taking the semantic star never leaves $\csrset$, and $\cstar{(-)}$
is its computable implementation. Two consequences deserve emphasis. First,
although \Cref{def:star-semiring} postulates no laws for $\cstar{(-)}$, an
embedded star semiring inherits every law of the semantic star. For
instance, for all $a, b \in \csrset$,
\[
  \semione + a \cdot \cstar{a} \;=\; \cstar{a}
  \qquad\text{and}\qquad
  \cstar{(a+b)} \;=\; \cstar{a} \cdot \cstar{(b \cdot \cstar{a})},
\]
the analogues of \textsc{fixpoint} and \textsc{denesting}
(\Cref{theorem:star-eqns}). Second, all weights appearing in policies and
data structures below are drawn from $\csrset$. $\pols_{\csr}$ denotes the policies with weights from $\csrset$.
Since we have embedded $\csrset$ in $\semi$, we use $w$ to
denote elements from both sets in the following sections.

\subsection{Weighted Symbolic Packet Programs}\label{sec:wspp}

Let us start with the central definition:
\begin{definition}
  \label{def:wspp-syntax}
  Write $X \rightharpoonup Y$ for the set of \emph{finite} partial maps from
  $X$ to $Y$, and fix an arbitrary total order on $\fields$. The set
  $\WSPP_{\csr}$ of \emph{weighted symbolic packet programs} over $\csr$ is
  defined inductively:
  \begin{enumerate}
    \item every weight $\wta \in \csrset$ yields a \emph{leaf}
      $\rulebox{\wta} \in \WSPP_{\csr}$;
    \item $\wspp(\fielda, b, m, d) \in
      \WSPP_{\csr}$ is a \emph{node}, where $\fielda$ is a field,
      $b\colon \values \rightharpoonup (\values
      \rightharpoonup \WSPP_{\csr})$ are the \emph{test-assignments},
      $m \colon \values \rightharpoonup
      \WSPP_{\csr}$ are the \emph{default assignments}, and
      $d \in \WSPP_{\csr}$ is the \emph{default identity}, provided every field occurring in $b$, 
       $m$, or
       $d$ is
       greater \mbox{than $\fielda$ \emph{(well-formedness)}.}
  \end{enumerate}
  We abbreviate $\rulebox{\semione}$ and $\rulebox{\semizero}$ by
  $\hat{\mathbb{1}}$ and $\hat{\mathbb{0}}$, respectively, and tacitly
  assume well-formedness from here on.
\end{definition}

Operationally, a wSPP transforms an input packet $\pkta$
into weighted output packets. A leaf $\rulebox{\wta}$ outputs $\pkta$ unchanged
with weight $\wta$. A node $\rho = \wspp(\fielda,b,m,d)$ inspects
$u = \pkta.\fielda$ and branches:
\begin{enumerate}
  \item if $u \in \dom(b)$, then $\rho$ rewrites $\fielda$ to some
    $v \in \dom(b(u))$ and continues with $b(u)(v)$,
    %
  \item otherwise $\rho$ rewrites $\fielda$ to some $z \in \dom(m)$ and
    continues with $m(z)$, 
  \item and if moreover $u \notin \dom(m)$, then $\rho$ may also leave
    $\fielda$ untouched and continue with $d$. 
\end{enumerate}
Well-formedness guarantees that along every root-to-leaf path fields
strictly increase, so each field is read and written at most once, an invariant that all algorithms below preserve
and rely on.

We package the behavior of all three cases above into a single function referred to as the \emph{branch map} of a node
$\rho = \wspp(\fielda, b, m, d)$, defined as
\[
  \rho(-) \;\colon\; \values \to (\values \rightharpoonup \WSPP_{\csr}),
  \quad
  \rho(u) \;\triangleq\;
  \begin{cases}
    b(u) & \text{if } u \in \dom(b),\\
    m \cup \{u \mapsto d\} & \text{if } u \notin \dom(b)\cup\dom(m),\,
      d \neq \hat{\mathbb{0}},\\
    m & \text{otherwise.}
  \end{cases}
\]
For every input value $u$, the branch $\rho(u)$ is a finite partial map
from output values to continuations: an entry $(v \mapsto \sigma) \in
\rho(u)$ reads \emph{``on input value $u$, the node may set $\fielda$ to
$v$ and continue with $\sigma$.''} The finitely many
input values in $\dom(b)$ have their branches listed explicitly, and
\emph{all other} input values $u$ share the same branch, described by
$m$ (the outputs that do not depend on $u$) together with $d$.

Formally, we give wSPPs meaning by "reading them back" into policies.

\begin{definition}
    \label{def:wspp-sem}
    The \emph{read-back policy} $\norm{\rho} \in \pols_{\csr}$ of a wSPP
    $\rho \in \WSPP_{\csr}$ is defined inductively by
    $\norm{\rulebox{\wta}} \triangleq \WEIGH{\wta}{\SKIP}$ and, for
    $\rho = \wspp(\fielda,b,m,d)$,
    \(
    \norm{\rho}
      \;\triangleq\;
      {B_\rho}
      \ADDN
      {M_\rho}
      \ADDN
      {D_\rho},
      ~\text{where}
    \)
    \[
    {B_\rho}
      \triangleq
      \bigoplus_{u \in \dom(b)}
        (\EQ{\fielda}{u}) \SEQN
        \Bigl(
          \bigoplus_{\mathclap{(v \mapsto \sigma) \in b(u)}}
            \;(\ASSN{\fielda}{v}) \SEQN
            \norm{\sigma}
        \Bigr),
    \qquad
    {M_\rho}
      \triangleq
      \Bigl(\prod_{\mathclap{u \in \dom(b)}} \NOTEQ{\fielda}{u}\Bigr) \SEQN
      \Bigl(
        \bigoplus_{\mathclap{(z \mapsto \sigma) \in m}}
          \;(\ASSN{\fielda}{z}) \SEQN
          \norm{\sigma}
      \Bigr),
    \]
    \[
    {D_\rho}
      \triangleq
      \Bigl(\prod_{\mathclap{u \in \dom(b)}} \NOTEQ{\fielda}{u}\Bigr) \SEQN
      \Bigl(\prod_{\mathclap{z \in \dom(m)}} \NOTEQ{\fielda}{z}\Bigr) \SEQN
      \norm{d}~.
    \]
    Here $\prod_i p_i$ denotes $\SEQ{\SEQ{p_1}{\cdots}}{p_n}$; empty
    products are $\SKIP$ and empty sums are $\DROP$. We write
    $\sem{\rho} \triangleq \sem{\norm{\rho}}$ and
    $\rho \equiv \rho'$ for $\norm{\rho} \equiv \norm{\rho'}$.
\end{definition}

The tests make the informal reading precise: the default assignments apply
only on input values not tested by $b$, and the default identity only on
values mentioned nowhere. \Cref{fig:wspp-example} shows a wSPP in all three
guises: as a term, as a tree, and as a policy. The following two crucial facts about the
read-back are proved by
unfolding \Cref{def:wspp-sem} and the packet algebra (E3):

\begin{lemma}
  \label{lemma:wspp-case-split}
  Let $\rho = \wspp(f, b, m, d), \rho' = \wspp(f, b', m', d')$ be wSPPs. Then:
  \begin{enumerate}
    \item For every $u \in \values$:
      $\SEQ{(\EQ{\fielda}{u})}{\norm{\rho}}
       \;\equiv\;
       \SEQ{(\EQ{\fielda}{u})}{\bigoplus_{(v \mapsto \sigma)\in\rho(v)}
         (\ASSN{\fielda}{v}) \SEQN \norm{\sigma}}$.
    \item If for every $u \in \values$:
      $\SEQ{(\EQ{\fielda}{u})}{\norm{\rho}} \equiv
       \SEQ{(\EQ{\fielda}{u})}{\norm{\rho'}}$,
      then $\rho \equiv \rho'$.
  \end{enumerate}
\end{lemma}

\begin{figure}[tb]
\raggedright
\begingroup
\small
\setlength{\tabcolsep}{0pt}
\setlength{\fboxsep}{1.5pt}
\newcommand{\wsppmathbox}[2]{\fcolorbox{#1}{#1!7}{$\displaystyle #2$}}
\newcommand{\wspptextbox}[2]{\fcolorbox{#1}{#1!7}{#2}}
\begin{tabular}{@{}p{0.40\textwidth}@{\hspace{0.012\textwidth}}p{0.588\textwidth}@{}}
\begin{minipage}[t]{\linewidth}
\vspace{0pt}
\centering
\textbf{Graphical Representation}
\vspace{0.35em}

\begin{tikzpicture}[
  >=Latex,
  line join=bevel,
  scale=0.8,
  transform shape,
  field/.style={draw=wsppfield!80!black,circle,minimum size=17pt,inner sep=0pt,font=\footnotesize},
  rootfield/.style={field,fill=wsppfield!20},
  subfield/.style={field,fill=wsppfield!10},
  gate/.style={draw,diamond,minimum size=9pt,inner sep=0pt,fill=black!8},
  leaf/.style={inner sep=1pt,font=\footnotesize},
  elabel/.style={font=\tiny,inner sep=1pt,fill=white},
  branchbox/.style={draw=wsppbranch,fill=wsppbranch!5,fill opacity=0.55,rounded corners=2pt,inner sep=3pt},
  missingbox/.style={draw=wsppmissing,fill=wsppmissing!5,fill opacity=0.55,rounded corners=2pt,inner sep=3pt},
  defaultbox/.style={draw=wsppdefault,fill=wsppdefault!5,fill opacity=0.55,rounded corners=2pt,inner sep=3pt}
]
\node[rootfield] (pt) at (0,7.55) {$\mathsf{pt}$};
\node[gate] (pt80) at (-2.15,5.9) {};
\node[gate] (pt22) at (-0.15,5.9) {};
\node[gate] (ptelse) at (1.45,5.9) {};

\node[subfield] (sw) at (-2.7,4.05) {$\mathsf{sw}$};
\node[leaf] (wone) at (-1.65,4.05) {$\rulebox{w_1}$};
\node[leaf] (wtwo) at (-0.15,4.05) {$\rulebox{w_2}$};
\node[leaf] (wthree) at (0.75,4.05) {$\rulebox{w_3}$};
\node[subfield] (swdefault) at (1.8,4.05) {$\mathsf{sw}$};

\node[gate] (swa) at (-3.05,2.55) {};
\node[gate] (swelse) at (-2.35,2.55) {};
\node[leaf] (wfour) at (-3.15,0.95) {$\rulebox{w_4}$};
\node[leaf] (wfive) at (-2.55,0.95) {$\rulebox{w_5}$};
\node[leaf] (idsw) at (-1.93,0.95) {$\rulebox{\semione}$};
\node[gate] (swdefaultfour) at (1.45,2.55) {};
\node[gate] (swdefaultelse) at (2.1,2.55) {};
\node[leaf] (wsix) at (1.4,0.95) {$\rulebox{w_6}$};
\node[leaf] (idroot) at (2.15,0.95) {$\rulebox{\semione}$};
\coordinate (ptmfit) at (0.88,5.18);
\coordinate (ptdfit) at (1.78,5.18);
\coordinate (swmfit) at (-2.55,1.88);
\coordinate (swdfit) at (-1.91,1.88);
\coordinate (swdfitwide) at (-2.16,1.75);
\coordinate (rootdfit) at (1.8,3.22);
\coordinate (rootbranchfit) at (1.28,1.78);
\coordinate (rootdefaultfit) at (2.28,1.76);

\draw[->] (pt) -- node[elabel,left] {$80$} (pt80);
\draw[->] (pt) -- node[elabel] {$22$} (pt22);
\draw[->,dashed] (pt) -- node[elabel,right] {$*$} (ptelse);
\draw[->] (pt80) -- node[elabel,left] {$443$} (sw);
\draw[->] (pt80) -- node[elabel,right] {$8080$} (wone);
\draw[->] (pt22) -- node[elabel,right] {$2222$} (wtwo);
\draw[->] (ptelse) -- node[elabel,left] {$53$} (wthree);
\draw[->,dashed] (ptelse) -- node[elabel,right] {$*$} (swdefault);

\draw[->] (sw) -- node[elabel,left] {$s_1$} (swa);
\draw[->,dashed] (sw) -- node[elabel,right] {$*$} (swelse);
\draw[->] (swa) -- node[elabel,left] {$s_2$} (wfour);
\draw[->] (swelse) -- node[elabel,left] {$s_3$} (wfive);
\draw[->,dashed] (swelse) -- node[elabel,right] {$*$} (idsw);
\draw[->] (swdefault) -- node[elabel,left] {$s_4$} (swdefaultfour);
\draw[->,dashed] (swdefault) -- node[elabel,right] {$*$} (swdefaultelse);
\draw[->] (swdefaultfour) -- node[elabel,left] {$s_5$} (wsix);
\draw[->,dashed] (swdefaultelse) -- node[elabel,right] {$*$} (idroot);

\begin{scope}[on background layer]
  \node[branchbox,fit=(pt80)(pt22)(sw)(wone)(wtwo)(wfour)(swa)] {};
  \node[missingbox,fit=(ptmfit)(wthree)] {};
  \node[defaultbox,fit=(ptdfit)(swdefault)(rootdfit)(swdefaultfour)(swdefaultelse)(wsix)(idroot)] {};
  \node[branchbox,fit=(swdefaultfour)(rootbranchfit)(wsix)] {};
  \node[defaultbox,fit=(swdefaultelse)(rootdefaultfit)(idroot)] {};
  \node[missingbox,fit=(swmfit)(wfive)] {};
  \node[defaultbox,fit=(swdfit)(swdfitwide)(idsw)] {};
\end{scope}
\end{tikzpicture}
\end{minipage}
&
\begin{minipage}[t]{\linewidth}
\vspace{0pt}
\fontsize{7}{8}\selectfont
\textbf{Syntax}
\[
\begin{aligned}
\rho
  =\wspp(&\mathsf{pt},
    \wspptextbox{wsppbranch}{\(\{80\mapsto\{443\mapsto\rho_s,\;8080\mapsto\rulebox{w_1}\},
      \;22\mapsto\{2222\mapsto\rulebox{w_2}\}\}\)},\\[-0.15em]
    &\wspptextbox{wsppmissing}{\(\{53\mapsto\rulebox{w_3}\}\)},
    \wspptextbox{wsppdefault}{\(\rho_d\)})
\\[0.25em]
\rho_s
  =\wspp(&\mathsf{sw},
    \wspptextbox{wsppbranch}{\(\{s_1\mapsto\{s_2\mapsto\rulebox{w_4}\}\}\)},
    \wspptextbox{wsppmissing}{\(\{s_3\mapsto\rulebox{w_5}\}\)},
    \wspptextbox{wsppdefault}{\(\rulebox{\semione}\)}),
\\[0.25em]
\rho_d
  =\wspp(&\mathsf{sw},
    \wspptextbox{wsppbranch}{\(\{s_4\mapsto\{s_5\mapsto\rulebox{w_6}\}\}\)},
    \wspptextbox{wsppmissing}{\(\varnothing\)},
    \wspptextbox{wsppdefault}{\(\rulebox{\semione}\)}).
\end{aligned}
\]

\vspace{0.3em}
\textbf{Read-back}
\[
\begin{aligned}
\norm{\rho}
  ={}&
  \wsppmathbox{wsppbranch}{
    \mathsf{pt}=80;
      \big(\mathsf{pt}\leftarrow443;\norm{\rho_s}
        \oplus \mathsf{pt}\leftarrow8080;(\WEIGH{w_1}{\SKIP})\big)}
\\[-0.1em]&\oplus
  \wsppmathbox{wsppbranch}{
    \mathsf{pt}=22;\mathsf{pt}\leftarrow2222;(\WEIGH{w_2}{\SKIP})}
\\[-0.1em]&\oplus
  \wsppmathbox{wsppmissing}{
    \mathsf{pt}\neq80;\mathsf{pt}\neq22;
    \mathsf{pt}\leftarrow53;(\WEIGH{w_3}{\SKIP})}
\\[-0.1em]&\oplus
  \wsppmathbox{wsppdefault}{
    \mathsf{pt}\neq80;\mathsf{pt}\neq22;\mathsf{pt}\neq53;
    \norm{\rho_d}}.
\end{aligned}
\]
\vspace{-0.9em}
\begin{center}
\resizebox{\linewidth}{!}{%
\begin{tabular}{@{}c@{\hspace{0.9em}}c@{}}
$\begin{aligned}[t]
\norm{\rho_s}
  ={}&
  \wsppmathbox{wsppbranch}{
    \mathsf{sw}=s_1;\mathsf{sw}\leftarrow s_2;(\WEIGH{w_4}{\SKIP})}
\\[-0.1em]&\oplus
  \wsppmathbox{wsppmissing}{
    \mathsf{sw}\neq s_1;\mathsf{sw}\leftarrow s_3;(\WEIGH{w_5}{\SKIP})}
\\[-0.1em]&\oplus
\wsppmathbox{wsppdefault}{
    \mathsf{sw}\neq s_1;\mathsf{sw}\neq s_3;(\WEIGH{\semione}{\SKIP})}.
\end{aligned}$&
$\begin{aligned}[t]
\norm{\rho_d}
  ={}&
  \wsppmathbox{wsppbranch}{
    \mathsf{sw}=s_4;\mathsf{sw}\leftarrow s_5;(\WEIGH{w_6}{\SKIP})}
\\[-0.1em]&\oplus
  \wsppmathbox{wsppdefault}{
    \mathsf{sw}\neq s_4;(\WEIGH{\semione}{\SKIP})}.
\end{aligned}$
\end{tabular}%
}
\end{center}
\end{minipage}
\end{tabular}
\endgroup
\caption{\footnotesize The three guises of a wSPP $\rho$: syntax (upper
right), read-back policy (lower right), and the graphical
representation as a tree (left). Vertices are labeled by the field the node
branches on.}
\label{fig:wspp-example}
\end{figure}
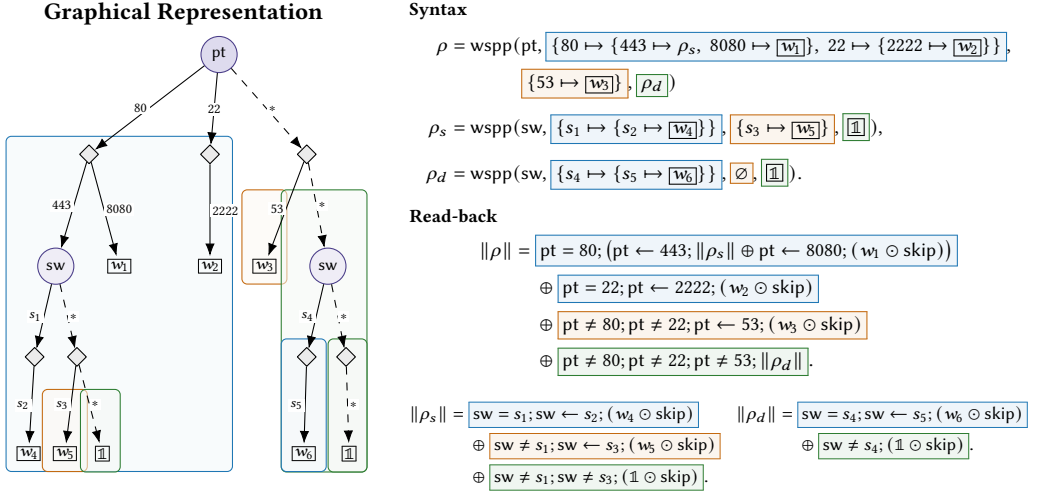


\subsection{Compiling Policies: The Easy Operations}\label{sec:wspp-easy-ops}%
\label{fig:wspp-easy-ops}

Our goal is a compiler $\hat{(-)} \colon \pols_{\csr} \to \WSPP_{\csr}$
such that $\norm{\hat{p}} \equiv p$. The
compiler is structural: atomic policies are translated directly, and each
policy combinator is matched by an operation on wSPPs,
\[
  ({\wsppadd}), ({\wsppmul}) \colon \WSPP_{\csr}^2 \to \WSPP_{\csr},
  \qquad
  ({\wsppscale}) \colon \csrset \times \WSPP_{\csr} \to \WSPP_{\csr},
  \qquad
  ({-})^{\wsppstar} \colon \WSPP_{\csr} \to \WSPP_{\csr},
\]
with precedence mirroring the policy operators ($\wsppscale$ and
$\wsppmul$ bind stronger than $\wsppadd$).
This subsection treats the atoms and the first three operations, which
generalize their SPP counterparts in \katch~\cite{KATch}. Hence, we keep the
discussion brief and refer to \Cref{sec:appendix-algorithms} for the
complete algorithms. The star, which has no SPP counterpart, is treated in \Cref{sec:wspp-star}.

\smallskip\noindent\textbf{Atoms.}
A test or assignment $e$ on a single field becomes a one-node wSPP $\strhat{e}$,
defined such that
{\small
\[
\begin{gathered}
  \strhat{(\EQ{\fielda}{v})} \triangleq
    \wspp(\fielda, \{v \mapsto \{v \mapsto \hat{\mathbb{1}}\}\},
      \varnothing, \hat{\mathbb{0}}),
  \qquad
  \strhat{(\NOTEQ{\fielda}{v})} \triangleq
    \wspp(\fielda, \{v \mapsto \varnothing\},
      \varnothing, \hat{\mathbb{1}}),\\
  \strhat{(\ASSN{\fielda}{v})} \triangleq
    \wspp(\fielda, \varnothing,
      \{v \mapsto \hat{\mathbb{1}}\}, \hat{\mathbb{0}}),
  \qquad
  \strhat{\TRUE} \triangleq \hat{\mathbb{1}},
  \qquad
  \strhat{\FALSE} \triangleq \hat{\mathbb{0}}.
\end{gathered}
\]
}
For instance,
$\strhat{(\NOTEQ{\fielda}{v})}$ tests value $v$ and offers it no
continuation (the empty row), while all other values pass unchanged into the
default identity $\hat{\mathbb{1}}$. Compound tests compiles as follows:
conjunction compiles to sequential composition, negation is pushed to the literals by
De~Morgan, and disjunction $\OR{t}{t'}$ compiles to the \emph{disjoint}
sum $\ADD{t}{(\SEQ{\NOT{t}}{t'})}$ to respect non-idempotence
(\Cref{example:non-idempotent}). 

\smallskip\noindent\textbf{Choice, weighting, sequencing.} To define weighting $\wsppscale$,
The weight is simply pushed to the leaves:
$\wta \wsppscale \rulebox{\wtb} \triangleq \rulebox{\wta \cdot \wtb}$, and
$\wta \wsppscale \wspp(\fielda,b,m,d)$ applies $\wta \wsppscale (-)$ to every
sub-wSPP.
The remaining two operations are best understood through the branch maps. Consider first two nodes
$\rho, \sigma$ on the \emph{same} field $\fielda$. Then
$\rho \wsppadd \sigma$ and $\rho \wsppmul \sigma$ are the
nodes on $\fielda$ whose branch maps satisfy, for every $v \in \values$,
\[
  (\rho \wsppadd \sigma)(u) \;=\; \rho(u) \wsppmapadd \sigma(u)
  \qquad\text{and}\qquad
  (\rho \wsppmul \sigma)(u)
    \;=\; \bigwsppmapadd_{\mathclap{(v \mapsto \tau)\,\in\,\rho(u)}}\;
      \tau \wsppmapmul \sigma(v),
\]
where $({\wsppmapadd})$ merges two branches by combining entries for the
same output value with $\wsppadd$ (missing entries count as
$\hat{\mathbb{0}}$), and
$\tau \wsppmapmul row
    \;\triangleq\;
    \{\, v \mapsto \tau \wsppmul \eta  \mid (v \mapsto \eta) \in row \}$ pre-composes a branch with a continuation. Namely, choice
merges branches, and sequencing routes every output value $v$ of
$\rho$ into the branch of $\sigma$ reading $v$, which could be seen as the
symbolic rendition of the matrix product with $({\wsppmapadd})$ summing over intermediate packets.
By defining $\rho \wsppadd \sigma$ and $\rho \wsppmul \sigma$ so that the above branch
map rows $(\rho \wsppadd \sigma)(u), (\rho \wsppmul \sigma)(u)$ with
identical behavior are represented once through default assignments and identity, the
operations preserve the compactness of the wSPP representation.
If the root fields differ, say
$\fielda_\rho < \fielda_\sigma$, then $\sigma$ is first \emph{lifted} to a
virtual node on $\fielda_\rho$ with empty $b$ and $m$ and default identity
$\sigma$, before choice and sequencing are carried out.
Fot these operations, we get:

\begin{theorem}
\label{theorem:wspp-add-mul-soundness}
For all wSPPs $\rho, \sigma \in \WSPP_{\csr}$, weights
$\wta \in \csrset$, and atoms $e$,
\[
\norm{\rho \wsppadd \sigma} \equiv \ADD{\norm{\rho}}{\norm{\sigma}},
\qquad
\norm{\rho \wsppmul \sigma} \equiv \SEQ{\norm{\rho}}{\norm{\sigma}},
\qquad
\norm{\wta \wsppscale \rho} \equiv \WEIGH{\wta}{\norm{\rho}},
\qquad
\norm{\strhat{e}} \equiv e.
\]
\end{theorem}


\input{syntax-semantics-default-first-star}

\subsection{Computing the Semantics of a Policy}\label{sec:wspp-soundness}

Assembling the operations gives the compiler by structural recursion. By induction, we thus get:

\begin{theorem}
  \label{theorem:wspp-main}
  For every policy $p \in \pols_{\csr}$, the wSPP $\hat{p}$ is computable,
  and $\norm{\hat{p}} \equiv p$.
\end{theorem}

The wSPP $\hat{p}$ is the promised finite handle on $\sem{p}$: a symbolic
matrix whose size is governed by $p$, not by $|\packets|$. Individual
weights are read off in one pass over the fields:

\begin{corollary}[reading off weights]
  \label{corollary:wspp-weights}
  Given packets $\pkta, \pktb$, the weight
  $\sem{p}(\pkta)(\pktb)$ is obtained from $\hat{p}$ by
  descending from the root: at a node $\nu$ on field $\fielda$, follow the
  branch $\nu(\pkta.\fielda)(\pktb.\fielda)$, returning $\semizero$
  if it is absent; at a leaf $\rulebox{a}$, return $a$ if $\pkta$ and
  $\pktb$ agree on all fields not visited along the way, and $\semizero$
  otherwise.
\end{corollary}

\FloatBarrier

\input{syntax-semantics-f9d-wspp-star}
\FloatBarrier

%% file: syntax-semantics-default-first-star.tex
\subsection{The Kleene Star of a wSPP}\label{sec:wspp-star}\label{sec:wspp-star-default-first}

Kleene star is where the weighted setting departs significantly \netkat's \katch. Iterating a node moves packets \emph{through the values
of a field}, and weights accumulate along cycles that visit several values.
The following will serve as a running example:
\begin{example}[running example]
  \label{example:df-wspp-star-cycles}
  Over $\csr = \Qext$, consider the one-field wSPP
  \[
    \rho \;=\; \wspp\bigl(\fielda,\;
    \textcolor{wsppbranch}{\{1 \mapsto \{1 \mapsto \rulebox{\nicefrac{1}{2}},\;
      4 \mapsto \rulebox{1}\},\;
      2 \mapsto \{1 \mapsto \rulebox{\nicefrac{1}{2}}\}\}},\;
    \textcolor{wsppmissing}{\{1 \mapsto \rulebox{\nicefrac{1}{5}},\;
      3 \mapsto \rulebox{\nicefrac{1}{4}}\}},\;
    \textcolor{wsppdefault}{\rulebox{\nicefrac{1}{3}}}\bigr),
  \]
  drawn on the left of \Cref{fig:df-star-run-example}: on input value $1$, the
  node may keep the value $1$ (weight $\nicefrac{1}{2}$) or write $4$
  (weight $1$); on input value $2$, it may write $1$ (weight
  $\nicefrac{1}{2}$); and on any untested input value it may write $1$
  (weight $\nicefrac{1}{5}$) or $3$ (weight $\nicefrac{1}{4}$), and may also
  keep its value through the default identity when neither default assignment
  names that value. Now, under $\ITER{\norm{\rho}}$, a packet with
  $\fielda = 2$ may reach an output with $\fielda = 1$ in several ways. One
  family first takes the edge $2 \to 1$, then loops at $1$ any number of
  times before stopping. Another takes $2 \to 1$, then $1 \to 4$, then loops
  at $4$ through the default identity any number of times before using the
  default assignment back to $1$. The sum of weights of the paths in the two families
  is given by
  \[
    \tfrac{1}{2}\sum_{k \in \N}\bigl(\tfrac{1}{2}\bigr)^k
    \;+\;
    \tfrac{1}{2}\cdot 1 \cdot
      \sum_{j \in \N}\bigl(\tfrac{1}{3}\bigr)^j\cdot\tfrac{1}{5}
    \;=\;
    \tfrac{1}{2}\cstar{(\tfrac{1}{2})}
    + \tfrac{1}{10}\cstar{(\tfrac{1}{3})}
    \;=\; 1+\tfrac{3}{20}=\tfrac{23}{20}.
  \]
  Thus the $2$-to-$1$ component of $\rho^{\wsppstar}$ must account for at
  least this much weight, while also considering the interleaving of paths 
  from the two families and other complications. \hfill $\triangle$
\end{example}

On the one hand, the only way to
sum infinitely many weights is the star $\cstar{(-)}$ of
$\csr$ and no amount of iterating $({\wsppadd})$ and $({\wsppmul})$ can
replace it (\Cref{sec:wspp-star}). On the other hand,
$\cstar{(-)}$ sums the powers $\semione, w, w^2, \ldots$ of a
\emph{single} weight, whereas the runs of a node consist of loops spanning
\emph{several} values, each contributing its own weight
(\Cref{example:df-wspp-star-cycles}). The star algorithm must therefore
\emph{disentangle} the iteration using the
equational laws of \Cref{sec:preliminaries} such that every remaining 
infinite sum is the sum of powers of one single weight which can then be delegated to $\cstar{(-)}$.

\begin{figure}[tb]
\centering
\begin{tikzpicture}[scale=.7,
  x=0.78cm, y=0.55cm,
  cell/.style={font=\footnotesize, inner sep=1.5pt},
  rowlab/.style={font=\footnotesize\itshape},
  collab/.style={font=\footnotesize\itshape},
]
  \foreach \c/\l in {1/1, 2/2, 3/3, 4/4, 5/\cdots}
    \node[collab] at (\c, 0.7) {$\l$};
  \node[collab,anchor=east] at (0.1, 0.7) {\scriptsize in $\backslash$ out};
  \foreach \r/\l in {1/1, 2/2, 3/3, 4/4, 5/\vdots}
    \node[rowlab] at (0, -\r) {$\l$};
  \draw[black!25] (0.5,0.2) grid[xstep=1, ystep=1] (5.5,-5.5);
  \node[cell,text=wsppbranch] at (1,-1) {$\rulebox{\nicefrac{1}{2}}$};
  \node[cell,text=wsppbranch] at (4,-1) {$\rulebox{1}$};
  \node[cell,text=wsppbranch] at (1,-2) {$\rulebox{\nicefrac{1}{2}}$};
  \node[cell,text=wsppmissing] at (1,-3) {$\rulebox{\nicefrac{1}{5}}$};
  \node[cell,text=wsppmissing] at (3,-3) {$\rulebox{\nicefrac{1}{4}}$};
  \node[cell,text=wsppmissing] at (1,-4) {$\rulebox{\nicefrac{1}{5}}$};
  \node[cell,text=wsppmissing] at (3,-4) {$\rulebox{\nicefrac{1}{4}}$};
  \node[cell,text=wsppmissing] at (1,-5) {$\rulebox{\nicefrac{1}{5}}$};
  \node[cell,text=wsppmissing] at (3,-5) {$\rulebox{\nicefrac{1}{4}}$};
  \node[cell,text=wsppdefault] at (4,-4) {$\rulebox{\nicefrac{1}{3}}$};
  \node[cell,text=wsppdefault] at (5,-5) {$\rulebox{\nicefrac{1}{3}}$};
  \node[anchor=west, font=\footnotesize, text=wsppbranch] at (6.3,-1)
    {rows $1,2$: explicit branches $\textcolor{wsppbranch}{b}$};
  \node[anchor=west, font=\footnotesize, text=wsppmissing] at (6.3,-2.3)
    {columns $1,3$: default assignments $\textcolor{wsppmissing}{m}$,};
  \node[anchor=west, font=\footnotesize, text=wsppmissing] at (6.6,-3.1)
    {\footnotesize shared by every untested row};
  \node[anchor=west, font=\footnotesize, text=wsppdefault] at (6.3,-4.3)
    {diagonal: the default identity $\textcolor{wsppdefault}{d}$, on rows};
  \node[anchor=west, font=\footnotesize, text=wsppdefault] at (6.6,-5.1)
    {\footnotesize mentioned by neither $b$ nor $m$};
\end{tikzpicture}
\caption{\footnotesize A wSPP node is a finitely presented
$\values \times \values$ matrix: the running example $\rho$
(\Cref{example:df-wspp-star-cycles}), with the entry in row $u$ and column $v$
being the continuation taken when $\fielda$ is rewritten from $u$ to $v$.
Empty positions are $\hat{\mathbb{0}}$.}
\Description{A five-by-five matrix schema with one explicit row, one
shared column, and a diagonal band, illustrating how the three components
of a wSPP node populate a value-by-value matrix.}
\label{fig:df-star-matrix}
\end{figure}
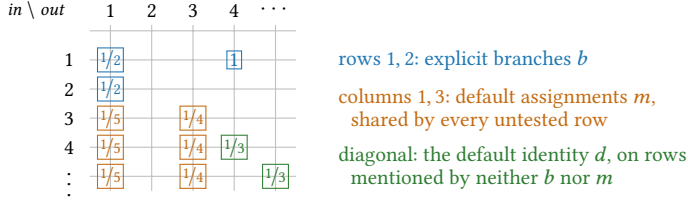

\smallskip\noindent\textbf{Key idea: remove one entry at a time.}
Consider Fig. \ref{fig:df-star-matrix}: a wSPP node can be viewed as a
a matrix whose rows are input values, columns are output values. 
Note that row $u$ is just the branch $\rho(u)$ of \Cref{sec:wspp}. Each 
$(v \mapsto \sigma) \in \rho(u)$ corresponds to an \emph{entry at $(u,v)$},
read \emph{``at value $u$, one step of $\rho$ may move to value $v$''.}
Rows in $\dom(b)$ are stored explicitly,
$m$ becomes a block of columns shared by all untested rows, and $d$ populates the
diagonal of the rows that are mentioned neither in $b$ nor in $m$. We can view
the values of $\fielda$ as states of a finite automaton with this matrix as
its transition matrix. The classical way to compute the unbounded iteration for such an automaton
is \textit{state elimination}; our algorithm carries this idea out
symbolically over wSPPs: \emph{remove one
entry at a time, until no entries remain}. Removing an
entry destroys paths and must be compensated:
\begin{enumerate}
  \item \emph{emit a factor}: a wSPP recording
    every contribution of the removed entry on its own by taking it never,
    once, or arbitrarily often (if possible), and
  \item \emph{update the remaining entries}: suppose the removed entry is at
    $(u,v)$, a path that arrives at value $u$ through another 
    entry needs to still be able to reach $v$, so its effect must be "spliced" onto
    every remaining entry \emph{into} $u$. 
\end{enumerate}
Once every entry is gone, no path is left to iterate, and the star is
simply the sequential product of the emitted factors. We note a crucial 
subtlety: A factor could be split off either to the left or right of the 
remaining star (corresponding to the effect of the factor being "spliced" 
before or after the remaining entries), so the two lists $\vec{\psi}$ and $\vec{\phi}$
are needed to store the emitted factors such that 
\[
  \rho^{\wsppstar}
  \;=\;
  \psi_1 \wsppmul \cdots \wsppmul \psi_k
  \wsppmul \phi_l \wsppmul \cdots \wsppmul \phi_1,
  \quad
  \text{where each factor has shape }
  \quad
  \hat{\mathbb{1}} \wsppadd
  \underbrace{\strhat{g} \wsppmul \strhat{a} \wsppmul \eta}_{%
    \text{\footnotesize traverse the removed entry}}. 
\]
Here
$g$ is a product of tests, the assignment $a$ may be
absent, and $\eta$ is the \emph{continuation} of iterating 
the removed entry after the toplevel field of 
$\rho$ has been read from and assigned to using $g$ and $a$. 

Let us carry out one removal by hand for both compensations before
stating the algorithm.

\begin{example}
  \label{example:df-wspp-star-step}
  We remove the default identity of the running example $\rho$ and work out
  both compensations; the general case is the same calculation with heavier
  bookkeeping. Abbreviate
  $t_v \triangleq \EQ{\fielda}{v}$,
  $\bar{t}_v \triangleq \NOTEQ{\fielda}{v}$, and
  $a_v \triangleq \ASSN{\fielda}{v}$, and recall
  $\norm{\rulebox{w}} = \WEIGH{w}{\SKIP}$; unit weights are dropped
  tacitly, $\WEIGH{\semione}{p} \equiv p$ (E2d).
  By \Cref{def:wspp-sem} and (E2a),
  $\norm{\rho} \equiv \ADD{D_\rho}{R}$, where the summand being removed is
  \[
    D_\rho \triangleq
    \bar{t}_1;\bar{t}_2;\bar{t}_3;(\WEIGH{\nicefrac13}{\SKIP}),
  \]
  one step through the default identity, and $R$ collects the remaining
  summands:
  {\small
  \[
    R \;\triangleq\;
    \underbrace{t_1;(a_1;(\WEIGH{\nicefrac12}{\SKIP})
      \oplus a_4;\SKIP)
      \oplus t_2;a_1;(\WEIGH{\nicefrac12}{\SKIP})}_{\textcolor{wsppbranch}{B_\rho}}
    \;\oplus\;
    \underbrace{\bar{t}_1;\bar{t}_2;
      \bigl(a_1;(\WEIGH{\nicefrac{1}{5}}{\SKIP})
      \oplus a_3;(\WEIGH{\nicefrac{1}{4}}{\SKIP})\bigr)}_{\textcolor{wsppmissing}{M_\rho}}.
  \]
  }
  \emph{Part 1: the factor.} The factor must capture
  everything the removed entry contributes on its own by taking it never,
  once, twice, \ldots, i.e., it must represent $\ITER{D_\rho}$. A
  closed form exists because two consecutive $D_\rho$-steps merge into one: the weighted $\SKIP$ mentions no field, so it commutes with
  $\bar{t}_1$, $\bar{t}_2$, and $\bar{t}_3$ (E3d), the tests are
  idempotent, and the two weights merge (E2c, E2d):
  \[
    D_\rho;D_\rho
    \;\equiv\;
    \bar{t}_1;\bar{t}_2;\bar{t}_3;
      \gbox{(\WEIGH{\nicefrac13}{\SKIP});\bar{t}_1;\bar{t}_2;\bar{t}_3};
      (\WEIGH{\nicefrac13}{\SKIP})
    \;\equiv\;
    \bar{t}_1;\bar{t}_2;\bar{t}_3;(\WEIGH{\nicefrac19}{\SKIP}).
  \]
  Extending the argument inductively then gives 
  $D_\rho^{(k)} \equiv
  \bar{t}_1;\bar{t}_2;\bar{t}_3;(\WEIGH{(\nicefrac13)^k}{\SKIP})$
  for $k \geq 1$, which means taking the default identity $k$ times only multiplies
  its weight $k$ times. The infinite sum
  $\ITER{D_\rho} = \bigoplus_k D_\rho^{(k)}$ has thereby been
  disentangled and collapses onto the sum of powers of the single weight
  $\nicefrac{1}{3}$, which is exactly what $\cstar{(-)}$ evaluates,
  $\sum_{k \geq 1}(\nicefrac{1}{3})^k
   = \nicefrac{1}{3} \cdot \cstar{(\nicefrac13)} = \nicefrac12$
  (the countable sum passes through sequencing and weighting by
  $\omega$-continuity, cf.\ E2a, E2b):
  \[
    \ITER{D_\rho}
    \;\equiv\;
    \SKIP \,\oplus\, \bar{t}_1;\bar{t}_2;\bar{t}_3;
      \bigl(\WEIGH{\textstyle\sum_{k\geq1}(\nicefrac13)^k}{\SKIP}\bigr)
    \;\equiv\;
    \SKIP \,\oplus\,
      \bar{t}_1;\bar{t}_2;\bar{t}_3;(\WEIGH{\nicefrac12}{\SKIP})
    \;\equiv\;
    \norm{\psi_1},
  \]
  where $\psi_1 \triangleq \hat{\mathbb{1}} \wsppadd
  \strhat{\bar{t}_1}\wsppmul\strhat{\bar{t}_2}\wsppmul
  \strhat{\bar{t}_3}\wsppmul \eta$ is the
  emitted factor and \emph{closure} $\eta \triangleq \rulebox{\nicefrac12}$
    is the continuation of iterating the default identity. In
  general, the default identity is not a leaf but a sub-wSPP $d$ on the
  remaining fields; the continuation is then
  $\eta = d \wsppmul \Star(d)$, delivered by a \emph{recursive call} on
  strictly fewer fields.

  \emph{Part 2: the update.} What is left of the star once the factor is
  split off? This is precisely the question \textsc{denesting}
  (\Cref{theorem:star-eqns}) answers:
  \[
    \ITER{\norm{\rho}}
    \;=\;
    \ITER{(\ADD{D_\rho}{R})}
    \;\equiv\;
    \SEQ{\ITER{D_\rho}}{\ITER{(\SEQ{R}{\ITER{D_\rho}})}}
    \;\equiv\;
    \SEQ{\norm{\psi_1}}{\ITER{(\SEQ{R}{\norm{\psi_1}})}}.
  \]
  Left to star is $\SEQ{R}{\norm{\psi_1}} \equiv R \oplus
  R;\bar{t}_1;\bar{t}_2;\bar{t}_3;\norm{\eta}$ (E2a, E2d): each remaining
  summand, optionally followed by one last pass through the removed default
  identity. For the algorithm to continue, we'd hope that 
  $\SEQ{R}{\norm{\psi_1}}$ is semantically equivalent to a smaller \textit{residual } wSPP on which 
  a new round of state elimination could begin. This is exactly the case:  
  A summand of $R$ contributes to $\SEQ{R}{\norm{\psi_1}}$ only if it can
  pass the tests $\bar{t}_1$, $\bar{t}_2$, and $\bar{t}_3$ of $\norm{\psi_1}$, i.e., only if it
  \emph{writes to f some value outside the values already named by $b$ and $m$}. In our
  example, only the entry at $(1,4)$ does, since the $a_1$ and $a_3$ assignments
  for the entries $(1, 1), (2, 1)$ and the default assignments all write blocked values, while $a_4$
  writes a value that the guard allows. Hence the second part simply extends
  that entry by the closed identity:
  \[
    t_1;a_4;\SKIP;\bar{t}_1;\bar{t}_2;\bar{t}_3;\norm{\eta}
    \;\equiv\;
    t_1;a_4;(\WEIGH{\nicefrac12}{\SKIP}),
  \]
  which combines with the original $(1,4)$ entry to give
  $t_1;a_4;(\WEIGH{1+\nicefrac12}{\SKIP})$. Therefore
  \[
    \ITER{\norm{\rho}}
    \;\equiv\;
    \norm{\psi_1} \SEQN \ITER{\norm{\rho_1}},
    \qquad
    \rho_1 \;=\; \wspp(\fielda,
    \{1 \mapsto \{1 \mapsto \rulebox{\nicefrac12},\;
      4 \mapsto \rulebox{\nicefrac32}\},\;
      2 \mapsto \{1 \mapsto \rulebox{\nicefrac12}\}\}, m, \hat{\mathbb{0}})
  \]
  as we'd hoped. In general, let $S = \dom(b) \cup \dom(m)$ and let
  $e_S = \prod_{x\in S} \NOTEQ{\fielda}{x}$. Removing the default identity
  emits $\hat{\mathbb{1}} \wsppadd \strhat{e_S} \wsppmul
  d \wsppmul \Star(d)$, sets the residual default identity to
  $\hat{\mathbb{0}}$, and updates exactly the explicit entries whose output
  value is not in $S$:
  \(
    b(u)(v) \;\leftarrow\; b(u)(v) \wsppmul \Star(d)
    \qquad (v \notin S).
  \) 
  Same as above, the default-assignments $m$ is unchanged. The rest
  of the algorithm then continues by editing the residual components
  $b$, $m$, and $d$. \hfill $\triangle$
\end{example}

\begin{figure}[tb]
\centering
\begin{algorithm}[H]
\DontPrintSemicolon\SetAlgoVlined\small
\SetCommentSty{AlgCommentFont}
\Fn{\Star{$\rho$}}{
  \lIf{$\rho = \rulebox{w}$}{%
    \KwReturn $\rulebox{\cstar{w}}$
    \tcp*[f]{the only use of $\cstar{(-)}$}}
  let $\rho = \wspp(\fielda,b,m,d)$;\quad
  $\vec\psi \leftarrow [\,]$;\ $\vec\phi \leftarrow [\,]$\;
  \BlankLine
  \tcc{{\normalfont\textcolor{wsppdefault}{\bfseries Pass I: remove the
    default identity $d$}}\quad --- maintains invariant
    \textnormal{(I)}}
  \If{$d \neq \hat{\mathbb{0}}$}{
    \AlgPointStart{wsppdefault}{I-s}
    \ForEach(\tcp*[f]{entries that can fall through $d$}){$u \in \dom(b)$\,,\, $v \in \dom(b(u)) \setminus (\text{dom}(b) \cup \text{dom}(m))$}{
      $b(u)(v) \leftarrow b(u)(v) \wsppmul \texttt{Star}(d)$\;
    }
    $\vec\psi \leftarrow \Bigl(\hat{\mathbb{1}} \wsppadd
      \guardprod{x\in \text{dom}(b)}{(\fielda\ne u)}
      \wsppmul \guardprod{x\in \text{dom}(m)}{(\fielda\ne z)} \wsppmul
      d \wsppmul  \texttt{Star}(d)\Bigr) :: \vec\psi$\;
    \AlgPointEnd{wsppdefault}{I-e}
    $d \leftarrow \hat{\mathbb{0}}$\;
  }
  \BlankLine
  \tcc{{\normalfont\textcolor{wsppbranch}{\bfseries Pass II: remove the
    entries of $b$}}\quad --- maintains invariant
    \textnormal{(I)}}
  \While{$\exists u \in \dom(b)\text{ with } b(u) \neq \varnothing$}{
    \AlgPointStart{wsppbranch}{II-s}
    $u \leftarrow$ smallest key of $b$ with $b(u) \neq \varnothing$\;
    $v \leftarrow \textbf{if } u \in \dom(b(u)) \textbf{ then } u
      \textbf{ else } \text{smallest key of } b(u)$;\quad
    $\sigma \leftarrow b(u)(v)$
    \tcp*[f]{select the entry $(u,v)$}\;
    $\eta \leftarrow \textbf{if } v = u \textbf{ then }
      \sigma \wsppmul \Star(\sigma) \textbf{ else } \sigma$
    \tcp*[f]{close the self-loop: recursion, strictly fewer fields}\;
    delete the entry $(u,v)$, i.e.\ remove $v$ from $b(u)$\;
    \ForEach(\tcp*[f]{steps into $u$ may still use the deleted entry}){$u' \in \dom(b)$
      \KwwithKw $u \in \dom(b(u'))$}{
      $b(u')(v) \leftarrow \dflt{b(u')(v)} \wsppadd b(u')(u) \wsppmul \eta$\;}
    \lIf{$u \in \dom(m)$}{
      $m(v) \leftarrow \dflt{m(v)} \wsppadd m(u) \wsppmul \eta$}
    \AlgPointEnd{wsppbranch}{II-e}
    $\vec\psi \leftarrow \bigl(\hat{\mathbb{1}} \wsppadd
      \strhat{(\fielda{=}u)} \wsppmul \strhat{(\fielda{\leftarrow}v)}
      \wsppmul \eta\bigr) :: \vec\psi$\;
  }
  \BlankLine
  \tcc{{\normalfont\textcolor{wsppmissing}{\bfseries Pass III: remove the
    default assignments $m$}}\quad --- maintains invariant
    \textnormal{(I)}}
  \While{$m \neq \varnothing$}{
    \AlgPointStart{wsppmissing}{III-s}
    $z \leftarrow$ smallest key of $m$;\quad
    $\sigma \leftarrow m(z)$;\quad delete $z$ from $m$\;
    \uIf(\tcp*[f]{loops at $z$ were already removed in Pass II}){$z \in \dom(b)$}{
      \AlgPointEndLow{wsppmissing}{III-e}
      $\vec\phi \leftarrow \bigl(\hat{\mathbb{1}} \wsppadd
        \guardprod{u\in \dom(b)}{(\fielda\ne u)} \wsppmul
        \strhat{(\fielda{\leftarrow}z)} \wsppmul \sigma\bigr) :: \vec\phi$\;
    }\Else(\tcp*[f]{writing $z$ re-enters $m(z)$: close that loop too}){
      \lForEach{$z' \in \dom(m)$}{
        $m(z') \leftarrow \texttt{Star}(\sigma) \wsppmul m(z')$}
      \AlgPointEndLow{wsppmissing}{III-e}
      $\vec\phi \leftarrow \bigl(\hat{\mathbb{1}} \wsppadd
        \guardprod{u\in \dom(b)}{(\fielda\ne u)} \wsppmul
        \strhat{(\fielda{\leftarrow}z)} \wsppmul \sigma \wsppmul \texttt{Star}(\sigma)\bigr) :: \vec\phi$\;
    }
  }
  \KwReturn $\bigl(\listbprod{\vec\psi}\bigr)\wsppmul\bigl(\listfprod{\vec\phi}\bigr)$
  \tcp*[f]{$\bprod$/$\fprod$: product in reverse/forward list order}\;
}
\end{algorithm}
\caption{The wSPP star algorithm: symbolic elimination of the
root node. $\dflt{g(x)}$ denotes $g(x)$ if $x \in \dom(g)$ and
$\hat{\mathbb{0}}$ otherwise; wide-hatted atoms denote their atomic wSPPs,
a wide hat $\Pi$ applies the $\wsppmul$-product
(\Cref{sec:wspp-easy-ops}).}
\Description{Pseudocode of the Star function with three passes: removing
the default identity, test-assignment entries, and default assignments.}
\label{alg:df-star}\label{alg:star}
\end{figure}

\smallskip\noindent\textbf{The algorithm.}
\Cref{alg:df-star} assembles these ingredients into $\Star(\rho)$. After
applying $\cstar{(-)}$ to leaves, the
algorithm removes the entries of the root node in three passes: first the
single default identity $d$ (Pass~I), then the explicit entries of $b$
(Pass~II), then the columns of $m$ (Pass~III). The emitted factors are
collected into two lists, $\vec\psi$ for the left and $\vec\phi$ for the
right of the final product, as anticipated by the two \textsc{denesting}
variants above.
\emph{Pass~I} removes $d$ exactly as in
\Cref{example:df-wspp-star-step}: it emits the closed default-identity
factor, multiplies by $\Star(d)$ each explicit entry whose output value can
fall through the old default identity, and then sets $d$ to
$\hat{\mathbb{0}}$. \emph{Pass~II} removes the entries of $b$ one at a
time, preferring self-loop entries $(u,u)$, whose factors close the loop via
a recursive call; when other entries flow into the removed entry, their
routes are spliced into the residual wSPP. After removing all entries
from a row $b(u)$ of $b$, the $u$-entry in $b$ is still retained
because it keeps the row explicit and prevents the defaults $m$ and
$d$ from filling it in. \emph{Pass~III} removes the
default assignments column by column: for example, value $z$ that is not an
explicit row leaves the packet able to take column $z$ again indefinitely, 
so the factor closes the $(z, z)$ self-loop with
$\sigma \wsppmul \Star(\sigma)$, and $\Star(\sigma)$ is pre-composed into the remaining
columns, which a run may take after looping at $z$.

\smallskip\noindent\textbf{Why this is correct: the invariant.}
At any point during the run, the loop
state consists of the current residual wSPP $\rho'$ and the factor lists
$\vec\psi, \vec\phi$. The algorithm maintains, after Pass~I and at every
iteration of Passes~II--III,
\(
 \quad \ITER{\norm{\rho}}
  \;\equiv\;
  \norm{\listbprod{\vec\psi}}
  \SEQN \ITER{\norm{\rho'}} \SEQN
  \norm{\listfprod{\vec\phi}}
  \quad\text{(I)}.
\)
It remains only to check that each primitive update reestablishes
\textnormal{(I)} with the newly emitted factor, and that a residual with no
remaining entries contributes the identity to the surrounding product.

\begin{lemma}[preservation and progress]
  \label{lemma:df-star-invariant}
  In any of the three passes of the algorithm, let $\rho'$ be the residual
  $\rho$ at program point \AlgPoint{wsppdefault}{I-s},
  \AlgPoint{wsppbranch}{II-s}, or \AlgPoint{wsppmissing}{III-s},
  and let $\rho''$ be the residual at the corresponding point
  \AlgPoint{wsppdefault}{I-e}, \AlgPoint{wsppbranch}{II-e}, or
  \AlgPoint{wsppmissing}{III-e} of Figure \ref{alg:star}. Then
  each step of \Cref{alg:df-star} preserves
  \textnormal{(I)}: Pass~I and each Pass~II removal satisfy
  $\ITER{\norm{\rho'}} \equiv
  \SEQ{\norm{\psi}}{\ITER{\norm{\rho''}}}$ for the emitted factor $\psi$;
  each Pass~III removal satisfies $\ITER{\norm{\rho'}} \equiv
  \SEQ{\ITER{\norm{\rho''}}}{\norm{\phi}}$ for the emitted factor $\phi$.
  Upon termination, for the final residual $\rho_{\mathsf{fin}}$,
  $b(u)=\varnothing$ for every $u\in\dom(b)$, $m = \varnothing$, and
  $d = \hat{\mathbb{0}}$, hence
  $\norm{\rho_{\mathsf{fin}}} \equiv \DROP$ and
  $\ITER{\norm{\rho_{\mathsf{fin}}}} \equiv \SKIP$.
\end{lemma}

\smallskip\noindent\textbf{Termination.}
Termination is subtle: the algorithm
\emph{rewrites} the maps rather than recursing into subterms, and single
steps may grow them. No naive syntactic size decreases. What \emph{does} decrease is the
lexicographic measure
\[
  \Bigl\langle\;
    |\fieldsof{\rho}|,\;\;
    \text{pass},\;\;
    \mathbf{1}_{d\neq\hat{\mathbb{0}}},\;\;
    |\{u \in \dom(b) \mid b(u)\neq\varnothing\}|,\;\;
    |b(u_{\min})|,\;\;
    |m|
  \;\Bigr\rangle,
\]
where the pass component is ordered by remaining work,
Pass~I $>$ Pass~II $>$ Pass~III, and $u_{\min}$ is the least key of $b$
whose row is currently nonempty (with the fifth component taken to be $0$
if no such row exists).
Recursive calls decrease the first component: by well-formedness,
sub-wSPPs omit the root field and their stars mention no new fields
(\Cref{theorem:df-wspp-star-soundness}), so the closures do not reintroduce
it. Pass~I update decreases the default-identity bit; advancing
from Pass~I to Pass~II decreases the pass component. In Pass~II, the
selected row $b(u_{\min})$ loses one entry; if
that entry was the last one, the number of nonempty rows decreases. Crucially, 
updates may add entries to later rows or to $m$, but
they never refill the selected row. 
Finally, Pass~III shrinks $|m|$ and does not introduce new keys. Since
$\fields$ is finite, the recursion depth is bounded by $|\fields|$, and the
algorithm terminates.

Together with \Cref{lemma:df-star-invariant}, we finally get:
\begin{theorem}
  \label{theorem:df-wspp-star-soundness}
  For every $\rho \in \WSPP_{\csr}$, \Cref{alg:df-star} terminates and
  returns a well-formed wSPP
  $\rho^{\wsppstar} \triangleq \Star(\rho) \in \WSPP_{\csr}$ that mentions
  only fields of $\rho$ and satisfies
  $\norm{\rho^{\wsppstar}} \equiv \ITER{\norm{\rho}}$.
\end{theorem}

\smallskip\noindent\textbf{A complete run.}
\Cref{fig:df-star-run-example} traces \Cref{alg:df-star} on the running
example, exercising all three passes. The first pass closes the diagonal
$d$, pushes $\Star(d)$ into the explicit entry $(1,4)$ that can fall through
to the unnamed value $4$, and leaves a residual wSPP with
$d = \hat{\mathbb{0}}$. Pass~II then removes the four explicit entries in
the two retained rows of $b$; the row markers remain even after their maps
become empty. Pass~III removes the default assignments, including the
new column $4$ introduced while splicing routes through $(1,4)$. Multiplying
the emitted factors (with $\wsppmul$ from \Cref{sec:wspp-easy-ops}) yields
$\rho^{\wsppstar}$.

\begin{figure}[tb]
\centering
\begingroup
\scriptsize
\setlength{\fboxsep}{0.8pt}
\newsavebox{\runfactorbox}
\newcommand{\runcard}[4]{%
\begin{minipage}[t]{#1}
\centering
\textbf{#2}\\[-0.15em]
#3
#4
\end{minipage}}
\newcommand{\runpic}[3]{%
\begin{tikzpicture}[
  >=Latex, line join=round, baseline=(base),
  x=0.34cm, y=0.50cm,
  field/.style={draw=wsppfield!80!black,circle,fill=wsppfield!16,
    minimum size=11pt,inner sep=0pt,font=\scriptsize},
  gate/.style={draw,diamond,minimum size=6.2pt,inner sep=0pt,fill=black!7},
  leaf/.style={font=\scriptsize,inner sep=0.4pt},
  elabel/.style={font=\tiny,inner sep=0.4pt,fill=white}
]
\coordinate (base) at (0,0.05);
\node[field] (f) at (0,3.6) {$\fielda$};
\node[gate] (b1) at (#1,2.2) {};
\node[gate] (b2) at (#2,2.2) {};
\draw[->,draw=wsppbranch!85!black] (f) -- node[elabel,left] {$1$} (b1);
\draw[->,draw=wsppbranch!85!black] (f) -- node[elabel] {$2$} (b2);
#3
\end{tikzpicture}}
\newcommand{\runfactor}[3]{%
\vspace{0.05em}\hrule height 0.25pt\vspace{0.08em}
\sbox{\runfactorbox}{{\fontsize{6.6pt}{7.2pt}\selectfont
  \textcolor{#1}{\texttt{#2}}}}%
\ifdim\wd\runfactorbox>\linewidth
  \resizebox{\linewidth}{!}{\usebox{\runfactorbox}}%
\else
  \usebox{\runfactorbox}%
\fi\\[-0.18em]
\sbox{\runfactorbox}{$#3$}%
\ifdim\wd\runfactorbox>\linewidth
  \resizebox{\linewidth}{!}{\usebox{\runfactorbox}}%
\else
  \usebox{\runfactorbox}%
\fi}
\newcommand{\dfcardzero}{\runpic{-1.9}{-0.45}{%
\node[gate] (oth) at (1.05,2.2) {};
  \draw[->,dashed] (f) to[bend left=8] node[elabel,right] {$*$} (oth);
\node[leaf] (s1) at (-3.3,0.55) {$\rulebox{\nicefrac12}$};
\node[leaf] (s2) at (-1.5,0.55) {$\rulebox{1}$};
\node[leaf] (s3) at (-0.7,0.55) {$\rulebox{\nicefrac12}$};
\node[leaf] (s4) at (0.75,0.55) {$\rulebox{\nicefrac15}$};
\node[leaf] (s5) at (1.95,0.55) {$\rulebox{\nicefrac14}$};
\node[leaf] (s6) at (3.25,0.55) {$\rulebox{\nicefrac13}$};
  \draw[->,draw=wsppbranch!85!black] (b1) -- node[elabel,left] {$1$} (s1);
  \draw[->,draw=wsppbranch!85!black] (b1) -- node[elabel,right] {$4$} (s2);
  \draw[->,draw=wsppbranch!85!black] (b2) -- node[elabel,left] {$1$} (s3);
  \draw[->,draw=wsppmissing!90!black] (oth) to[bend right=10] node[elabel,left] {$1$} (s4);
  \draw[->,draw=wsppmissing!90!black] (oth) -- node[elabel] {$3$} (s5);
  \draw[->,dashed,draw=wsppdefault!90!black] (oth) to[bend left=14] node[elabel,right] {$*$} (s6);
}}
\newcommand{\dfcardone}{\runpic{-1.9}{-0.45}{%
\node[gate] (oth) at (1.05,2.2) {};
  \draw[->,dashed] (f) to[bend left=8] node[elabel,right] {$*$} (oth);
\node[leaf] (s1) at (-3.0,0.55) {$\rulebox{\nicefrac12}$};
\node[leaf] (s2) at (-1.5,0.55) {$\rulebox{\nicefrac32}$};
\node[leaf] (s3) at (-0.45,0.55) {$\rulebox{\nicefrac12}$};
\node[leaf] (s4) at (0.65,0.55) {$\rulebox{\nicefrac15}$};
\node[leaf] (s5) at (2.1,0.55) {$\rulebox{\nicefrac14}$};
  \draw[->,draw=wsppbranch!85!black] (b1) -- node[elabel,left] {$1$} (s1);
  \draw[->,draw=wsppbranch!85!black] (b1) -- node[elabel,right] {$4$} (s2);
  \draw[->,draw=wsppbranch!85!black] (b2) -- node[elabel,left] {$1$} (s3);
  \draw[->,draw=wsppmissing!90!black] (oth) -- node[elabel,left] {$1$} (s4);
  \draw[->,draw=wsppmissing!90!black] (oth) -- node[elabel,right] {$3$} (s5);
}}
\newcommand{\dfcardtwo}{\runpic{-2.5}{-1.2}{%
\node[gate] (oth) at (1.05,2.2) {};
  \draw[->,dashed] (f) to[bend left=8] node[elabel,right] {$*$} (oth);
\node[leaf] (s2) at (-2.5,0.55) {$\rulebox{\nicefrac32}$};
\node[leaf] (s3) at (-1.2,0.55) {$\rulebox{1}$};
\node[leaf] (s4) at (0.0,0.55) {$\rulebox{\nicefrac25}$};
\node[leaf] (s5) at (1.9,0.55) {$\rulebox{\nicefrac14}$};
  \draw[->,draw=wsppbranch!85!black] (b1) -- node[elabel,right] {$4$} (s2);
  \draw[->,draw=wsppbranch!85!black] (b2) -- node[elabel,left] {$1$} (s3);
  \draw[->,draw=wsppmissing!90!black] (oth) -- node[elabel,left] {$1$} (s4);
  \draw[->,draw=wsppmissing!90!black] (oth) -- node[elabel,right] {$3$} (s5);
}}
\newcommand{\dfcardthree}{\runpic{-3.5}{-1.5}{%
\node[gate] (oth) at (0.5,2.2) {};
  \draw[->,dashed] (f) to[bend left=8] node[elabel,right] {$*$} (oth);
\node[leaf] (s3) at (-3.1,0.55) {$\rulebox{1}$};
\node[leaf] (s7) at (-1.5,0.55) {$\rulebox{\nicefrac32}$};
\node[leaf] (s4) at (-0.2,0.55) {$\rulebox{\nicefrac25}$};
\node[leaf] (s5) at (1.05,0.55) {$\rulebox{\nicefrac14}$};
\node[leaf] (s8) at (2.35,0.55) {$\rulebox{\nicefrac35}$};
  \draw[->,draw=wsppbranch!85!black] (b2) -- node[elabel,left] {$1$} (s3);
  \draw[->,draw=wsppbranch!85!black] (b2) -- node[elabel] {$4$} (s7);
  \draw[->,draw=wsppmissing!90!black] (oth) to[bend right=10] node[elabel,left] {$1$} (s4);
  \draw[->,draw=wsppmissing!90!black] (oth) -- node[elabel] {$3$} (s5);
  \draw[->,draw=wsppmissing!90!black] (oth) to[bend left=12] node[elabel,right] {$4$} (s8);
}}
\newcommand{\dfcardfour}{\runpic{-2.8}{-0.8}{%
\node[gate] (oth) at (0.5,2.2) {};
  \draw[->,dashed] (f) to[bend left=8] node[elabel,right] {$*$} (oth);
\node[leaf] (s7) at (-1.5,0.55) {$\rulebox{\nicefrac32}$};
\node[leaf] (s4) at (-0.2,0.55) {$\rulebox{\nicefrac25}$};
\node[leaf] (s5) at (1.05,0.55) {$\rulebox{\nicefrac14}$};
\node[leaf] (s8) at (2.35,0.55) {$\rulebox{\nicefrac35}$};
  \draw[->,draw=wsppbranch!85!black] (b2) -- node[elabel,left] {$4$} (s7);
  \draw[->,draw=wsppmissing!90!black] (oth) to[bend right=10] node[elabel,left] {$1$} (s4);
  \draw[->,draw=wsppmissing!90!black] (oth) -- node[elabel] {$3$} (s5);
  \draw[->,draw=wsppmissing!90!black] (oth) to[bend left=12] node[elabel,right] {$4$} (s8);
}}
\newcommand{\dfcardfive}{\runpic{-2.8}{-1.6}{%
\node[gate] (oth) at (0.5,2.2) {};
  \draw[->,dashed] (f) to[bend left=8] node[elabel,right] {$*$} (oth);
\node[leaf] (s4) at (-0.9,0.55) {$\rulebox{\nicefrac25}$};
\node[leaf] (s5) at (0.5,0.55) {$\rulebox{\nicefrac14}$};
\node[leaf] (s8) at (2.0,0.55) {$\rulebox{\nicefrac35}$};
  \draw[->,draw=wsppmissing!90!black] (oth) to[bend right=12] node[elabel,left] {$1$} (s4);
  \draw[->,draw=wsppmissing!90!black] (oth) -- node[elabel] {$3$} (s5);
  \draw[->,draw=wsppmissing!90!black] (oth) to[bend left=12] node[elabel,right] {$4$} (s8);
}}
\newcommand{\dfcardsix}{\runpic{-2.8}{-1.6}{%
\node[gate] (oth) at (0.5,2.2) {};
  \draw[->,dashed] (f) to[bend left=8] node[elabel,right] {$*$} (oth);
\node[leaf] (s5) at (-0.3,0.55) {$\rulebox{\nicefrac14}$};
\node[leaf] (s8) at (1.8,0.55) {$\rulebox{\nicefrac35}$};
  \draw[->,draw=wsppmissing!90!black] (oth) -- node[elabel,left] {$3$} (s5);
  \draw[->,draw=wsppmissing!90!black] (oth) -- node[elabel,right] {$4$} (s8);
}}
\newcommand{\dfcardseven}{\runpic{-2.8}{-1.6}{%
\node[gate] (oth) at (0.5,2.2) {};
  \draw[->,dashed] (f) to[bend left=8] node[elabel,right] {$*$} (oth);
\node[leaf] (s8) at (0.5,0.55) {$\rulebox{\nicefrac45}$};
  \draw[->,draw=wsppmissing!90!black] (oth) -- node[elabel,left] {$4$} (s8);
}}
\newcommand{\dfcardeight}{\runpic{-2.8}{-1.6}{}}
\begin{tabular}{@{}c@{\hspace{0.35em}}c@{\hspace{0.35em}}c@{\hspace{0.35em}}c@{}}
\runcard{0.235\textwidth}{$\rho = \rho_0$}{\dfcardzero}
  {\runfactor{wsppdefault}{Pass I: remove $d$}
  {\psi_1=\hat{\mathbb{1}}\wsppadd
    \strhat{\bar{t}_1\bar{t}_2\bar{t}_3}\wsppmul\rulebox{\nicefrac12}}}
&
\runcard{0.235\textwidth}{$\rho_1$}{\dfcardone}
  {\runfactor{wsppbranch}{Pass II: remove $(1,1)$}
  {\psi_2=\hat{\mathbb{1}}\wsppadd
    \strhat{t_1}\wsppmul\strhat{a_1}\wsppmul\rulebox{1}}}
&
\runcard{0.235\textwidth}{$\rho_2$}{\dfcardtwo}
  {\runfactor{wsppbranch}{Pass II: remove $(1,4)$}
  {\psi_3=\hat{\mathbb{1}}\wsppadd
    \strhat{t_1}\wsppmul\strhat{a_4}\wsppmul\rulebox{\nicefrac32}}}
&
\runcard{0.235\textwidth}{$\rho_3$}{\dfcardthree}
  {\runfactor{wsppbranch}{Pass II: remove $(2,1)$}
  {\psi_4=\hat{\mathbb{1}}\wsppadd
    \strhat{t_2}\wsppmul\strhat{a_1}\wsppmul\rulebox{1}}}
\\[0.45em]
\multicolumn{4}{c}{%
\begin{tabular}{@{}c@{\hspace{0.35em}}c@{\hspace{0.35em}}c@{\hspace{0.35em}}c@{\hspace{0.35em}}c@{}}
\runcard{0.185\textwidth}{$\rho_4$}{\dfcardfour}
  {\runfactor{wsppbranch}{Pass II: remove $(2,4)$}
  {\psi_5=\hat{\mathbb{1}}\wsppadd
    \strhat{t_2}\wsppmul\strhat{a_4}\wsppmul\rulebox{\nicefrac32}}}
&
\runcard{0.185\textwidth}{$\rho_5$}{\dfcardfive}
  {\runfactor{wsppmissing}{Pass III: remove $m(1)$}
  {\phi_1=\hat{\mathbb{1}}\wsppadd
    \strhat{\bar{t}_1\bar{t}_2}\wsppmul\strhat{a_1}
    \wsppmul\rulebox{\nicefrac25}}}
&
\runcard{0.185\textwidth}{$\rho_6$}{\dfcardsix}
  {\runfactor{wsppmissing}{Pass III: remove $m(3)$}
  {\phi_2=\hat{\mathbb{1}}\wsppadd
    \strhat{\bar{t}_1\bar{t}_2}\wsppmul\strhat{a_3}
    \wsppmul\rulebox{\nicefrac13}}}
&
\runcard{0.185\textwidth}{$\rho_7$}{\dfcardseven}
  {\runfactor{wsppmissing}{Pass III: remove $m(4)$}
  {\phi_3=\hat{\mathbb{1}}\wsppadd
    \strhat{\bar{t}_1\bar{t}_2}\wsppmul\strhat{a_4}
    \wsppmul\rulebox{4}}}
&
\runcard{0.185\textwidth}{$\rho_8$}{\dfcardeight}
  {\runfactor{black!60}{halt}
  {\rho^{\wsppstar}=\psi_1\wsppmul\cdots\wsppmul\psi_5
    \wsppmul\phi_3\wsppmul\phi_2\wsppmul\phi_1}}
\end{tabular}}
\end{tabular}
\endgroup
\caption{\footnotesize A complete run of Fig. \ref{alg:df-star} on the running
example $\rho$ (Ex. \ref{example:df-wspp-star-cycles}). Each card shows the
working node before the step named below the rule; atoms are as in
Ex. \ref{example:df-wspp-star-step}, and wide hats over products denote the
corresponding $\wsppmul$-products of atomic wSPPs. Since the default identity 
$\rulebox{\nicefrac{1}{3}}$ has been set to $\rulebox{0}$ by Pass~I, 
its branch is not shown in all subsequent wSPPs for brevity.  Stars of leaf weights computed in the
run are $\nicefrac13\cdot\cstar{(\nicefrac13)}=\nicefrac12$,
$\nicefrac12\cdot\cstar{(\nicefrac12)}=1$,
$\nicefrac14\cdot\cstar{(\nicefrac14)}=\nicefrac13$, and
$\nicefrac45\cdot\cstar{(\nicefrac45)}=4$.}
\Description{Nine cards showing the working wSPP after each step of the star
algorithm, with the emitted factor listed beneath each card.}
\label{fig:df-star-run-example}
\end{figure}
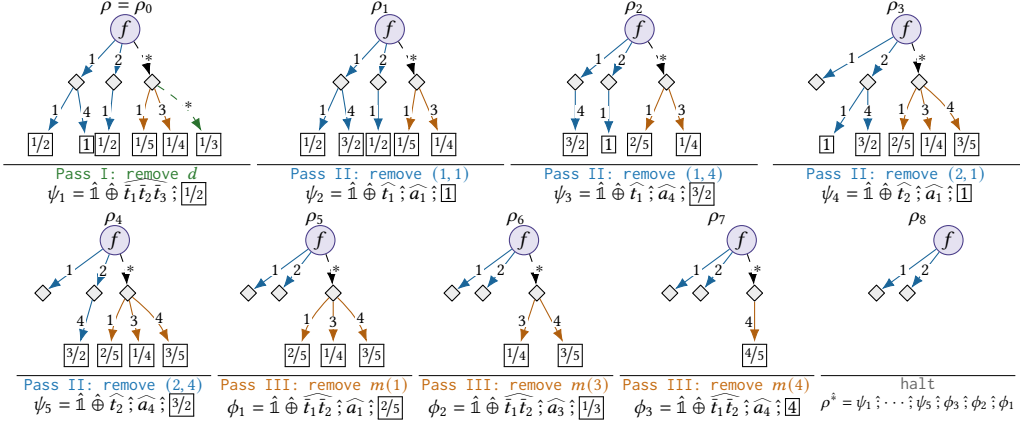

%% file: pareto.tex

\newcommand{\paretoa}{\mathfrak{p}}
\newcommand{\paretob}{\mathfrak{q}}
\Cref{sec:pareto-semirings} recaps \emph{Pareto semirings} and their finite
\emph{frontier} representations, and fits it to the embedding discipline of
\Cref{def:embedding}. In \Cref{sec:pareto-traces}, we present our novel
\emph{trace-carrying} Pareto semirings, which pair each Pareto-optimal
trade-off with a witnessing path, so that the semantics reports not just
\emph{which} trade-offs a policy realizes but \emph{along which route(s)}.

\subsection{Pareto Semirings and Their Frontiers}\label{sec:pareto-semirings}

Throughout, fix a partially ordered monoid
$T = (T, \cdot, \semione_T, \preceq)$ with multiplication monotone in each
argument, oriented so that \emph{bigger is better}. In the applications,
$T$ is a finite product of weight carriers with componentwise
multiplication and order; our running example pairs the tropical semiring
$\Trop = (\N\cup\{\infty\},\min,+,\infty,0)$ for latency with the
bottleneck semiring
$\Btl = (\N\cup\{\pm\infty\},\max,\min,-\infty,\infty)$ for bandwidth,
the tropical factor under the reversed numeric order (cheaper is bigger),
so that $\paretoa \preceq \paretob$ is \emph{Pareto
dominance}: $\paretob$ is no slower and no narrower. Writing
$\dcl X \triangleq \{\paretob \mid \exists \paretoa \in X.\; \paretob \preceq \paretoa\}$
for downward closure, sets of vectors become weights as follows:

\begin{definition}[\cite{samborskii1992fourier}]\label{def:pareto-semiring}
  The \emph{Pareto semiring} $\Par(T)$  has as
  elements the lower sets of $T$ (the downward closed subsets, $A = \dcl A$), with
  \[
    \semizero \triangleq \emptyset, \qquad
    \semione \triangleq \dcl\{\semione_T\}, \qquad
    A + B \triangleq A \cup B, \qquad
    A \cdot B \triangleq \dcl\{\,p \cdot q \mid p \in A,\, q \in B\,\},
  \]
  ordered
   by set inclusion, i.e., $A \preceq B$ iff $A \subseteq B$.
\end{definition}

A weight is thus the set of all vectors that some behavior dominates: two
parallel links with latency--bandwidth vectors $(5,2)$ and $(1,1)$ enter
as $\dcl\{(5,2)\}$ and $\dcl\{(1,1)\}$, and the maximal points of their
sum are exactly the two achievable trade-offs. $\Par(T)$ is an $\omega$-continuous semiring \cite{HandbookCh1},
hence a legitimate semantic semiring for \Cref{sec:preliminaries}, but not
one to compute with: $\dcl\{(5,2),(1,1)\}$ is an \emph{infinite} subset of
$\N^2$. What \emph{is} finite about that weight is its \emph{frontier} of
maximal points, which enables leveraging the embedding discipline of \Cref{def:embedding}: Call $F \subseteq T$ an \emph{antichain} if
its elements are pairwise incomparable, and write
$\fmax F \triangleq \{\paretoa \in F \mid \text{no } \paretob \in F \text{ has } \paretoa \prec \paretob\}$
for the maximal elements of a \emph{finite} $F \subseteq T$: a finite
antichain with $\dcl{(\fmax F)} = \dcl F$.

\begin{definition}[\cite{geilen2007algebra, gondran2008graphs}]\label{def:frontier-semiring}
  The \emph{frontier semiring} $\Fr(T)$ consists of the finite
  antichains of $T$, with
  \[
    \semizero \triangleq \emptyset, \quad
    \semione \triangleq \{\semione_T\}, \quad
    F + G \triangleq \fmax(F \cup G), \quad
    F \cdot G \triangleq \fmax\{\,\paretoa\cdot \paretob \mid \paretoa\in F,\, \paretob\in G\,\},
  \]
  ordered by domination
  ($F \preceq G \iff \forall \paretoa \in F.\ \exists \paretob\in G.\ \paretoa \preceq \paretob$),
  and with star $\cstar{F} \triangleq G_N$ for the first $N$ with
  $G_{N+1} = G_N$, where $G_0 \triangleq \semione$ and
  $G_{n+1} \triangleq \semione + F\cdot G_n$ (undefined if the iteration
  never stabilizes).
\end{definition}

The key result on frontier semirings is the following theorem.
\begin{theorem}\label{theorem:pareto-embedding}
  Let $T$ be a computable partially ordered monoid with monotone
  multiplication such that the star of every finitely generated lower set
  is finitely generated such as any product of bounded factors. Then
  $\Fr(T)$ is a computable star semiring and
  $\dcl{(-)} \colon \Fr(T) \hookrightarrow \Par(T)$ is an embedding.
\end{theorem}

Consequently (\Cref{theorem:wspp-main}), wSPPs over $\Fr(T)$ compute the
full Pareto semantics of policies.

\subsection{Trace-Carrying Pareto Semirings}\label{sec:pareto-traces}
Being able to compute frontiers raises the question: \emph{which traces realize a given frontier point?} In this section, we show that \emph{the semiring itself} can carry the traces, and wSPPs compute them.

Fix a finite, nonempty alphabet $\Sigma$ of \emph{events} (link or hop
identifiers). A
\emph{weighted trace} is a pair $(\paretoa, u) \in T \times \Sigma^*$: a cost
vector together with the trace of the run that produced it. Weighted traces
multiply by
$(\paretoa,u)\cdot(\paretob,v) \triangleq (\paretoa \cdot \paretob,\, uv)$, with unit
$(\semione_T, \varepsilon)$. Events enter
through the weights of a policy: annotating a link with the weight
$\{((5,2), \lett{a})\}$ both charges the cost and logs the event, so the
weighting operator $\odot$ of \wnetkat doubles as the logging primitive.
As a running example, consider a network in which two links $\lett{a}$
and $\lett{b}$ lead from switch $s$ to a hub $h$ carrying a costless
monitoring self-loop $\lett{c}$, and two links $\lett{d}$ and $\lett{e}$
exit to $t$---so $\Sigma = \{\lett{a}, \ldots, \lett{e}\}$. As a policy:
\[
\begin{aligned}
  p \;\triangleq\;
  &\bigl(\SEQ{\EQ{\mathit{sw}}{s}}{\WEIGH{w_{\lett{a}}}{\ASSN{\mathit{sw}}{h}}}
    \;\oplus\;
    \SEQ{\EQ{\mathit{sw}}{s}}{\WEIGH{w_{\lett{b}}}{\ASSN{\mathit{sw}}{h}}}\bigr)
  \SEQN
  \ITER{\bigl(\SEQ{\EQ{\mathit{sw}}{h}}{\WEIGH{w_{\lett{c}}}{\ASSN{\mathit{sw}}{h}}}\bigr)}\\
  &\SEQN
  \bigl(\SEQ{\EQ{\mathit{sw}}{h}}{\WEIGH{w_{\lett{d}}}{\ASSN{\mathit{sw}}{t}}}
    \;\oplus\;
    \SEQ{\EQ{\mathit{sw}}{h}}{\WEIGH{w_{\lett{e}}}{\ASSN{\mathit{sw}}{t}}}\bigr)
\end{aligned}
\]
with weights pairing each link's latency--bandwidth vector with its name:
$w_{\lett{a}} = \{((5,2),\lett{a})\}$ and
$w_{\lett{b}} = \{((1,1),\lett{b})\}$ for a slow, wide and a fast, narrow
$s$-to-$h$ link,
$w_{\lett{d}} = \{((2,3),\lett{d})\}$ and
$w_{\lett{e}} = \{((2,5),\lett{e})\}$ for the equal-latency exits, and
$w_{\lett{c}} = \{((0,\infty),\lett{c})\} = \{(\semione_T,\lett{c})\}$ a
pure logging step.

It remains to order weighted traces. The obvious choice is $(\paretoa,u) \preceq (\paretob,v)$ iff $\paretoa \preceq \paretob$ and
$u = v$. This keeps \emph{all} traces at the cost of finiteness, as the unrollings of
the loop yield entries $((5,2),\lett{a}\lett{c}^k)$ with pairwise distinct
traces, and finite frontiers are lost as in
\Cref{example:pareto-arctic}---this time because the
traces never stop to grow. We fix this by making the order on traces
deliberately lossy. Say that $v$ is \emph{a subsequence of} $u$ if $u$ can
be obtained from $v$ by inserting events at arbitrary positions; for
instance, $\lett{ac}$ is a subsequence of $\lett{abc}$. Now consider the following:

\begin{definition}[Trace order]\label{def:trace-order}
  For all $u, v \in \Sigma^*$, define the \emph{trace order} by
  \[
    u \trleq v \quad\Longleftrightarrow\quad
    (u = \varepsilon \wedge v = \varepsilon)
    \;\vee\;
    (v \neq \varepsilon \wedge v \text{ is a subsequence of } u)~,
  \]
  and order weighted words componentwise:
  $(\paretoa,u) \preceq (\paretob,v) ~\text{iff}~ \paretoa \preceq \paretob \wedge u \trleq v$.
\end{definition}

Two facts about this order are crucial. First, descending from $(\paretob,v)$
means worsening the cost vector \emph{and inserting arbitrary extra events
into $v$}, i.e., a longer trace are dominated,
and maximal elements---the ones frontiers keep---carry the
\emph{shortest} trace. Second, the empty word is \emph{isolated}:
$u \trleq \varepsilon$ and $\varepsilon \trleq v$ both force the other
word to be empty, i.e., a run that logged nothing and a run that logged
$\lett{c}$ are \emph{never comparable}. As $\trleq$ is a partial order with
monotone concatenation, $T \times \Sigma^*$ is again
a partially ordered monoid and $\Par(T\times\Sigma^*)$ is an
$\omega$-continuous semiring whose maximal points pair
a Pareto-optimal cost vector with a trace. We call it the
\emph{trace-carrying Pareto semiring}; its computable side is the semiring
$\Fr(T\times\Sigma^*)$ of \emph{trace frontiers}. With this construction, we get:

\begin{theorem}\label{theorem:trace-absorption}
  Let the monoid $T$ be bounded (that is, if its unit is the greatest element of its order \Cref{def:bounded}). Then every
  $L \in \Par(T\times\Sigma^*)$ satisfies $L \cdot L \preceq \semione + L$,
  and consequently $L^* = \semione + L$.
\end{theorem}

This is \emph{not} an instance of the bounded star
(\cite{mohri2002semiring}, cf \Cref{lemma:bounded-star}): isolating $\varepsilon$
makes $\Par(T\times\Sigma^*)$ \emph{unbounded} even for bounded $T$, as
the unit $(\semione_T,\varepsilon)$ dominates no weighted trace with a
nonempty trace. For bounded semiring, we would have $L^* = \semione$, \emph{erasing every witnessing trace}. The
absorption $L\cdot L \preceq \semione + L$ instead stops one step short of
the unit: the surviving summand $L$ \emph{is} the trace, one traversal
retained. This way, we keep traces while still collapsing the infinite star to a finite frontier.

To see that \Cref{theorem:trace-absorption}, i.e., $L\cdot L \preceq \semione + L$ and $L^* = \semione + L$, holds consider a generator
$(\paretoa,u)\cdot(\paretob,v) = (\paretoa\cdot \paretob,\, uv)$ of $L \cdot L$ with
$(\paretoa,u),(\paretob,v) \in L$, and split on the traces:
\[
  \begin{array}{@{}l@{\qquad}l@{\quad}l@{}}
    u \neq \varepsilon\colon
      & (\paretoa\cdot \paretob,\, uv) \preceq (\paretoa, u) \in L
      & \text{\footnotesize($q \preceq \semione_T$ absorbs the cost; $u$ is a subsequence of $uv$)}\\
    u = \varepsilon,\ v \neq \varepsilon\colon
      & (\paretoa\cdot \paretob,\, v) \preceq (\paretob, v) \in L
      & \text{\footnotesize(symmetrically)}\\
    u = v = \varepsilon\colon
      & (\paretoa\cdot \paretob,\, \varepsilon) \preceq (\semione_T, \varepsilon) \in \semione
      & \text{\footnotesize(both costs below $\semione_T$),}
  \end{array}
\]
where the three cases need three \emph{different} dominators
because $\varepsilon$ is isolated. From $L\cdot L \preceq \semione + L$, induction gives
$L^n \preceq \semione + L$ for all $n$, so the infinite sum
$L^* = \sum_n L^n$ collapses onto $\semione + L$.

$L^*$ thus keeps the empty trace and \emph{one traversal}: a second traversal
worsens the cost vector (boundedness of $T$) while merely inserting events
into the trace (lossiness of $\trleq$), and is therefore dominated by the
first. It is in this sense in which trace-carrying is \emph{tractable} over
a bounded \emph{cost} monoid $T$. On the computable side, the star
iteration of \Cref{def:frontier-semiring} stabilizes after one step at
$\cstar{F} = \semione + F$, without inspecting its argument's traces:

\begin{theorem}\label{theorem:trace-embedding}
  Let $T$ be a computable, bounded, partially ordered monoid with monotone
  multiplication, and let $\Sigma$ be finite and nonempty. Then
  $\Fr(T\times\Sigma^*)$ is a computable star semiring, and
  $\dcl{(-)} \colon \Fr(T\times\Sigma^*) \hookrightarrow
  \Par(T\times\Sigma^*)$ is an embedding.
\end{theorem}

\begin{example}\label{example:trace-running}
  By \Cref{theorem:wspp-main}, wSPPs over trace frontiers thus compute the
  trace-carrying Pareto semantics; let us compile the policy $p$ above to
  a wSPP over
  $\Fr((\Trop\times\Btl)\times\Sigma^*)$. Semantically, a packet reaches
  $t$ along infinitely many runs: take $\lett{a}$ or $\lett{b}$, loop $k$
  times, exit over $\lett{d}$ or $\lett{e}$, so its trace-carrying weight
  is the infinite lower set
  $\dcl\{((7,2), \lett{a}\lett{c}^k x),\,
         ((3,1), \lett{b}\lett{c}^k x) \mid
         k \in \N,\, x \in \{\lett{d},\lett{e}\}\}$,
  whose frontier is nevertheless finite. The computation over trace
  frontiers finds it: the star of the loop weight is
  $\cstar{w_{\lett{c}}} = \semione + w_{\lett{c}} =
  \{(\semione_T,\varepsilon),\,(\semione_T,\lett{c})\}$---where
  $\varepsilon$-isolation keeps \emph{both} the ``never entered the loop''
  entry and the ``looped once'' entry---the subsequent products discard
  every unrolling ($k \geq 1$) as dominated by its loop-free core (e.g.,
  $((5,2),\lett{a}\lett{c}) \preceq ((5,2),\lett{a})$. The
  $s$-to-$t$ branch of the wSPP carries the leaf
  \[
    \{\,((7,2),\, \lett{a}\lett{d}),\;\; ((7,2),\, \lett{a}\lett{e}),\;\;
        ((3,1),\, \lett{b}\lett{d}),\;\; ((3,1),\, \lett{b}\lett{e})\,\}~.
  \]
  Note that each trade-off comes with \emph{two} traces: $(3,1)$ is
  realized by both $\lett{b}\lett{d}$ and $\lett{b}\lett{e}$. The runs
  cost the same and neither trace is a subsequence of the other, so both are
  maximal. \hfill $\triangle$
\end{example}

\smallskip\noindent\textbf{Reading traces.}
A frontier entry $(\paretoa,u)$ asserts that some run achieves the cost
vector $\paretoa$ and that $u$ one of its traces \emph{up to absorbed detours}. Runs
whose traces merely insert extra events into $u$, at no better cost, are
identified with it. It does \emph{not} enumerate all optimal traces (a
second route with the same cost whose log contains $u$ as a subsequence is
absorbed), does not count cycles, not an exhaustive log
of one specific packet. What \emph{is} guaranteed:
distinct frontier entries are incomparable, so runs with
subsequence-incomparable logs are never conflated---distinct optima are
all reported, and a single optimum keeps one witness per route that ties for
it (like $\lett{b}\lett{d}$ and $\lett{b}\lett{e}$ above).

%% file: implementation.tex

\input{optimization.tex}

\input{benchmarks.tex}

\input{casestudy.tex}

%% file: optimization.tex
While our verified Lean implementation of wSPPs is acceptable in terms performance,
the lack of sharing (hashconsing) and caching restricts its scalability when applied to 
large real-world networks. To push the performance of wSPPs further, we implement 
the \wspptool tool in Rust for constructing and querying wSPPs. It is important
to note that the wSPP construction algorithm of \wspptool is the same as our 
verified Lean version and all algorithmic optimizations have been formally verified. 
The only difference between the two versions is that we implement the aforementioned
sharing and caching memory optimizations in \wspptool which are unverified but entirely standard. 
We argue that it is fair for us to compare our memory optimized implementation
with prior works as they all implement equivalent sharing/caching optimizations of 
their own. 
\subsection{Optimizations}
We now give an overview of the verified algorithmic optimizations made 
to wSPPs.

\smallskip\textit{wSPP Canonicalization.} Prior work \cite{KATch} on SPPs describes
a canonicalization procedure such that two canonicalized SPPs
are semantically equivalent if and only if they are syntactically equal.
Although deciding semantic equivalence is not the focus of \wnetkat,
Although deciding semantic equivalence is not the focus of \wnetkat,
canonicalization remains valuable because it makes
hash-based sharing and caching \emph{semantically aware}: 
semantically equivalent policies are compiled into the same canonical wSPP representation,
hence sharing a common hash. Cached results are reused not only when an operation is repeated on
syntactically identical inputs, but whenever it is applied to semantically equivalent ones as well. 
We apply it as a smart constructor whenever a new wSPP node is built, 
giving us \Cref{theorem:wspp-canonical}.

\begin{theorem}\label{theorem:wspp-canonical}
  For \wnetkat policies $p$ and $q$, they are equivalent if and only if
  their canonical wSPP representations $\hat{p}$ and $\hat{q}$ are syntactically equal.
\end{theorem}
%




\smallskip\textit{Untested-field Decomposition. }
This optimization applies when a policy $p$ contains a field $g$ whose input
value is irrelevant: if packets $\alpha$ and $\alpha'$ differ only on $g$, then
$\sem{p}(\alpha)=\sem{p}(\alpha')$. In the wSPP
$\rho\triangleq\hat{p}$, this means that descendants rooted at $g$ only write
to $g$ and never test it, typically having the form
$\wspp(g,\varnothing,m,\rulebox{\semizero})$. In this case, the optimization splits 
$\rho$ into $\rho'$ and $\sigma$ so that
$\norm{\rho}\equiv\norm{\rho'}\oplus\norm{\sigma}$, where $\rho'$ collects the
branches that leave $g$ untouched and $\sigma$ the branches that assign $g$.
Using the identities in E1 and E2 of
\Cref{fig:preliminaries-dependency-map}, we obtain
\[
\begin{aligned}
  \norm{\rho}^{*}
  &\equiv
    \norm{\rho'}^{*}
    \oplus
    (\norm{\rho'}\oplus\norm{\sigma})^{*};
    \norm{\sigma};\norm{\rho'}^{*} 
\end{aligned}
\]
The crucial observation is that, since $\rho'$ neither tests nor assigns $g$, every
write to $g$ made inside the middle star passes through $\rho'$ unchanged and
is overwritten by the trailing $\sigma$. Thus, the identity above stil holds if 
every descendant wSPP 
$\wspp(g,\varnothing,m,\rulebox{\semizero})$ in the $\sigma$ inside the middle star is replaced by
$\bighoplus_{\scriptscriptstyle v\in \text{dom}(m)}m(v)$, thereby erasing the write to $g$ while preserving the
continuations. Thus the original star reduces to the stars of two wSPPs $\rho'$ and 
$\rho' \wsppadd \sigma$ that both 
omit the field $g$. This optimization targets
\textit{lexically scoped} local variables, which we encode as
$\textsf{let }x=v\textsf{ in }p \triangleq
(f_x\leftarrow v);p[f_x/x];(f_x\leftarrow v)$, making $f_x$ an
input-independent field in the resulting wSPP.

\smallskip\textit{Star with a Trailing Policy. } This optimization recognizes \textit{contextual stars}, i.e. \wnetkat expressions of form 
$p^*;q$ and appends the wSPP computed for $q$ to the list of factors of $p^*$ on its 
right. When $q$ is simple, this could greatly reduce the size of intermediate wSPPs 
while computing the sequential composition of factors in $p^*$. Moreover, if $\rho = \hat{p}$ contains an untested field $f$, there is
\[
\norm{\rho}^*; q \equiv  \norm{\rho'}^*;q \oplus (\norm{\rho'} \oplus \norm{\sigma})^*;\norm{\sigma}; \norm{\rho'}^*;q
\]
This means calculating $\rho^{\wsppstar}\wsppmul \hat{q}$ 
reduces to calculating two smaller contextual stars $(\rho')^{\wsppstar}\wsppmul \hat{q}$ and $(\rho' \wsppadd \sigma)^{\wsppstar}\wsppmul \sigma $, which 
further enhances the effect of untested field decomposition.

%% file: benchmarks.tex
\subsection{Scalability Benchmarks}
To demonstrate the practicality and scalability of wSPP construction, we compare
the performance of \wspptool and our verified Lean implementation against
state-of-the-art verification tools such as KATch~\cite{KATch}, KATch2~\cite{KATch2},
McNetKAT~\cite{smolka_scalable_2019} and the Storm~\cite{hensel_probabilistic_2020} probabilistic model checker.
We partition our scalability benchmarks into two sets: 
Boolean reachability and probabilistic reachability. This is because KATch and KATch2 only handle Boolean
reachability and conversely, McNetKAT and Storm are specialized for probabilistic 
workloads. On the other hand, due to \wnetkat being parameterized over semirings,
our wSPP construction algorithms are generic over semirings. This allows our
wSPP tools to work \emph{uniformly} on both benchmark sets.
Our experiments are all conducted using an Intel Xeon Gold 6348 CPU with
20GB memory- and 20min time-limit.

\smallskip
\noindent\textit{\textbf{Correctness and Cross-Validation}}
In our experiments we check that the tools agree on their answers.
For the boolean benchmark, all
four tools (\wspptool, Lean, KATch and KATch2) return
the same verdict on $275$ of the $276$ Topology Zoo networks. On the largest
network (Kdl, $754$ nodes) Lean times out, and the remaining
three engines agree. For the probabilistic benchmark the tools 
also agree up to high numerical precision on all runs that completed within the resource limits.

\smallskip
\noindent\textit{\textbf{Boolean Reachability}}
We compare \wspptool and our Lean implementation
against KATch(2) on the Topology Zoo~\cite{topology_zoo} benchmark.
KATch and KATch2 both implement the Symbolic Packet Program (SPP)~\cite{KATch} decision procedure for NetKAT. 
The primary difference between these two tools is that KATch is implemented in Scala whereas 
KATch2 is implemented in Rust with much more emphasis performance engineering.
The results of our experiments are presented in \Cref{fig:bool-reachability}.
\begin{figure}
    \begin{subfigure}{0.4\textwidth}
        \centering
        \includegraphics[width=\textwidth]{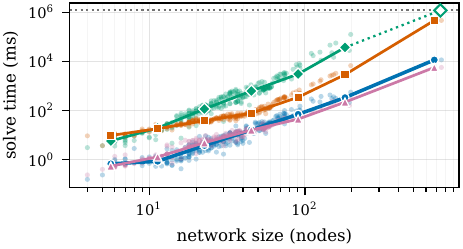}
    \end{subfigure}
    \hfill
    \begin{subfigure}{0.4\textwidth}
        \centering
        \includegraphics[width=\textwidth]{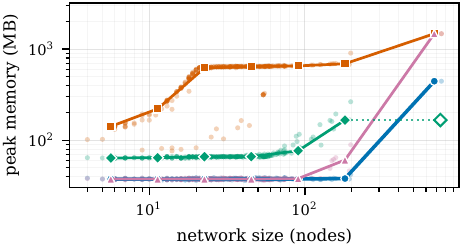}
    \end{subfigure}
          \begin{subfigure}{0.17\textwidth}
        \centering
        \includegraphics[width=\textwidth]{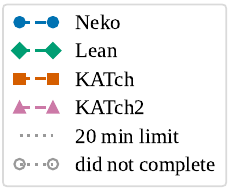}
        \vspace{.8em}
    \end{subfigure}
    \caption{Boolean reachability scalability: running time (left)
    and peak memory usage (right) .}
    \label{fig:bool-reachability}\vspace{-5pt}
\end{figure}
 \wspptool is very competitive even
against the highly optimized KATch2 both in terms of solve time and peak memory usage.
This is in spite of the fact that the star construction~(\Cref{alg:star}) of wSPP is much more 
involved than SPP which is a straightforward fixpoint iteration.
The verified Lean trails in speed because it forgoes the sharing and
caching optimizations employed by the other tools.
Where those tools resolve a repeated computation with a single cache lookup,
Lean must instead recompute the corresponding wSPP subtrees from scratch,
incurring substantial redundant work.

\begin{figure}
    \begin{subfigure}{0.4\textwidth}
        \centering
        \includegraphics[width=\textwidth]{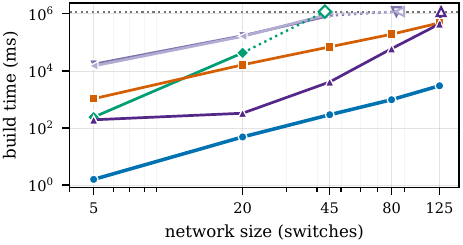}
    \end{subfigure}
    \hfill
    \begin{subfigure}{0.4\textwidth}
        \centering
        \includegraphics[width=\textwidth]{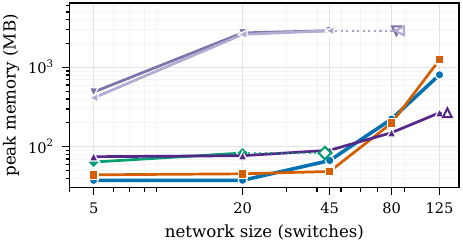}
    \end{subfigure}\hfill
      \begin{subfigure}{0.17\textwidth}
        \centering
        \includegraphics[width=\textwidth]{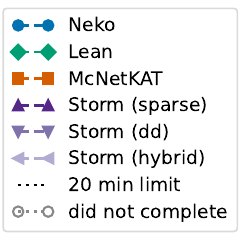}
        \vspace{.8em}
    \end{subfigure}
    \caption{Probabilistic reachability scalability: running time (left)
    and peak memory usage (right).}
   \vspace{-5pt}
    \label{fig:real-reachability}
\end{figure}

\smallskip
\noindent\textit{\textbf{Probabilistic Reachability}}
McNetKAT extends Forwarding Decision Diagrams (FDD)~\cite{smolka_fast_2015} to reason efficiently about the
dup-free fragment of ProbNetKAT~\cite{prob_netkat_2016}. In their original evaluation, Smolka
et al. demonstrated that it performed favorably against Prism~\cite{kwiatkowska_prism_2011}, the de facto
standard probabilistic model checker. Prism has since been largely superseded by 
Storm~\cite{hensel_probabilistic_2020}, a state-of-the-art model checker that offers substantially improved performance. 
We adapt the McNetKAT benchmarking harness to target Storm in place of Prism. 
We evaluate our wSPP tools, instantiated with the
Real semiring, against both baselines on the Bayonet~\cite{bayonet} and F10 AB FatTree case
studies used in the evaluation of McNetKAT. 
\Cref{fig:real-reachability} presents the end-to-end time and memory each engine uses
to answer reachability queries. 
Construction of the core data structure
(wSPP for \wspptool and Lean, FDD for McNetKAT, DTMC for Storm) 
dominates the running time in every case. 

From \Cref{fig:real-reachability} we can see that the efficiency of Storm makes it
compare more favorably to McNetKAT than Prism on small to medium sized networks.
However, on large networks with 125 switches, the state space explosion incurred by
encoding ProbNetKAT expressions in Storm overwhelms it 
(confirmed on Storm's dd, sparse and hybrid engines). 
Compared to prior works, \wspptool is orders-of-magnitudes
faster while using comparable or less memory. 

\smallskip
\noindent\textit{\textbf{Ablation Experiments}}
We conduct ablation experiments to measure each optimization's contribution and the 
impact of witness traces for Pareto semirings.

\smallskip\textit{Optimizations} 
For studying the effects of optimization we use a leave-one-out ablation:
starting from the fully-optimized \wspptool, we disable one optimization at
a time and rebuild the F10 benchmark. 
We find that canonicalization is the dominant optimization. Removing it slows construction by
$4.8\times$ in geometric mean and it accounts for a $2.3\times$ difference in
peak memory. Without it the wSPP retains semantically redundant branches and the
six largest policies exhaust the $20$\,GB budget. The
star with trailing-policy optimization is the next largest contributor at $2.7\times$, followed by
untested-field decomposition at $1.9\times$. Removing the latter drives the same six
instances out of memory. Each ratio is measured only on policies both
configurations complete so the memory failures are excluded and the
canonicalization and untested-field numbers are conservative.

\smallskip\textit{Witness Traces} 
We also evaluate the impact of witness traces on Pareto semirings.
During our case study on Fat Tree vs Jellyfish topologies in \Cref{sec:case-study}, 
we observed that the witness traces only add a small overhead to the construction time 
(approximately 5\% on average). This is due to the fact that the overhead of maintaining
witness traces scales with the \emph{width} of Pareto frontiers and not directly with the
size of the network. We showcase this empirically in \Cref{sec:appendix-synthetic} through a series of
synthetic adversarial benchmarks that vary both the width of the Pareto frontier and the 
size of the network. In practice, the width of Pareto frontiers do not grow to be very large 
(average of 1.7 in our case study), and thus the overhead of witness traces is relatively small.

%% file: casestudy.tex
\subsection{Case Study: Fat Tree vs. Jellyfish Topologies}\label{sec:case-study}

We now present the full case study of topology tradeoff analysis that we previewed in \Cref{sec:overview}.

\subsubsection*{Background}
Fat-tree topologies have one primary drawback~\cite{singh2015clos}: Once all the ports of the central backbone switches have been connected, we cannot incrementally add new servers at the leaves of the tree. Rather, we must completely replace the switches in order to add capacity for new servers.
Expanding ``incrementally'' is therefore a nonstarter for this type of network.

To address this issue, Singla et al. introduced ``Jellyfish'' topologies: switches reserve some ports for hosts and use the remaining ports to form a random regular graph among switches, instead of being arranged into structured aggregation/core layers~\cite{singla2012jellyfish}.
These are the same Jellyfish topologies behind Amazon's shift away from Fat-trees that we noted in \Cref{sec:overview}.

\subsubsection*{Design of our Case Study}
Our case study focuses on the design phase of a data center network, when the topology and basic routing policies are selected.
To investigate the tradeoffs between Jellyfish and Fat-tree topologies, we begin by generating a collection of instances of each network layout.
In order to focus the analysis squarely on the \emph{topology} tradeoffs, we assume that the same hardware is available for each network.
Since the cost of a switch depends on the number of ports it has, we fix a common max degree of 10 for nodes across the experiment (resulting in 250 supported hosts).
Similarly, link latency and bandwidth costs are generated randomly, but a fixed collection of switches, links, and link weights is used for each topology in head-to-head comparisons under the principle of keeping the hardware constant.
For each topology, we generate three routing policies to evaluate: 4-shortest paths (i.e., $k$-shortest paths~\cite{yen1971kshortest} with $k=4$), 8-shortest paths, and Equal-cost Multipath (ECMP)~\cite{rfc2992ecmp}. These are the same three policies evaluated as data center routing schemes for Jellyfish topologies in~\cite{singla2012jellyfish}.
Because the Jellyfish topologies are random, we generated topologies using 50 different randomness seeds and took averages of the results.

\subsubsection*{Results and Insights}

Evaluating the nuance of the design tradeoffs requires several views of data. Our results are centered on three main figures:
\Cref{fig:jelly-scatter} provides a 3D scatterplot of the Pareto-optimal points, where a point's height is the proportion of host pairs whose Pareto frontier can match or outperform that point's latency and bandwidth. In \Cref{fig:jelly-bars} we give a direct comparison of the networks by considering the proportion of host pairs for which one frontier dominates the other, the two are equal or are Pareto incomparable. Finally, powered by a trace-carrying Pareto semiring (cf. \Cref{sec:pareto-traces}), \Cref{fig:jelly-hotspot} gives an overview of each network’s Pareto-optimal path distribution: it plots what percentage of the network’s switches ($y$) are needed to realize growing percentages of Pareto-optimal paths ($x$).

\smallskip\textit{Finding 1: Jellyfish provides better Pareto performance.}
\Cref{fig:jelly-scatter} shows that nearly all plotted latency/bandwidth frontier points, the highest points belong to the Jellyfish network with 8-shortest routing. This suggests that no matter the specific Pareto frontier point considered, a larger proportion of host pairs can achieve or outperform this point in this network than in other networks. Fat-tree with ECMP is the next strongest, in particular at the higher latency values. \Cref{fig:jelly-bars} supports this finding, with Jellyfish dominating more host pairs than Fat-tree networks.


\smallskip\textit{Finding 2: Pareto and scalar objectives favor different Jellyfish routing.}
While our results show a consistent favoring of Jellyfish networks, the preferred routing method within Jellyfish depends on the precise objective being optimized. In particular, the Pareto data of \Cref{fig:jelly-scatter} and \Cref{fig:jelly-bars} favors the use of 8-shortest routing in Jellyfish, consistent with \Citet{singla2012jellyfish}'s evaluation of 8-shortest routing.
However, when considering median and mean values of latency in isolation, i.e., ignoring bandwidth, scalar latency performance seems to be best in ECMP Jellyfish networks, rather than $k$-shortest ones~(\Cref{tab:jelly-summary}). This is because Pareto frontiers in $k$-shortest routing will often include extra routes that provide better bandwidth, but are not necessarily the shortest paths (thus higher latencies). ECMP routing rejects paths that have higher hop-counts, so they cannot increase median/mean latency of the network.

In practice, this means that although our analysis supports the use of Jellyfish over Fat-tree networks for better latency and bandwidth, the routing policies should be selected for the expected workloads. For workloads that need lower average latency, ECMP is likely to perform better. However, for workloads that need a balance of both low latency and high bandwidth simultaneously, 8-shortest routing covers more host pairs at any joint performance target.

\smallskip\textit{Finding 3: ECMP best distributes Pareto-optimal paths across switches.}
Finally, \Cref{fig:jelly-hotspot} shows that there is a stark difference in the switch utilization of Pareto-optimal paths between Jellyfish and Fat-tree networks. In particular, Fat-tree tend to have more concentrated path distribution over the network, especially when using $k$-shortest routing, where severe switch hotspots are created: as many as 80\% of Pareto-optimal paths are realized by fewer than 40\% of the switches. On the other hand, in Jellyfish this difference is not so striking, and they retain overall excellent path distribution. Surprisingly, ECMP routing provides the best avoidance of hotspots in both topologies. The better latency and bandwidth performance of Jellyfish with $k$-shortest routing is usually attributed to the better \emph{path diversity}. However, after removing equal-weight paths that are subsequence-redundant with more direct paths, the superior Pareto performance of $k$-shortest routing on Jellyfish cannot be explained by lower hotspot presence among the retained paths.

\newlength{\jellyfigheight}
\setlength{\jellyfigheight}{5.3in}
\begin{figure}
    \centering

    \begin{minipage}[t][\jellyfigheight][t]{0.56\columnwidth}
        \strut\par
        \centering
        \includegraphics[width=0.88\linewidth]{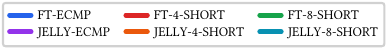}\\[2pt]

        \begin{subfigure}[t]{0.53\linewidth}
            \includegraphics[width=\linewidth]{./plots/3d-compact.pdf}
            \phantomsubcaption\label{fig:jelly-scatter}
        \end{subfigure}\hfill
        \begin{subfigure}[t]{0.41\linewidth}
            \includegraphics[width=\linewidth]{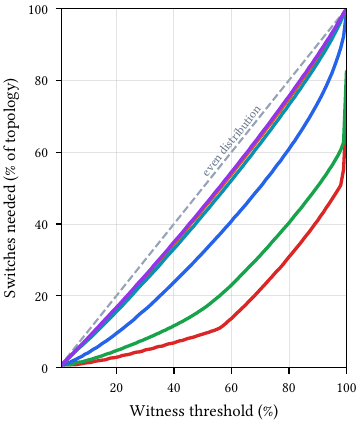}
            \phantomsubcaption\label{fig:jelly-hotspot}
        \end{subfigure}
    \end{minipage}\hfill
    \begin{minipage}[t][\jellyfigheight][t]{0.43\columnwidth}
        \strut\par
        \centering
        \begin{subfigure}[t]{\linewidth}
            \includegraphics[width=0.87\linewidth]{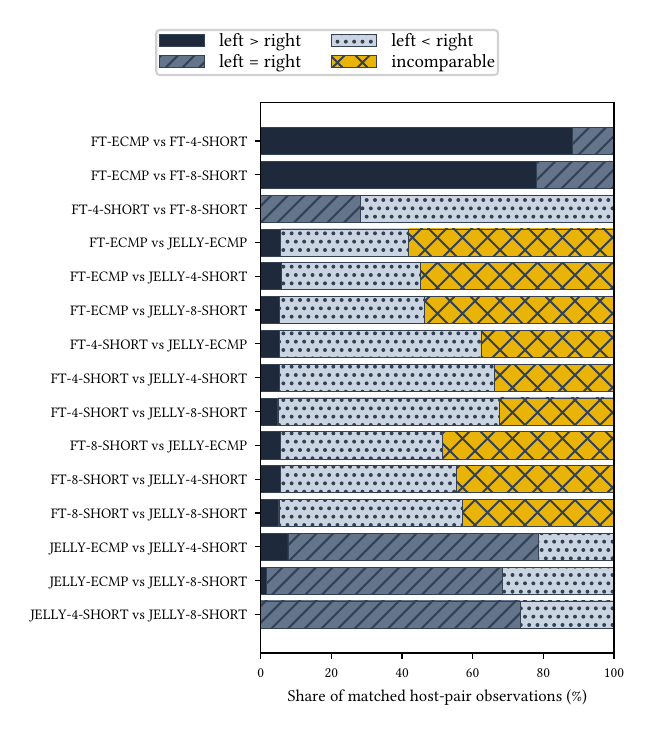}
            \phantomsubcaption\label{fig:jelly-bars}
        \end{subfigure}
    \end{minipage}
    \vspace{-7.2cm}
    \caption{Case study results: Fat tree vs. Jellyfish analysis.
    (a, left) Scatterplot of Pareto-optimal points across the analyzed networks.
    (b, middle) Hotspot concentration, showing how many switches are needed to realize a given share of Pareto-optimal paths.
    (c, right) Pointwise network-to-network comparison of Pareto frontiers.}
    \label{fig:jelly}
\end{figure}

\begin{table}
    \small
    \begin{tabular}{l r r r r r r r}
        \input{plots/summary.tex}
    \end{tabular}
    \caption{Summary of Jellyfish vs. Fat-tree comparison. Latencies are in ms, and bandwidths are in Mbps.}\vspace*{-1em}
    \label{tab:jelly-summary}
\end{table}

%% file: plots/summary.tex
Network & Frontier & Latency & Latency & Latency & Bandwidth & Bandwidth & Bandwidth \\
Topology + Policy & Size & Median & Mean & Mean-SD & Median & Mean & Mean-SD \\
\hline
FT-ECMP & 2.47 & 8.81 & 8.56 & 0.04 & 932.00 & 936.53 & 3.86 \\
FT-4-SHORT & 1.44 & 8.88 & 8.45 & 0.09 & 920.00 & 927.33 & 4.96 \\
FT-8-SHORT & 1.88 & 8.86 & 8.52 & 0.06 & 925.00 & 930.89 & 4.58 \\
JELLY-ECMP & 1.38 & 6.85 & 6.27 & 0.03 & 927.00 & 933.23 & 3.16 \\
JELLY-4-SHORT & 1.61 & 6.89 & 6.48 & 0.03 & 928.00 & 934.37 & 3.21 \\
JELLY-8-SHORT & 1.82 & 6.92 & 6.72 & 0.04 & 932.00 & 936.66 & 3.35\\[-.5em]

%% file: related_work.tex
In what follows, we discuss related works on (i) \netkat and other analysis tools, (ii) decision diagrams for quantitative analysis, and (iii) constructions on trace-carrying quantitative structures.

\smallskip
\noindent\textit{\textbf{\netkat and other analysis tools.}}
\netkat \cite{netkat} is a DSL for modeling packet forwarding and comes with a sound and complete
equational theory. Forwarding decision diagrams
(FDDs) \cite{smolka_fast_2015} and \katch
\cite{KATch} check equivalence symbolically, the latter via SPPs---the data structure which inspired wSPPs. \probnetkat \cite{prob_netkat_2016} extends \netkat by probabilistic modeling aspects. The model checker \mcnetkat has been developed to check \probnetkat's history-free fragment
at scale \cite{smolka_scalable_2019} by converting FDDs to sparse matrices and leveraging linear algebra solvers. \wnetkat \cite{wNetKAT}
subsumes that history-free fragment and, due to it being parametric in an $\omega$-continuous
semiring, handles a large variety of further quantitative aspects, including multi-objective analysis as established in this paper. To the best of our knowledge, ours is the first semiring-agnostic framework fully automatically analyzing multiobjective trade-offs in networks while providing (possibly multiple) shortest traces for Pareto-optimal points. Besides \netkat, other prominent network analysis tools include Header Space Analysis~\cite{header-space,netplumber}, VeriFlow~\cite{veriflow}, Atomic Predicates~\cite{atomic-predicates,ap-keep}, and various tools for analyzing P4 data planes~\cite{p4v,p4testgen,vera}. There are also tools for analyzing control planes such as Batfish~\cite{batfish}, MineSweeper~\cite{minesweeper}, Tiramisu~\cite{tiramisu}, and ARC~\cite{arc}. Finally, network calculus provides mathematical tools for analyzing quantitative properties, such as worst-case congestion bounds~\cite{boudec-thiran}.

\smallskip
\noindent\textit{\textbf{Decision Diagrams for Quantitative Analysis.}}
BDDs \cite{bryant1986graph} canonically represent Boolean functions. For quantitative analysis, they have been generalized to
multi-terminal BDDs \cite{fujita1997multi} and algebraic decision diagrams
\cite{bahar1997algebraic}, admitting numeric terminals. \emph{Semiring extensions} are edge-valued BDDs, which factor additive weights along edges
\cite{lai1992edge}, and semiring-labelled decision diagrams \cite{wilson2005decision,fargier2014knowledge}, compiling
valuations over \emph{commutative} semirings. All of these map
Boolean (or finite-domain) valuations to constants, with fixed points
iterated outside the diagram algebra. wSPPs instead branch on field tests, denote weighted packet \emph{relations}, draw leaves
from any embeddable (possibly \emph{noncommutative}) star semiring, and
compute the Kleene star \emph{directly on the diagrams.}

\smallskip
\noindent\textit{\textbf{Pareto Semirings and Traces.}}
Semirings of Pareto frontiers are classical in multicriteria path algebra
\cite{samborskii1992fourier,gondran2008graphs,geilen2007algebra} as well as soft constraints and multi-metric networking
\cite{bistarelli2008soft,larrosa2010semiring,somasundaram2011multi,vieira2023automating}, computed by multiobjective search
\cite{casas2021improved,casas2023targeted}. Witness-carrying weights are
also established, but each design forfeits an ingredient of
\Cref{sec:pareto-traces}, rendering them unsuitable for our needs: Derivation semirings \cite{goodman1999semiring}
and path-listing semirings \cite{manger2020algebraic} keep witnesses inert
or exact by pruning by score only, or not at all.
Absorptive provenance polynomials
\cite{green2007provenance,deutch2014circuits,dannert2021semiring} are antichains of dominance-pruned proof witnesses but monomials are commutative
multisets (yielding the event order to be lost), and absorption makes the semiring bounded ($\semione + a = \semione$) (which is the
collapse our $\varepsilon$-isolation avoids). The shortest-path provenance
semiring of \cite{ramusat2021provenance} multiplies cost--word pairs like our weighted traces, but its sum selects a \emph{single} pair
under a total order, and $\varepsilon$ being the
least word there forces $\cstar{a} = \semione$. Routing algebras carry
paths inside route weights, also under a shorter-is-better preference
\cite{griffin2005metarouting,daggitt2018asynchronous}. Preferences are
total, and there is no semiring product or star. To our knowledge,
$\Fr(T\times\Sigma^*)$ is the first semiring whose frontiers range over
\emph{cost--trace pairs} with a nontrivial order on the traces themselves, absorbing any run that inserts detour events at no
better cost by the order-construction in such a way that each Pareto-optimal cost vector retains one shortest. Enabling all this directly on a semiring-level enables leveraging wSPPs for obtaining traces without modifying them.

%% file: conclusion.tex
We introduced wSPPs, a symbolic data structure for computing the automatic, quantitative analysis of
\wnetkat policies for any embeddable $\omega$-continuous
semiring---including the Kleene star, directly on the diagrams---with
formally verified algorithms and a high-performance implementation. We also
developed trace-carrying Pareto semirings, whose frontiers report
every optimal trade-off together with at least one shortest witnessing
trace. Our evaluation demonstrates that a single semiring-generic engine is
competitive with the specialized state of the art on Boolean reachability
(KATch2), orders of magnitude faster than probabilistic baselines
(\mcnetkat, Storm), and turns the Fat-tree vs.\ Jellyfish design question
into the automatic analysis of Pareto frontiers.

\smallskip
\noindent\textit{\textbf{Future Work.}}
Our Pareto constructions require a bounded cost monoid applying to, e.g., latency,
bandwidth, and Viterbi-style best-run probabilities. \emph{Fully
probabilistic} multiobjective analysis, where expectations accumulate
over branching and achievable trade-offs form convex Pareto curves
realized by randomized schedulers, do, however, not form such a monoid. We therefore plan to extend wSPP weights to (possibly trace-carrying) convex frontiers arising, to enable mutli-objective probabilistic network analysis à la model
checking of Markov decision processes
\cite{etessami2007multiobjective,forejt2011quantitative}.

%% file: appendix-algorithms.tex

This appendix spells out the algorithms and proof obligations for the wSPP
operations that \Cref{sec:wspp-easy-ops} treats abstractly: the compilation of
predicates, weighting, choice, sequential composition, and the canonicalizing
smart constructor. All node constructions are routed through the smart
constructor $\mathsf{mk}$ of \Cref{sec:appendix-mk}; when a proof reasons
about a displayed raw node, \Cref{theorem:mk-soundness} justifies replacing it
by its canonicalized form. 
\allowdisplaybreaks

\vspace{5px} 

\noindent\textbf{Equational theory of \wnetkat.}
We first record the semantic equivalence relation and the \wnetkat identities
used by the operation soundness proofs. The identities are routine semantic
facts, and their formal proofs are included in the Lean development.

\begin{definition}
    Define the \textit{semantic equivalence} relation $(\equiv): \Pol \times \Pol $ such that $p \equiv q$ if and only if $\sem{p} = \sem{q}$ using
    the usual component-wise equality on the maps $\sem{p}, \sem{q}: \packets \rightarrow \packets \rightarrow \semi$.
\end{definition}
\begin{lemma}
    $(\equiv)$ is an equivalence relation and, for any $p, q\in \Pol$ such that $p \equiv q$, we have
    \begin{itemize}
        \item $w \odot p  \equiv w \odot q$
        \item $p \oplus r \equiv q \oplus r$
        \item $p; r \equiv q; r$ and $r;p\equiv r;q$
        \item $p^* \equiv q^*$.
    \end{itemize}
\end{lemma}
\begin{lemma}
    \label{appendix:semiring-eqns}
    Abbreviate the \wnetkat predicates $\top \triangleq \TRUE$ and $\bot \triangleq \FALSE$. Then,
    for any $p, q, r\in \Pol$, the following identities concerning semiring operations are satisfied:
    \begin{enumerate}
        \item $p \oplus q \equiv q \oplus  p$
        \item $p \oplus (q \oplus r) \equiv (p \oplus q)\oplus r$
        \item $\bot \oplus p \equiv p\oplus \bot \equiv p$
        \item $(p;q);r \equiv p;(q;r)$
        \item $\top ; r
        \equiv r;\top \equiv r$
        \item $\bot; r \equiv r;\bot \equiv \bot$
        \item $(p \oplus q); r\equiv p;r \oplus q;r$,  $r;(p \oplus q)\equiv r;p \oplus r;q$
    \end{enumerate}
    the following identities concerning weighting are satisfied:
    \begin{enumerate}
        \setcounter{enumi}{7}
        \item $w \odot (r_1 \oplus r_2) \equiv w \odot r_1 \oplus w \odot r_2$
        \item $w \odot (r_1;r_2) \equiv (w \odot r_1);r_2$
        \item $(w_1 \cdot w_2) \odot r \equiv w_1 \odot (w_2 \odot r)$
        \item $(w_1 \oplus w_2) \odot r \equiv w_1 \odot r \oplus w_2 \odot r$
    \end{enumerate}
    and the following identities concerning Kleene stars are satisfied:
    \begin{enumerate}
        \setcounter{enumi}{11}
        \item $p^* \equiv \top \oplus p;p^*$
        \item $p^* \equiv \top \oplus p^*;p$
        \item $(p \oplus q)^* \equiv (p^*;q)^*;p^*$
        \item $(p \oplus q)^*  \equiv p^*;(q;p^*)^*$
    \end{enumerate}
\end{lemma}

\begin{lemma}
     \label{appendix:packet-alg}
    For any $f_1, f_2\in \fields$ and $v_1, v_2\in \values$, the following properties about sequential compositions of tests are satisfied:
    \begin{enumerate}
    \item $ f_1 = v_1 ; f_1 = v_2 \equiv f_1 = v_2 ; f_1 = v_1 \equiv \begin{cases}f_1 = v_1 & \text{if } v_1 = v_2 \\ \bot & \text{otherwise}\end{cases}$
    \item $ f_1 = v_1 ; f_1 \neq v_2 \equiv f_1 \neq v_1 ; f_1 = v_2 \equiv  \begin{cases}\bot  & \text{if } v_1=v_2 \\ f_1=v_1& \text{otherwise}\end{cases}$.
    \item $ f_1 \neq v_1 ; f_1 \neq v_2 \equiv f_1 \neq v_1 ; f_1 \neq v_2 \equiv  \begin{cases}f_1\neq v_2 & \text{if } v_1=v_2 \\ f_1 \neq v_1 ; f_1 \neq v_2  & \text{otherwise}\end{cases}$
    \item $f_1=v_1; f_2=v_2 \equiv f_2=v_2;f_1=v_1$
    \item $f_1\neq v_1; f_2\neq v_2 \equiv f_2\neq v_2;f_1\neq v_1$
    \item $f_1=v_1;f_2\neq v_2 \equiv f_2\neq v_2; f_1=v_1$
    \end{enumerate}
    the following properties about sequential compositions of assignments and tests are satisfied:
    \begin{enumerate}
        \setcounter{enumi}{6}
	    \item $ f_1 \leftarrow v_1 ; f_1\leftarrow v_2 \equiv f_1 \leftarrow v_2$
	    \item $ f_1\leftarrow v_1; f_1 = v_2 \equiv \begin{cases}f_1\leftarrow v_1 & \text{if } v_1=v_2 \\ \bot & \text{otherwise}\end{cases}$
	    \item $ f_1\leftarrow v_1; f_1\neq v_2 \equiv
	      \begin{cases}
	        \bot & \text{if } v_1=v_2,\\
	        f_1\leftarrow v_1 & \text{otherwise}.
	      \end{cases}$
	    \item $ f_1 =v_1; f_1\leftarrow v_1 \equiv f_1 = v_1$
	    \item $f_1 \leftarrow v_1; f_2\leftarrow v_2 \equiv f_2\leftarrow v_2; f_1\leftarrow v_1$
	    \item $f_1\leftarrow v_1; f_2=v_2 \equiv f_2=v_2;f_1\leftarrow v_1$ if $f_1\neq f_2$
	    \item $f_1\leftarrow v_1; f_2\neq v_2 \equiv f_2\neq v_2;f_1\leftarrow v_1$ if $f_1\neq f_2$
	  \end{enumerate}
	  the following properties about commutativity between policies is satisfied:
	  \begin{enumerate}
	       \setcounter{enumi}{13}
	    \item If $e$ is an atom (meaning $e\in \{f = v, f\neq v, f\leftarrow v, \top, \bot \}$) and $p\in \Pol$, then for any $w\in \semi$, $(w \odot e);p \equiv e;(w\odot p)$.
	    \item If $e$ is an atom that operates on $f_1$ and $q\in \Pol$ but $p$ does not operate
	          on $f_1$, then $e;p \equiv p;e$.
	  \end{enumerate}
\end{lemma}

\subsection{Helpers: Branch Maps and Lifted Row Operations}
\label{appendix:helpers}

A \emph{row} is a finite map $\values \rightharpoonup \WSPP_{\csr}$. The
operations act on rows via the following lifts, where
$\keys(m_1,\ldots,m_n) \triangleq
\dom(m_1)\cup\cdots\cup\dom(m_n)$, and
$\dflt{g(x)} \triangleq g(x)$ if $x \in \keys(g)$ and $\hat{\mathbb{0}}$
otherwise:
\[
\begin{aligned}
  \mathit{row}_1 \wsppmapadd \mathit{row}_2
    &\;\triangleq\;
    \{\, v \mapsto \dflt{\mathit{row}_1(v)} \wsppadd \dflt{\mathit{row}_2(v)}
       \mid v \in \keys(\mathit{row}_1,\mathit{row}_2) \,\},\\
  w \wsppmapscale \mathit{row}
    &\;\triangleq\;
    \{\, v \mapsto w \wsppscale \sigma \mid (v \mapsto \sigma) \in \mathit{row} \,\},\\
  \tau \wsppmapmul \mathit{row}
    &\;\triangleq\;
    \{\, v \mapsto \tau \wsppmul \sigma \mid (v \mapsto \sigma) \in \mathit{row} \,\}.
\end{aligned}
\]

Recalling \Cref{sec:wspp}, for a node $\rho=\wspp(\fielda,b,m,d)$,
its branch map is the finite-row valued function
\[
  \rho(u) \triangleq
  \begin{cases}
    b(u) & \text{if } u \in \keys(b),\\
    m \cup \{u \mapsto d\} & \text{if } u \notin \keys(b,m)
      \text{ and } d \neq \hat{\mathbb{0}},\\
    m & \text{otherwise.}
  \end{cases}
\]
For a leaf, or for a node whose root field is strictly greater than $\fielda$
when viewed at field $\fielda$, we write $\rho(u)\triangleq\{u\mapsto\rho\}$.

\emph{Lifting} is used to align operands: any wSPP $\sigma$
whose fields all exceed $\fielda$ may be viewed as the node
$\wspp(\fielda, \varnothing, \varnothing, \sigma)$; this is semantically
invisible, $\norm{\wspp(\fielda,\varnothing,\varnothing,\sigma)} \equiv
\norm{\sigma}$, since both guard products in \Cref{def:wspp-sem} are empty.
Given two wSPPs whose root fields differ---or a node and a leaf---the binary
operations below first lift the operand with the greater (or absent) root
field to the smaller root field. We may therefore assume both operands are
nodes on a common field $\fielda$.

We will use three helper facts below. The first says that lifted row addition
matches semantic choice under the induction hypothesis; the second connects a
branch map with the read-back semantics of its node; the third lets us prove
policy equivalence by slicing on all input values of a field.

\begin{lemma}[Soundness of lifted map sums]
  \label{appendix:map-lift-add}
  Let $f\in\fields$ and $m_1,m_2:\values\rightharpoonup\WSPP$. Suppose
  $\norm{\sigma_1\wsppadd\sigma_2}\equiv
  \norm{\sigma_1}\oplus\norm{\sigma_2}$ for all
  $\sigma_1,\sigma_2\in\WSPP$. Then
  \[
    \bigoplus_{v\mapsto\sigma\in m_1} f\leftarrow v;\norm{\sigma}
    \oplus
    \bigoplus_{v\mapsto\sigma\in m_2} f\leftarrow v;\norm{\sigma}
    \equiv
    \bigoplus_{v\mapsto\sigma\in m_1\wsppmapadd m_2}
      f\leftarrow v;\norm{\sigma}.
  \]
\end{lemma}
\begin{proof}
  Expanding $m_1\wsppmapadd m_2$ by the three kinds of keys gives
  {\small
  \begin{align*}
  &\bigoplus_{v\mapsto\sigma\in m_1\wsppmapadd m_2}
      f\leftarrow v;\norm{\sigma} \\
  &=
    \bigoplus_{\substack{v\mapsto\sigma\in m_1\\v\notin\keys(m_2)}}
      f\leftarrow v;\norm{\sigma\wsppadd\rulebox{\semizero}}
    \oplus
    \bigoplus_{\substack{v\mapsto\sigma\in m_2\\v\notin\keys(m_1)}}
      f\leftarrow v;\norm{\rulebox{\semizero}\wsppadd\sigma} \\
  &\qquad\oplus
    \bigoplus_{\substack{v\mapsto\sigma_1\in m_1\\v\mapsto\sigma_2\in m_2}}
      f\leftarrow v;\norm{\sigma_1\wsppadd\sigma_2}
      && \text{(definition of $\wsppmapadd$)}\\
  &\equiv
    \bigoplus_{\substack{v\mapsto\sigma\in m_1\\v\notin\keys(m_2)}}
      f\leftarrow v;(\norm{\sigma}\oplus\bot)
    \oplus
    \bigoplus_{\substack{v\mapsto\sigma\in m_2\\v\notin\keys(m_1)}}
      f\leftarrow v;(\bot\oplus\norm{\sigma}) \\
  &\qquad\oplus
    \bigoplus_{\substack{v\mapsto\sigma_1\in m_1\\v\mapsto\sigma_2\in m_2}}
      f\leftarrow v;(\norm{\sigma_1}\oplus\norm{\sigma_2})
      && \text{(assumption)}\\
  &\equiv
    \bigoplus_{\substack{v\mapsto\sigma\in m_1\\v\notin\keys(m_2)}}
      f\leftarrow v;\norm{\sigma}
    \oplus
    \bigoplus_{\substack{v\mapsto\sigma\in m_2\\v\notin\keys(m_1)}}
      f\leftarrow v;\norm{\sigma} \\
  &\qquad\oplus
    \bigoplus_{\substack{v\mapsto\sigma_1\in m_1\\v\mapsto\sigma_2\in m_2}}
      \big(f\leftarrow v;\norm{\sigma_1}
        \oplus f\leftarrow v;\norm{\sigma_2}\big)
      && \text{(Lem.~\ref{appendix:semiring-eqns}(3), (7))}\\
  &\equiv
    \bigoplus_{v\mapsto\sigma\in m_1} f\leftarrow v;\norm{\sigma}
    \oplus
    \bigoplus_{v\mapsto\sigma\in m_2} f\leftarrow v;\norm{\sigma}.
      && \text{(Lem.~\ref{appendix:semiring-eqns}(1), (2), repetitively).}
  \end{align*}
  }
\end{proof}

\begin{lemma}[Soundness of the branch map]
    \label{appendix:lookup-behavior}
    For any $\rho=\wspp(f, b, m, d)\in \WSPP$, let $\rho(u)$ be the branch map defined above. Then
    \[
    (f=u);\norm{\rho} \equiv  (f = u);\bigoplus_{v \mapsto \sigma\in \rho(u)} f \leftarrow v;\norm{\sigma}.
    \]
\end{lemma}
\begin{proof}
  For a fixed value $u$, abbreviate
  \[
  \begin{gathered}
  A_u \triangleq \bigoplus_{v\mapsto \sigma\in b(u)} f\leftarrow v;\norm{\sigma},
  \qquad
  C \triangleq \bigoplus_{z\mapsto \gamma\in m} f\leftarrow z;\norm{\gamma}.
  \end{gathered}
  \]
  We first compute the three components of $\norm{\rho}$. For the branch component,
  the leading test distributes over the branch sum and eliminates every
  incompatible first-level key:
  \begin{align*}
  (f=u);B_\rho
    &=
    (f=u);
    \Big(
      \bigoplus_{u'\in\keys(b)} f=u';
      \big(\bigoplus_{v\mapsto\sigma\in b(u')}
        f\leftarrow v;\norm{\sigma}\big)
    \Big)
      && \text{(Definition~\ref{def:wspp-sem})}\\
    &\equiv
    \bigoplus_{u'\in\keys(b)}
      (f=u);(f=u');
      \big(\bigoplus_{v\mapsto\sigma\in b(u')}
        f\leftarrow v;\norm{\sigma}\big)
      && \text{(Lemma~\ref{appendix:semiring-eqns}(7))}\\
    &\equiv
    \begin{cases}
      (f=u);\big(\bigoplus_{v\mapsto\sigma\in b(u)}
        f\leftarrow v;\norm{\sigma}\big)
        & \text{if }u\in\keys(b),\\
      \bot
        & \text{otherwise}
    \end{cases}
      && \text{(Lems.~\ref{appendix:packet-alg}(1),~\ref{appendix:semiring-eqns}(6))}\\
    &=
    \begin{cases}
      (f=u);A_u & \text{if }u\in\keys(b),\\
      \bot & \text{otherwise.}
    \end{cases}
  \end{align*}
  For the missing-assignment component, the same leading test is instead checked
  against the failed-branch tests:
  \begin{align*}
  (f=u);M_\rho
    &=
    (f=u);
    \Big(\prod_{u'\in\keys(b)}f\neq u'\Big);
    C
      && \text{(Definition~\ref{def:wspp-sem})}\\
    &\equiv
    \begin{cases}
      \bot & \text{if }u\in\keys(b),\\
      (f=u);C & \text{otherwise,}
    \end{cases}
      && \text{(Lems.~\ref{appendix:packet-alg}(2),~\ref{appendix:semiring-eqns}(6), rep.).}
  \end{align*}
  The default component must also pass the failed tests for missing assignments:
  \begin{align*}
  (f=u);D_\rho
    &=
    (f=u);
    \Big(\prod_{u'\in\keys(b)}f\neq u'\Big);
    \Big(\prod_{z\in\keys(m)}f\neq z\Big);
    \norm{d}
      && \text{(Definition~\ref{def:wspp-sem})}\\
    &\equiv
    \begin{cases}
      \bot
        & \text{if }u\in\keys(b),\\
      (f=u);\big(\prod_{z\in\keys(m)}f\neq z\big);\norm{d}
        & \text{otherwise}
    \end{cases}
      && \text{(Lems.~\ref{appendix:packet-alg}(2),~\ref{appendix:semiring-eqns}(6))}\\
    &\equiv
    \begin{cases}
      \bot
        & \text{if }u\in\keys(b)\text{ or }u\in\keys(m),\\
      (f=u);\norm{d}
        & \text{otherwise}
    \end{cases}
      && \text{(Lems.~\ref{appendix:packet-alg}(2),~\ref{appendix:semiring-eqns}(6))}\\
    &\equiv
    \begin{cases}
      \bot
        & \text{if }u\in\keys(b)\text{ or }u\in\keys(m),\\
      (f=u);(f\leftarrow u);\norm{d}
        & \text{otherwise}
    \end{cases}
      && \text{(Lem.~\ref{appendix:packet-alg}(10))}
  \end{align*}
  Combining the three component calculations with
  $\norm{\rho}=B_\rho\oplus M_\rho\oplus D_\rho$, we obtain
  \begin{align*}
  (f=u);\norm{\rho}
    &=
    (f=u);(B_\rho\oplus M_\rho\oplus D_\rho)
      \\
    &\equiv
    (f=u);B_\rho\oplus (f=u);M_\rho\oplus (f=u);D_\rho
      \\
    &\equiv
    \begin{cases}
      (f=u);A_u
        & \text{if }u\in\keys(b),\\
      (f=u);C
        & \text{if }u\notin\keys(b)\text{ and }u\in\keys(m),\\
      (f=u);\big(C\oplus f\leftarrow u;\norm{d}\big)
        & \text{otherwise}
    \end{cases}
  \end{align*}
  The first line uses Definition~\ref{def:wspp-sem}; the second uses
  Lemma~\ref{appendix:semiring-eqns}(7) repetitively; and the last uses the
  three component calculations above together with Lemma~\ref{appendix:semiring-eqns}(3).
  By the branch-map definition above, $\rho(u)=b(u)$ in
  the first case, $\rho(u)=m$ in the second, and
  $\rho(u)=m\cup\{u\mapsto d\}$ in the third. Thus all three cases are exactly
  \[
    (f=u);\bigoplus_{v\mapsto \sigma\in \rho(u)} f\leftarrow v;\norm{\sigma}.
  \]
\end{proof}

\begin{lemma}[Input slicing]
    \label{appendix:test-splitting}
    For any $p,q\in \Pol$, if there exists some $f\in \fields$ such that $(f=v);p \equiv (f=v);q$ for any $v\in \values$, then $p\equiv q$.
\end{lemma}
\begin{proof}
  From the semantics of predicates and $(\oplus)$, it could easily be observed that

  $\bigoplus_{v\in \values} f=v \equiv \top$. Thus we have
  \[
  p \equiv \top; p \equiv \Bigl(\bigoplus_{v\in \values} f=v\Bigr); p \equiv \Bigl(\bigoplus_{v\in \values} f=v; p\Bigr) \equiv
  \Bigl(\bigoplus_{v\in \values } f=v; q\Bigr) \equiv \Bigl(\bigoplus_{v\in \values} f=v\Bigr); q \equiv \top; q \equiv q,
  \]
  using Lemma \ref{appendix:semiring-eqns}(5), the observed fact, Lemma \ref{appendix:semiring-eqns}(7),
  the assumption, and then the same set of identities in reverse.
\end{proof}

\subsection{Predicates and Atoms}

The atoms $\strhat{(\EQ{\fielda}{v})}$, $\strhat{(\NOTEQ{\fielda}{v})}$,
$\strhat{(\ASSN{\fielda}{v})}$, $\strhat{\TRUE} = \hat{\mathbb{1}}$, and
$\strhat{\FALSE} = \hat{\mathbb{0}}$ are as in \Cref{sec:wspp-easy-ops}.
Compound predicates are compiled by
\[
  \strhat{(\AND{t}{t'})} \triangleq \strhat{t} \wsppmul \strhat{t'},
  \qquad
  \strhat{(\OR{t}{t'})} \triangleq \strhat{t} \wsppadd
    (\strhat{\NOT{t}} \wsppmul \strhat{t'}),
  \qquad
  \strhat{(\NOT{t})} \triangleq \strhat{\bar{t}},
\]
where $\bar{t}$ pushes the negation to the literals by De~Morgan
($\overline{\AND{t}{t'}} = \OR{\bar{t}}{\bar{t'}}$,
$\overline{\OR{t}{t'}} = \AND{\bar{t}}{\bar{t'}}$,
$\overline{\EQ{\fielda}{v}} = \NOTEQ{\fielda}{v}$, etc.). The disjoint
encoding of $\vee$ is forced by non-idempotence
(\Cref{example:non-idempotent}): the naive
$\strhat{t} \wsppadd \strhat{t'}$ would count packets satisfying both
disjuncts twice. Note that predicates compile to wSPPs whose leaves are
$\hat{\mathbb{1}}$ or $\hat{\mathbb{0}}$ and whose rows only ever rewrite a
field to its tested value.

\begin{theorem}[Soundness of atomic wSPPs]
  For every atomic policy $e$, $\norm{\strhat{e}} \equiv e$.
\end{theorem}
The proof is immediate by unfolding \Cref{def:wspp-sem}; each atomic wSPP has
only the one branch needed to realize the corresponding test, assignment,
$\top$, or $\bot$.

\subsection{Weighting}

Weighting multiplies every leaf from the left, preserving the tree
structure:
\[
  w \wsppscale \rulebox{w'} \;\triangleq\; \rulebox{w \cdot w'},
  \qquad
  w \wsppscale \wspp(\fielda, b, m, d)
  \;\triangleq\;
  \mathsf{mk}\bigl(\fielda,\;
    \{u \mapsto w \wsppmapscale b(u) \mid u \in \keys(b)\},\;
    w \wsppmapscale m,\;
    w \wsppscale d\bigr).
\]

The following theorem states soundness of weighting. As throughout this appendix,
we reason about the raw node displayed by the definition; the final call to
$\mathsf{mk}$ is semantics-preserving by \Cref{theorem:mk-soundness}.
\begin{theorem}
  For any $w\in \semi$ and $\rho\in \WSPP$, we have $\norm{w \wsppscale \rho} \equiv w \odot \norm{\rho}$.
\end{theorem}
\begin{proof}
  By induction on the structure of $\rho$.
  If $\rho=\rulebox{w_0}$, then
  \[
    \norm{w\wsppscale \rho}
    =
    \norm{\rulebox{w\cdot w_0}}
    =
    (w\cdot w_0)\odot\top
    \equiv
    w\odot(w_0\odot\top)
    =
    w\odot\norm{\rho}
  \]
  by Lemma \ref{appendix:semiring-eqns} (10). For the inductive case, suppose
  $\rho=\wspp(f,b,m,d)$, with components interpreted as finite maps as in
  Definition \ref{def:wspp-sem}.
  Let $\rho'=w\wsppscale\rho=\wspp(f,b',m',d')$, where
  \[
  \begin{aligned}
    b'(u) &= \{\,v\mapsto w\wsppscale \sigma
      \mid v\mapsto\sigma\in b(u)\,\}
      &&(u\in\keys(b)),\\
    m' &= \{\,z\mapsto w\wsppscale \gamma \mid z\mapsto\gamma\in m\,\},
    \qquad
    d'=w\wsppscale d.
  \end{aligned}
  \]
  By the definition of $\norm{\cdot}$, it suffices to compare the three
  components $B_\rho$, $M_\rho$, and $D_\rho$. We first compute
  $w \odot B_{\rho}$:
  {\small
  \begin{align*}
  w\odot B_\rho
    &=
    w\odot
    \Big(
      \bigoplus_{u\mapsto r\in b} f=u;
      \big(\bigoplus_{v\mapsto\sigma\in r}
        f\leftarrow v;\norm{\sigma}\big)
    \Big)
      && \text{(Definition~\ref{def:wspp-sem})}\\
    &\equiv
    \bigoplus_{u\mapsto r\in b}
      w\odot\Big(f=u;
      \big(\bigoplus_{v\mapsto\sigma\in r}
        f\leftarrow v;\norm{\sigma}\big)\Big)
      && \text{(Lemma~\ref{appendix:semiring-eqns}(8), repetitively)}\\
    &\equiv
    \bigoplus_{u\mapsto r\in b}
      f=u;
      \Big(w\odot
      \big(\bigoplus_{v\mapsto\sigma\in r}
        f\leftarrow v;\norm{\sigma}\big)\Big)
      && \text{(Lemmas~\ref{appendix:semiring-eqns}(9) and~\ref{appendix:packet-alg}(14), repetitively)}\\
    &\equiv
    \bigoplus_{u\mapsto r\in b}
      f=u;
      \Big(\bigoplus_{v\mapsto\sigma\in r}
        w\odot(f\leftarrow v;\norm{\sigma})\Big)
      && \text{(Lemma~\ref{appendix:semiring-eqns}(8), repetitively)}\\
    &\equiv
    \bigoplus_{u\mapsto r\in b}
      f=u;
      \Big(\bigoplus_{v\mapsto\sigma\in r}
        f\leftarrow v;(w\odot\norm{\sigma})\Big)
      && \text{(Lemmas~\ref{appendix:semiring-eqns}(9) and~\ref{appendix:packet-alg}(14), repetitively)}\\
    &\equiv
    \bigoplus_{u\mapsto r\in b}
      f=u;
      \Big(\bigoplus_{v\mapsto\sigma\in r}
        f\leftarrow v;\norm{w\wsppscale \sigma}\Big)
      && \text{(induction hypothesis, repetitively)}\\
    &=
    B_{\rho'}
      && \text{(Definition~\ref{def:wspp-sem}, definition of $\rho'$).}
  \end{align*}
  }

  The missing-assignment component is analogous, but the first use of
  Lemmas~\ref{appendix:semiring-eqns}(9) and~\ref{appendix:packet-alg}(14)
  pushes the weight through the product of failed branch tests:
  {\small
  \begin{align*}
  w\odot M_\rho
    &=
    w\odot\Big(
      \big(\prod_{u\in\keys(b)}f\neq u\big);
      \big(\bigoplus_{z\mapsto\gamma\in m}
        f\leftarrow z;\norm{\gamma}\big)
    \Big)
      && \text{(Definition~\ref{def:wspp-sem})}\\
    &\equiv
    \big(\prod_{u\in\keys(b)}f\neq u\big);
    \Big(w\odot
      \big(\bigoplus_{z\mapsto\gamma\in m}
        f\leftarrow z;\norm{\gamma}\big)\Big)
      && \text{(Lemmas~\ref{appendix:semiring-eqns}(9) and~\ref{appendix:packet-alg}(14), repetitively)}\\
    &\equiv
    \big(\prod_{u\in\keys(b)}f\neq u\big);
    \Big(\bigoplus_{z\mapsto\gamma\in m}
      w\odot(f\leftarrow z;\norm{\gamma})\Big)
      && \text{(Lemma~\ref{appendix:semiring-eqns}(8), repetitively)}\\
    &\equiv
    \big(\prod_{u\in\keys(b)}f\neq u\big);
    \Big(\bigoplus_{z\mapsto\gamma\in m}
      f\leftarrow z;(w\odot\norm{\gamma})\Big)
      && \text{(Lemmas~\ref{appendix:semiring-eqns}(9) and~\ref{appendix:packet-alg}(14), repetitively)}\\
    &\equiv
    \big(\prod_{u\in\keys(b)}f\neq u\big);
    \Big(\bigoplus_{z\mapsto\gamma\in m}
      f\leftarrow z;\norm{w\wsppscale \gamma}\Big)
      && \text{(induction hypothesis, repetitively)}\\
    &=
    M_{\rho'}
      && \text{(Definition~\ref{def:wspp-sem}, definition of $\rho'$).}
  \end{align*}
  }

  Finally, for the default component,
  {\small
  \begin{align*}
  w\odot D_\rho
    &=
    w\odot\Big(
      \big(\prod_{u\in\keys(b)}f\neq u\big);
      \big(\prod_{z\in\keys(m)}f\neq z\big);
      \norm{d}
    \Big)
      && \text{(Definition~\ref{def:wspp-sem})}\\
    &\equiv
    \big(\prod_{u\in\keys(b)}f\neq u\big);
    \big(\prod_{z\in\keys(m)}f\neq z\big);
    (w\odot\norm{d})
      && \text{(Lemmas~\ref{appendix:semiring-eqns}(9) and~\ref{appendix:packet-alg}(14), repetitively)}\\
    &\equiv
    \big(\prod_{u\in\keys(b)}f\neq u\big);
    \big(\prod_{z\in\keys(m)}f\neq z\big);
    \norm{w\wsppscale d}
      && \text{(induction hypothesis)}\\
    &=
    D_{\rho'}
      && \text{(Definition~\ref{def:wspp-sem}, definition of $\rho'$).}
  \end{align*}
  }
  Therefore
  \begin{align*}
  \norm{w\wsppscale\rho}
    &= B_{\rho'}\oplus M_{\rho'}\oplus D_{\rho'}
      && \text{(Definition~\ref{def:wspp-sem})}\\
    &\equiv (w\odot B_\rho)\oplus(w\odot M_\rho)\oplus(w\odot D_\rho)
      && \text{(component calculations above)}\\
    &\equiv w\odot(B_\rho\oplus M_\rho\oplus D_\rho)
      && \text{(Lemma~\ref{appendix:semiring-eqns}(8), repetitively, from right to left)}\\
    &= w\odot\norm{\rho}
      && \text{(Definition~\ref{def:wspp-sem}).}
  \end{align*}
\end{proof}

\subsection{Choice}

For nodes $\rho = \wspp(\fielda, b_1, m_1, d_1)$ and
$\sigma = \wspp(\fielda, b_2, m_2, d_2)$ on a common field:
\[
  \rho \wsppadd \sigma
  \;\triangleq\;
  \mathsf{mk}(\fielda,\; b_{+},\; m_1 \wsppmapadd m_2,\; d_1 \wsppadd d_2),
  \qquad
  b_{+} \triangleq
  \{\, u \mapsto \rho(u) \wsppmapadd \sigma(u) \mid u \in K \,\},
\]
where $K \triangleq \keys(b_1,b_2,m_1,m_2)$.

\begin{theorem}[Correctness of $\wsppadd$]
  \label{theorem:wspp-add-soundness}
  For any $\rho,\sigma\in \WSPP$, we have
  \[
    \norm{\rho\wsppadd\sigma}\equiv \norm{\rho}\oplus\norm{\sigma},
  \]
  where $(\wsppadd)$ is defined above.
\end{theorem}
\begin{proof}

  WLOG, we assume $\rho, \sigma$ have been 
  lifted to the same top-level field as described in Section \ref{appendix:helpers} 
  and will assume so implicitly for all binary operations on wSPPs. 
  We now proceed by induction. If
  $\rho=\rulebox{w_1}$ and $\sigma=\rulebox{w_2}$, then
  \begin{align*}
  \norm{\rho\wsppadd\sigma}
    &=
    \norm{\rulebox{w_1+w_2}}\\
    &=
    (w_1+w_2)\odot\top\\
    &\equiv
    (w_1\odot\top)\oplus(w_2\odot\top)
      && \text{(Lemma~\ref{appendix:semiring-eqns}(11))}\\
    &=
    \norm{\rho}\oplus\norm{\sigma}.
  \end{align*}
  For the inductive case, suppose
  \[
    \rho=\wspp(f,b_\rho,m_\rho,d_\rho),\qquad
    \sigma=\wspp(f,b_\sigma,m_\sigma,d_\sigma),
  \]
  and let $\tau\triangleq\rho\wsppadd\sigma$. For each $u\in\values$, the
  lookup and map-lift lemmas give
  {\small
  \begin{align*}
  &(f=u);(\norm{\rho}\oplus\norm{\sigma})\\
  &\equiv
    (f=u);\norm{\rho}\oplus(f=u);\norm{\sigma}
      && \text{(Lemma~\ref{appendix:semiring-eqns}(7))}\\
  &\equiv
    (f=u);\bigoplus_{v\mapsto\eta\in\rho(u)}
      f\leftarrow v;\norm{\eta}
    \oplus
    (f=u);\bigoplus_{v\mapsto\eta\in\sigma(u)}
      f\leftarrow v;\norm{\eta}
      && \text{(Lemma~\ref{appendix:lookup-behavior})}\\
  &\equiv
    (f=u);
    \left(
      \bigoplus_{v\mapsto\eta\in\rho(u)} f\leftarrow v;\norm{\eta}
      \oplus
      \bigoplus_{v\mapsto\eta\in\sigma(u)} f\leftarrow v;\norm{\eta}
    \right)
      && \text{(Lemma~\ref{appendix:semiring-eqns}(7))}\\
  &\equiv
    (f=u);
    \bigoplus_{v\mapsto\eta\in\rho(u)\wsppmapadd\sigma(u)}
      f\leftarrow v;\norm{\eta}
      && \text{(Lemma~\ref{appendix:map-lift-add}, induction hypothesis).}
  \end{align*}
  }
  It remains only to identify the lookup map of $\tau$. If
  $u\in\keys(b_\rho,b_\sigma,m_\rho,m_\sigma)$, then
  the definition of $(\wsppadd)$ gives
  $\tau(u)=\rho(u)\wsppmapadd\sigma(u)$. If otherwise, then
  \[
    \rho(u)\wsppmapadd\sigma(u)
    =
    (m_\rho\cup\{u\mapsto d_\rho\})\wsppmapadd
    (m_\sigma\cup\{u\mapsto d_\sigma\})
    =
    m_+\cup\{u\mapsto d_+\}
    =
    \tau(u),
  \]
  where $m_+\triangleq m_\rho\wsppmapadd m_\sigma$ and
  $d_+\triangleq d_\rho\wsppadd d_\sigma$ as in
  the definition of $(\wsppadd)$. Therefore, for all $u\in\values$,
  \begin{align*}
  (f=u);(\norm{\rho}\oplus\norm{\sigma})
    &\equiv
    (f=u);
    \bigoplus_{v\mapsto\eta\in\tau(u)} f\leftarrow v;\norm{\eta}\\
    &\equiv
    (f=u);\norm{\tau}
      && \text{(Lemma~\ref{appendix:lookup-behavior}, from right to left).}
  \end{align*}
  This allows us to conclude $\norm{\rho}\oplus\norm{\sigma}\equiv\norm{\tau}$ by
  Lemma~\ref{appendix:test-splitting}, as required.
\end{proof}

\subsection{Sequential Composition}

Sequencing composes every continuation of the left operand with the branch
of the right operand at the intermediate value. For a row $\mathit{row}$
(a branch of $\rho$) define its composition with $\sigma$ by
\[
  \mathit{row} \rowmul \sigma
  \;\triangleq\;
  \bigwsppmapadd_{(v \mapsto \tau) \in \mathit{row}}
    \tau \wsppmapmul \sigma(v),
\]
the symbolic rendition of the matrix product: $v$ ranges over the values
written by the left operand, and entries for equal final values are merged
with $\wsppadd$. Then, for nodes operating on a common field $\fielda$, define 
\[
  \rho \wsppmul \sigma
  \;\triangleq\;
  \mathsf{mk}(\fielda,\; b_{;},\; m_{;},\; d_1 \wsppmul d_2),
  \qquad
  \begin{aligned}
    m_A &\triangleq m_1 \rowmul \sigma, \qquad
    m_{;} \triangleq m_A \wsppmapadd (d_1 \wsppmapmul m_2),\\
    b_{;} &\triangleq \{\, u \mapsto \rho(u) \rowmul \sigma \mid u \in K \,\},
  \end{aligned}
\]
where $K$ now additionally includes $\keys(m_A)$. 

Expanding the row-level definition above gives the following pointwise form,
which is the form used in the soundness proof.

\begin{definition}[Expanded sequencing form]\label{appendix:wsppmul-definition}
  {\small
  For $\rho=\wspp(f,b_\rho,m_\rho,d_\rho)$ and
  $\sigma=\wspp(f,b_\sigma,m_\sigma,d_\sigma)$, define
  \[
  \begin{aligned}
    \rho\wsppmul\sigma &\triangleq \wspp(f,b',m',d_\rho\wsppmul d_\sigma),\\[-0.15em]
    b' &\triangleq
      \left\{\,u \mapsto
      \displaystyle\bighoplusmap_{v\mapsto\eta\in\rho(u)}
        \{\,v'\mapsto \eta\wsppmul\theta \mid v'\mapsto\theta\in\sigma(v)\,\}
      \,\middle|\,
      u\in\keys(b_\rho,b_\sigma,m_\rho,m_\sigma,m_A)
      \right\},\\[-0.15em]
    m' &\triangleq m_A\wsppmapadd m_B,\qquad
    m_A \triangleq
      \displaystyle\bighoplusmap_{v\mapsto\eta\in m_\rho}
        \{\,v'\mapsto \eta\wsppmul\theta \mid v'\mapsto\theta\in\sigma(v)\,\},\\[-0.15em]
    m_B &\triangleq
      \{\,v'\mapsto d_\rho\wsppmul\theta \mid v'\mapsto\theta\in m_\sigma\,\}.
  \end{aligned}
  \]
  }
\end{definition}
\begin{theorem}
    For any $\rho, \sigma\in \WSPP$, $\norm{\rho \wsppmul \sigma} \equiv \norm{\rho}; \norm{\sigma}$.
\end{theorem}
\begin{proof}
Instead of inducting on the structure of $\rho$ and $\sigma$ again,
we need a stronger inductive hypothesis: that is, we induct downwards on
$0 \leq i \leq |\fields|$ and assume the proposition
for any $\rho, \sigma$ that operate on fields greater than the
$i$th field in $\fields$. For the base case, $\rho = \rulebox{w_1}$ and $\sigma = \rulebox{w_2}$, so we have
\[\norm{\rulebox{w_1}}; \norm{\rulebox{w_2}} = (w_1 \odot \top) ; (w_2 \odot \top) \equiv (w_1 \cdot w_2) \odot \top = \norm{\rulebox{w_1 \cdot w_2}} = \norm{\rulebox{w_1}\wsppmul \rulebox{w_2}}\]
where the equivalence is by Lemma~\ref{appendix:semiring-eqns}(10). For the inductive case, as in the choice case, we may assume that $\rho = \wspp(f, b_\rho, m_\rho, d_\rho), \sigma = \wspp(f, b_\sigma, m_\sigma, d_\sigma)$, and let $\rho \wsppmul \sigma = \wspp(f, b', m', d_\rho \wsppmul d_\sigma)$ be as defined in
Definition \ref{appendix:wsppmul-definition}.
We have
{\small
\begin{align*}
     (f=u);\norm{\rho} & \equiv (f=u);\bigoplus_{v\mapsto \eta \in \rho(u)} f\leftarrow v;\norm{\eta} & & \text{(Lemma~\ref{appendix:lookup-behavior})}\\
     &  \equiv (f=u);\bigoplus_{v\mapsto \eta \in \rho(u)} \norm{\eta};(f \leftarrow v). & & \text{(Lemma~\ref{appendix:packet-alg}(15))} \\
     & \equiv (f=u);\bigoplus_{v\mapsto \eta\in \rho(u)} \norm{\eta};(f \leftarrow v);(f=v)  & & \text{(Lemma~\ref{appendix:packet-alg}(8))} \\
    (f=u);\norm{\rho};\norm{\sigma} &  \equiv (f=u);\bigoplus_{v\mapsto \eta \in \rho(u)} \norm{\eta};(f \leftarrow v);(f=v);\norm{\sigma} \\
     & \equiv (f=u);\bigoplus_{v\mapsto \eta\in \rho(u)}\norm{\eta};(f \leftarrow v);(f=v);\bigoplus_{v' \mapsto \theta\in \sigma(v)} f\leftarrow v';\norm{\theta}& & \text{(Lemma~\ref{appendix:lookup-behavior})} \\
     & \equiv (f=u);\bigoplus_{v\mapsto \eta\in \rho(u)}\norm{\eta};(f \leftarrow v);\bigoplus_{v' \mapsto \theta\in \sigma(v)} f\leftarrow v';\norm{\theta}& & \text{(Lemma~\ref{appendix:packet-alg}(10))} \\
     &  \equiv (f=u);\bigoplus_{v\mapsto \eta\in \rho(u)}\norm{\eta};\bigoplus_{v' \mapsto \theta\in \sigma(v)} (f \leftarrow v);(f\leftarrow v');\norm{\theta} \\
     &  \equiv (f=u);\bigoplus_{v\mapsto \eta\in \rho(u)}\norm{\eta};\bigoplus_{v' \mapsto \theta\in \sigma(v)} (f \leftarrow v');\norm{\theta}& & \text{(Lemma~\ref{appendix:packet-alg}(7))} \\
     & \equiv (f=u);\bigoplus_{v\mapsto \eta\in \rho(u)}\bigoplus_{v' \mapsto \theta\in \sigma(v)} \norm{\eta};(f \leftarrow v');\norm{\theta} \\
     &  \equiv (f=u);\bigoplus_{v\mapsto \eta\in \rho(u)}\bigoplus_{v' \mapsto \theta\in \sigma(v)} (f\leftarrow v');\norm{\eta};\norm{\theta}& & \text{(Lemma~\ref{appendix:packet-alg}(15))}  \\
     & \equiv (f=u);\bigoplus_{v\mapsto \eta\in \rho(u)}\bigoplus_{v' \mapsto \theta\in \sigma(v)} (f\leftarrow v');\norm{\eta \wsppmul \theta}. & & \text{(IH)}  \\
\end{align*}
}
This means
{\small
\begin{equation}
\label{eq:wspp-mul-proof-2}
 (f=u);\norm{\rho};\norm{\sigma} \equiv (f=u);\bigoplus_{v\mapsto \eta\in \rho(u)}\bigoplus_{v'\mapsto \theta\in b'_{v, \eta}}(f \leftarrow v'); \norm{\theta},
\end{equation}
}
where $b'_{v, \eta} \triangleq \{v' \mapsto  \eta\wsppmul \theta \,|\, v' \mapsto \theta \in \sigma(v) \}$.
Then, by applying Lemma~\ref{appendix:map-lift-add} repeatedly over the list of all $b'_{v, \eta}$ in the inner sum of (\ref{eq:wspp-mul-proof-2}) (the condition
for the lemma is satisfied by Theorem \ref{theorem:wspp-add-soundness}), we obtain
{\small
\begin{equation}
\label{eq:wspp-mul-proof-3}
(f=u);\norm{\rho};\norm{\sigma} \equiv (f=u);\bigoplus_{v' \mapsto \gamma \in b'_u} (f\leftarrow v');\norm{\gamma},
\end{equation}
}
where $b'_u \triangleq \displaystyle\bighoplusmap_{v\mapsto \eta\in \rho(u)} b'_{v, \eta}$. When
\[
u\in \keys(b_\rho,b_\sigma,m_\rho,m_\sigma,m_A),
\]
Definition~\ref{appendix:wsppmul-definition} gives $(\rho\wsppmul\sigma)(u) = b'(u) = b'_u$, so we immediately have from (\ref{eq:wspp-mul-proof-3}) that

\begin{align*}
    (f=u);\norm{\rho};\norm{\sigma}
    & \equiv   (f=u); \bigoplus_{v' \mapsto \gamma \in (\rho\wsppmul\sigma)(u)} (f\leftarrow  v');\norm{\gamma} \\
    & \equiv (f=u); \norm{\rho\wsppmul\sigma}. & & \text{(Lemma~\ref{appendix:lookup-behavior})}
\end{align*}

	Alternatively, suppose
	\[
	u\notin \keys(b_\rho,b_\sigma,m_\rho,m_\sigma,m_A).
	\]
	Then $\rho(u) = m_\rho \cup \{u \mapsto d_\rho\}$ and $\sigma(u) = m_\sigma \cup \{u \mapsto d_\sigma\}$. Using the default lookup for $\rho$ in (\ref{eq:wspp-mul-proof-2}), we obtain
	{\small
	\begin{align*}
	     &(f=u);\norm{\rho};\norm{\sigma}\\[-0.25em]
	      \equiv & (f=u);\bigoplus_{v\mapsto \eta\in (\{u \mapsto d_\rho\} \cup m_\rho)}\bigoplus_{v' \mapsto \theta\in \sigma(v)} (f\leftarrow v');\norm{\eta \wsppmul \theta}\\
	     \shortintertext{\hfill{\footnotesize(split the $v=u$ summand)}}
	     \equiv & (f=u);\left(\bigoplus_{v' \mapsto \theta\in \sigma(u)} (f\leftarrow v');\norm{d_\rho \wsppmul \theta} \oplus  \bigoplus_{v\mapsto \eta\in m_\rho}\bigoplus_{v' \mapsto \theta\in \sigma(v)} (f\leftarrow v');\norm{\eta \wsppmul \theta}\right)\\
	     \shortintertext{\hfill{\footnotesize(Definition~\ref*{appendix:wsppmul-definition} ($m_A$), Lemma~\ref*{appendix:map-lift-add}, similar to above)}}
	     \equiv & (f=u);\left(\bigoplus_{v' \mapsto \theta\in \sigma(u)} (f\leftarrow v');\norm{d_\rho \wsppmul \theta} \oplus \bigoplus_{v' \mapsto \gamma\in m_A} f\leftarrow v';\norm{\gamma} \right)\\
	     \shortintertext{\hfill{\footnotesize(Definition of $\sigma(u)$)}}
	      \equiv & (f=u);\left(\bigoplus_{v' \mapsto \theta\in (\{u \mapsto d_\sigma\} \cup m_\sigma )} (f\leftarrow v');\norm{d_\rho \wsppmul \theta} \oplus  \bigoplus_{v' \mapsto \gamma\in m_A} f\leftarrow v';\norm{\gamma}\right)\\
	     \shortintertext{\hfill{\footnotesize(split the $v'=u$ summand)}}
	      \equiv & (f=u);\left(
	        \begin{aligned}
	        &\bigoplus_{v' \mapsto \gamma\in \{u \mapsto d_\rho\wsppmul d_\sigma\}} f\leftarrow v';\norm{\gamma}
	        \oplus \bigoplus_{v' \mapsto \theta\in m_\sigma} (f\leftarrow v');\norm{d_\rho \wsppmul \theta}\\
	        &\oplus \bigoplus_{v' \mapsto \gamma\in m_A} f\leftarrow v';\norm{\gamma}
	        \end{aligned}
	        \right)\\
	     \shortintertext{\hfill{\footnotesize(Definition~\ref*{appendix:wsppmul-definition} ($m_B$))}}
	      \equiv & (f=u);\left(
	        \begin{aligned}
	        &\bigoplus_{v' \mapsto \gamma\in \{u \mapsto d_\rho\wsppmul d_\sigma\}} f\leftarrow v';\norm{\gamma}
	        \oplus \bigoplus_{v' \mapsto \gamma\in m_B} f\leftarrow v';\norm{\gamma}\\
	        &\oplus \bigoplus_{v' \mapsto \gamma\in m_A} f\leftarrow v';\norm{\gamma}
	        \end{aligned}
	        \right)\\
	     \shortintertext{\hfill{\footnotesize(Definition~\ref*{appendix:wsppmul-definition} ($m'$), Lemma~\ref*{appendix:map-lift-add})}}
	     \equiv&  (f=u);\left(
	        \begin{aligned}
	        &\bigoplus_{v' \mapsto \gamma\in \{u \mapsto d_\rho\wsppmul d_\sigma\}} f\leftarrow v';\norm{\gamma}
	        \oplus \bigoplus_{v' \mapsto \gamma\in m'} f\leftarrow v';\norm{\gamma}
	        \end{aligned}
	        \right) \\
	     \equiv & (f=u);\bigoplus_{\substack{v' \mapsto \gamma\in\\(\{u \mapsto d_\rho \wsppmul d_\sigma\} \cup m')}} f\leftarrow v';\norm{\gamma}
	\end{align*}
	}
	Since
	$u\notin \keys(b_\rho,b_\sigma,m_\rho,m_\sigma,m_A) = \keys(b',m')$,
	Definition~\ref{appendix:wsppmul-definition} gives
	$(\rho \wsppmul \sigma)(u) = \{u \mapsto d_\rho \wsppmul d_\sigma\} \cup m'$.
The last line above therefore has the same form as in the preceding case, and
	Lemma~\ref{appendix:lookup-behavior} again yields
	\[
	    (f=u);\norm{\rho};\norm{\sigma}
	    \equiv (f=u); \norm{\rho\wsppmul\sigma}.
	\]
	
	Since $(f=u);\norm{\rho\wsppmul\sigma} \equiv (f=u);(\norm{\rho} ; \norm{\sigma})$ for all $u\in \values$, we conclude that $\norm{\rho\wsppmul\sigma} \equiv \norm{\rho} ; \norm{\sigma}$ by Lemma~\ref{appendix:test-splitting}.
\end{proof}

\subsection{The Smart Constructor}\label{sec:appendix-mk}

$\mathsf{mk}(\fielda, b, m, d)$ restores canonicity after an operation has
assembled raw components. It performs four phases:
\begin{enumerate}
  \item \emph{Unblock dead defaults.} An entry $(v \mapsto
    \hat{\mathbb{0}}) \in m$ is semantically dead, but its presence
    \emph{suppresses} the diagonal default $d$ at $v$
    (cf.\ the branch map in \Cref{sec:wspp}). Deleting it naively would
    unshadow $d$ and change the semantics. Hence: remove the entry and,
    unless $v \in \keys(b)$ already, freeze the current behavior by
    inserting the explicit row $b(v) \leftarrow m$ (the \emph{original}
    $m$).
  \item \emph{Prune zeros.} Delete all entries $(v \mapsto
    \hat{\mathbb{0}})$ inside the rows of $b$.
  \item \emph{Drop redundant rows.} Delete a row $b(u)$ iff it equals the
    branch that the defaults would produce at $u$ anyway---that is,
    $m \cup \{u \mapsto d\}$ if the diagonal applies at $u$
    ($u \notin \keys(m)$ and $d \neq \hat{\mathbb{0}}$), and $m$ otherwise.
  \item \emph{Collapse trivial nodes.} If $b$ and $m$ are now empty, return
    $d$ instead of a node.
\end{enumerate}
Phases 1 and 3 are the reason equality of sub-wSPPs must be decidable (it
is, structurally, since $=$ is decidable on $\csrset$); with hash-consing,
these tests are constant time in practice
(\Cref{sec:implementation-evaluation}).

\begin{theorem}[Soundness of the smart constructor]
  \label{theorem:mk-soundness}
  For all well-formed components $f,b,m,d$,
  \[
    \norm{\mathsf{mk}(f,b,m,d)} \equiv \norm{\wspp(f,b,m,d)}.
  \]
\end{theorem}
We do not reproduce the proof here, as it could be more directly adopted from
; it is discharged in the Lean
formalization by checking that each normalization phase preserves the read-back
semantics of the node.

\subsection{The Kleene Star}
The complete algorithm for computing the Kleene star $\rho^{\wsppstar}$ of
an wSPP $\rho$ is given in Figure \ref{alg:df-star}. The main proof obligation
is local preservation: each pass emits one factor and rewrites the current
residual wSPP so that the surrounding star is unchanged. This is the content of
Lemma \ref{lemma:df-star-invariant}; the remaining termination and assembly
arguments are either immediate or already spelled out around
Figure \ref{alg:df-star}.
\begin{lemma}
  In any of the three passes in an instance $\texttt{Star}(\rho)$ of the algorithm
  in  Figure \ref{alg:df-star}, let $\rho'$ be the residual
  wSPP $\,\wspp(f, b, m, d)$ at program point \AlgPoint{wsppdefault}{I-s},
  \AlgPoint{wsppbranch}{II-s}, or \AlgPoint{wsppmissing}{III-s},
  and let $\rho''$ be the residual at the corresponding point
  \AlgPoint{wsppdefault}{I-e}, \AlgPoint{wsppbranch}{II-e}, or
  \AlgPoint{wsppmissing}{III-e} of Figure \ref{alg:star}. If $\texttt{Star}(\sigma)$
  is sound on all wSPP $\sigma$ that operate on fields greater than $f$ (i.e. the
  top-level field of $\rho'$), then
  \begin{enumerate}
    \item Pass~I satisfies  $\ITER{\norm{\rho'}} \equiv
  \SEQ{\norm{\psi}}{\ITER{\norm{\rho''}}}$ for the emitted factor $\psi$;
    \item each iteration of Pass~II satisfies $\ITER{\norm{\rho'}} \equiv
  \SEQ{\norm{\psi}}{\ITER{\norm{\rho''}}}$ for the emitted factor $\psi$;
    \item each iteration of Pass~III satisfies $\ITER{\norm{\rho'}} \equiv
  \SEQ{\ITER{\norm{\rho''}}}{\norm{\phi}}$ for the emitted factor $\phi$.
  \end{enumerate}
\end{lemma}
\begin{proof}
  We spell out the Pass~I item. Let the residual at
  \AlgPoint{wsppdefault}{I-s} be
  $\rho'=\wspp(f,b,m,d)$, and abbreviate
  \[
    G_b \triangleq
    \prod_{\mathclap{u\in\keys(b)}} \NOTEQ{f}{u},
    \qquad
    G_m \triangleq
    \prod_{\mathclap{z\in\keys(m)}} \NOTEQ{f}{z},
    \qquad
    G \triangleq G_b\SEQN G_m,
    \qquad
    D_{\rho'}^+ \triangleq D_{\rho'};D_{\rho'}^*.
  \]
  Pass~I emits
  \[
    \psi =
    \hat{\mathbb{1}}\wsppadd
    \guardprod{u\in\keys(b)}{(f\ne u)}
    \wsppmul
    \guardprod{z\in\keys(m)}{(f\ne z)}
    \wsppmul d\wsppmul\Star(d),
  \]
  modifies every entry $b(u)(v)=\sigma$ when
  $v\notin\keys(b,m)$, leaves all other entries unchanged, and then
  sets $d$ to $\hat{\mathbb{0}}$. Thus the residual at
  \AlgPoint{wsppdefault}{I-e} is
  $\rho''=\wspp(f,b'',m,\hat{\mathbb{0}})$, where $b''$
  is the updated test-assignment map just described.

  We first isolate the default-identity summand. By
  \Cref{def:wspp-sem},
  $\norm{\rho'}\equiv B_{\rho'}\oplus M_{\rho'}\oplus D_{\rho'}$. Hence
  \begin{align}
    \ITER{\norm{\rho'}}
    &\equiv \ITER{\bigl(D_{\rho'}\oplus (B_{\rho'}\oplus M_{\rho'})\bigr)}
      && \text{(Definition~\ref*{def:wspp-sem}, Lemma~\ref{appendix:semiring-eqns}(1),(2))}
      \notag \\
    &\equiv D_{\rho'}^*((B_{\rho'} \oplus M_{\rho'});D_{\rho'}^*)^*
      && \text{(Lemma~\ref{appendix:semiring-eqns}(15))}
      \notag \\
    &\equiv D_{\rho'}^*((B_{\rho'} \oplus M_{\rho'});
      (\top \oplus D_{\rho'}^+))^*
      && \text{(Lemma~\ref{appendix:semiring-eqns}(12))}
      \notag \\
    &\equiv D_{\rho'}^*
      \bigl((B_{\rho'} \oplus B_{\rho'};D_{\rho'}^+)
      \oplus (M_{\rho'} \oplus M_{\rho'};D_{\rho'}^+) \bigr)^*
      && \text{(Lemma~\ref{appendix:semiring-eqns}(7))}
      \label{eq:pass-1-factorization}
  \end{align}

  It remains to identify both the emitted factor and the residual wSPP in
  \eqref{eq:pass-1-factorization}. First, we compute $D_{\rho'}^+$.
  All infinite sums and products below are taken in the
  $\omega$-continuous semantic domain $\semdom \triangleq \Pk \rightarrow \Pk \rightarrow S$ of
  \Cref{sec:prelim-wnetkat}.
{\small
\begin{align*}
  \sem{D_{\rho'}^+}
  &= \sem{D_{\rho'}}\cdot\sem{D_{\rho'}^*}
    && \text{(definition of $D_{\rho'}^+$)}\\
  &= \sem{D_{\rho'}}\cdot\sum_{i\in\N}\sem{D_{\rho'}^{(i)}}
    && \text{($\omega$-continuity of $\semdom$)}\\
  &= \sum_{i=1}^{\infty}\sem{D_{\rho'}^{(i)}}
    && \text{(continuity of matrix multiplication)}\\
  &= \sum_{i=1}^{\infty}\sem{(G;\norm{d})^{(i)}}
    && \text{(Definition~\ref*{def:wspp-sem})}\\
  &= \sum_{i=1}^{\infty}\sem{G;\norm{d}^{(i)}}
    && \text{(Lemma~\ref{appendix:packet-alg}(3),(15), inductively)}\\
  &= \sem{G}\cdot \sum_{i=1}^{\infty}\sem{\norm{d}^{(i)}}
    && \text{($\omega$-continuity of $\semdom$)}\\
  &= \sem{G;\norm{d}^+}.
\end{align*}
}
  Hence $D_{\rho'}^+ \equiv G;\norm{d}^+$. Since $d$ mentions only fields
  greater than $f$, the premise gives
  $\norm{\Star(d)}\equiv \ITER{\norm{d}}$. Therefore the factor emitted by
  Pass~I satisfies
  \[
    \norm{\psi}
    \equiv \top \oplus G;\norm{d};\norm{\Star(d)}
    \equiv \top \oplus D_{\rho'}^+
    \equiv D_{\rho'}^*,
  \]
  using \Cref{theorem:wspp-add-mul-soundness} for the wide-hatted atoms,
  $\wsppadd$, and $\wsppmul$, and Lemma~\ref{appendix:semiring-eqns}(12) for
  the final equality.

  Now compute how sequencing $D_{\rho'}^+$ to the left changes the remaining two
  summands of \eqref{eq:pass-1-factorization}.
{\small
\begin{align*}
  B_{\rho'};D_{\rho'}^+
  &\equiv
    \bigoplus_{u\in\keys(b)}(\EQ{f}{u})\SEQN
      \bigoplus_{\mathclap{(v\mapsto\sigma)\in b(u)}}
        (\ASSN{f}{v})\SEQN\norm{\sigma}\SEQN G\SEQN\norm{d}^+
    \\
  \shortintertext{\hfill{\footnotesize(Definition~\ref*{def:wspp-sem}; $D_{\rho'}^+\equiv G;\norm{d}^+$)}}
  &\equiv
    \bigoplus_{u\in\keys(b)}(\EQ{f}{u})\SEQN
      \bigoplus_{\mathclap{(v\mapsto\sigma)\in b(u)}}
        (\ASSN{f}{v})\SEQN G\SEQN\norm{\sigma}\SEQN\norm{d}^+
    \\
  \shortintertext{\hfill{\footnotesize(Lemma~\ref{appendix:packet-alg}(15); well-formedness)}}
  &\equiv
    \bigoplus_{u\in\keys(b)}(\EQ{f}{u})\SEQN
      \bigoplus_{\mathclap{(v\mapsto\sigma)\in b(u)}}
      \hspace{10px}\begin{cases}
        (\ASSN{f}{v})\SEQN\norm{\sigma}\SEQN\norm{d}^+
          & v\notin\keys(b,m),\\
        \bot & \text{otherwise}
      \end{cases}
    \\
  \shortintertext{\hfill{\footnotesize(Lemma~\ref{appendix:packet-alg}(9), repetitively over $G$)}}
  &\equiv
    \bigoplus_{u\in\keys(b)}(\EQ{f}{u})\SEQN
      \bigoplus_{\mathclap{\substack{(v\mapsto\sigma)\in b(u)\\
        v\notin\keys(b,m)}}}
        (\ASSN{f}{v})\SEQN\norm{\sigma}\SEQN\norm{d}^+
    \\
  \shortintertext{\hfill{\footnotesize(Lemma~\ref{appendix:semiring-eqns}(3),(6))}}
  M_{\rho'};D_{\rho'}^+
  &\equiv
    G_b\SEQN
    \Bigl(\bigoplus_{\mathclap{(z\mapsto\sigma)\in m}}
      (\ASSN{f}{z})\SEQN\norm{\sigma}\SEQN G\SEQN\norm{d}^+\Bigr)
    \\
  \shortintertext{\hfill{\footnotesize(Definition~\ref*{def:wspp-sem}; $D_{\rho'}^+\equiv G;\norm{d}^+$)}}
  &\equiv
    G_b\SEQN
    \Bigl(\bigoplus_{\mathclap{(z\mapsto\sigma)\in m}}
      (\ASSN{f}{z})\SEQN G\SEQN\norm{\sigma}\SEQN\norm{d}^+\Bigr)
    \\
  \shortintertext{\hfill{\footnotesize(Lemma~\ref{appendix:packet-alg}(15); well-formedness)}}
  &\equiv
    G_b\SEQN
    \Bigl(\bigoplus_{\mathclap{(z\mapsto\sigma)\in m}}
      \bot\SEQN\norm{\sigma}\SEQN\norm{d}^+\Bigr)
    \\
  \shortintertext{\hfill{\footnotesize($z\in\keys(m)$; Lemma~\ref{appendix:packet-alg}(9))}}
  &\equiv \bot
    && \text{(Lemma~\ref{appendix:semiring-eqns}(3),(6)).}
\end{align*}
}
  Thus $M_{\rho'}\oplus M_{\rho'};D_{\rho'}^+\equiv M_{\rho'}$. For the
  branch component, regroup the original branch term with the new fall-through
  contributions:
{\small
\begin{align*}
B_{\rho'} \oplus B_{\rho'};D_{\rho'}^+
&\equiv
\bigoplus_{u\in\keys(b)}(\EQ{f}{u})\SEQN
\biggl(
  \bigoplus_{\mathclap{\substack{(v\mapsto\sigma)\in b(u)\\
    v\in\keys(b,m)}}}
    (\ASSN{f}{v})\SEQN\norm{\sigma}
\\[-0.1em]
&\hspace{8.5em}\oplus
  \bigoplus_{\mathclap{\substack{(v\mapsto\sigma)\in b(u)\\
    v\notin\keys(b,m)}}}
    (\ASSN{f}{v})\SEQN(\norm{\sigma}\oplus\norm{\sigma}\SEQN\norm{d}^+)
\biggr)
  && \text{(previous result, Lemma~\ref{appendix:semiring-eqns}(7))}\\
&\equiv
\bigoplus_{u\in\keys(b)}(\EQ{f}{u})\SEQN
\biggl(
  \bigoplus_{\mathclap{\substack{(v\mapsto\sigma)\in b(u)\\
    v\in\keys(b,m)}}}
    (\ASSN{f}{v})\SEQN\norm{\sigma}
\\[-0.1em]
&\hspace{8.5em}\oplus
  \bigoplus_{\mathclap{\substack{(v\mapsto\sigma)\in b(u)\\
    v\notin\keys(b,m)}}}
    (\ASSN{f}{v})\SEQN\norm{\sigma}\SEQN\ITER{\norm{d}}
\biggr)
  && \text{(Lemma~\ref{appendix:semiring-eqns}(12),(7))}\\
&\equiv
\bigoplus_{u\in\keys(b)}(\EQ{f}{u})\SEQN
\biggl(
  \bigoplus_{\mathclap{\substack{(v\mapsto\sigma)\in b(u)\\
    v\in\keys(b,m)}}}
    (\ASSN{f}{v})\SEQN\norm{\sigma}
\\[-0.1em]
&\hspace{8.5em}\oplus
  \bigoplus_{\mathclap{\substack{(v\mapsto\sigma)\in b(u)\\
    v\notin\keys(b,m)}}}
    (\ASSN{f}{v})\SEQN\norm{\sigma\wsppmul\Star(d)}
  \biggr)
\end{align*}
}
  For $v\notin\keys(b,m)$, Pass~I defines
  $b''(u)(v)=\sigma\wsppmul\Star(d)$, and the premise for $\Star(d)$ together
  with \Cref{theorem:wspp-add-mul-soundness} identifies the corresponding
  summand with $(\ASSN{f}{v});\norm{b''(u)(v)}$. For
  $v\in\keys(b,m)$, Pass~I leaves $b''(u)(v)=\sigma$. Therefore, by
  \Cref{def:wspp-sem},
  \[
    B_{\rho'} \oplus B_{\rho'};D_{\rho'}^+
    \equiv B_{\rho''}.
  \]

  Pass~I leaves $m$ unchanged and clears the default identity, so
  $M_{\rho'}\equiv M_{\rho''}$ and $D_{\rho''}\equiv\bot$. Combining the last
  two component calculations gives
  \[
    (B_{\rho'} \oplus B_{\rho'};D_{\rho'}^+)
    \oplus (M_{\rho'} \oplus M_{\rho'};D_{\rho'}^+)
    \equiv B_{\rho''}\oplus M_{\rho''}\oplus D_{\rho''}
    \equiv \norm{\rho''},
  \]
  by \Cref{def:wspp-sem} and Lemma~\ref{appendix:semiring-eqns}(3). Substituting
  this equality and $\norm{\psi}\equiv D_{\rho'}^*$ into
  \eqref{eq:pass-1-factorization} yields
  \[
    \ITER{\norm{\rho'}}
    \equiv
    \SEQ{\norm{\psi}}{\ITER{\norm{\rho''}}},
  \]
  which is the Pass~I preservation claim.

  We next prove the Pass~II item. At \AlgPoint{wsppbranch}{II-s}, Pass~I has
  already cleared the default identity, so write the current residual as
  $\rho'=\wspp(f,b,m,\hat{\mathbb{0}})$. Consider an arbitrary iteration that
  selects $u$, then $v$, and then $\sigma=b(u)(v)$, exactly as in the program. Let
  $b^{-}$ be the map obtained by deleting only this selected entry and retaining
  all first-level rows of $b$. Abbreviate
  \[
    s \triangleq (\EQ{f}{u})\SEQN(\ASSN{f}{v})\SEQN\norm{\sigma},
    \qquad
    s^+ \triangleq s;s^*,
    \qquad
    G_b \triangleq \prod_{\mathclap{x\in\keys(b)}}\NOTEQ{f}{x}.
  \]
  Let $B^{-}$ be the branch component obtained from $b^{-}$, explicitly
  \[
    B^{-} \triangleq
    \bigoplus_{x\in\keys(b)}(\EQ{f}{x})\SEQN
      \bigoplus_{y\mapsto\tau\in b^{-}(x)}
        (\ASSN{f}{y})\SEQN\norm{\tau}.
  \]
  The selected entry is one summand of $B_{\rho'}$, while $D_{\rho'}\equiv\bot$, so
  \[
    \norm{\rho'} \equiv s\oplus R,
    \qquad
    R \triangleq B^{-}\oplus M_{\rho'}.
  \]
  Denesting the selected summand gives the residual form that the algorithm
  must realize:
  \begin{align}
    \ITER{\norm{\rho'}}
    &\equiv \ITER{(s\oplus R)}
      && \text{(Definition~\ref*{def:wspp-sem}, Lemma~\ref{appendix:semiring-eqns}(1),(2),(3))}
      \notag\\
    &\equiv s^*;(R;s^*)^*
      && \text{(Lemma~\ref{appendix:semiring-eqns}(15))}
      \notag\\
    &\equiv (\top\oplus s^+);(R;(\top\oplus s^+))^*
      && \text{(Lemma~\ref{appendix:semiring-eqns}(12))}
      \notag\\
    &\equiv (\top\oplus s^+);(R\oplus R;s^+)^*.
      \label{eq:pass-2-factorization}
  \end{align}
  The last line uses Lemma~\ref{appendix:semiring-eqns}(5),(7). Thus the
  emitted factor should denote $\top\oplus s^+$, and the new residual should
  denote $R\oplus R;s^+$.

  We first compute $s^+$. Since $\sigma$ operates only on fields
  greater than $f$, $\norm{\sigma}$ commutes with tests and assignments on $f$ by
  Lemma~\ref{appendix:packet-alg}(15). If $v\neq u$, one traversal of the
  selected entry writes a value that cannot satisfy the next guard:
{\small
\begin{align*}
  s;s
  &\equiv
    (\EQ{f}{u})\SEQN(\ASSN{f}{v})\SEQN\norm{\sigma}\SEQN
    (\EQ{f}{u})\SEQN(\ASSN{f}{v})\SEQN\norm{\sigma}
    \\
  &\equiv
    (\EQ{f}{u})\SEQN(\ASSN{f}{v})\SEQN(\EQ{f}{u})\SEQN
    \norm{\sigma}\SEQN(\ASSN{f}{v})\SEQN\norm{\sigma}
    && \text{(Lemma~\ref{appendix:packet-alg}(15))}\\
  &\equiv \bot
    && \text{($v\neq u$; Lemma~\ref{appendix:packet-alg}(8),
      Lemma~\ref{appendix:semiring-eqns}(6)).}
\end{align*}
}
  Hence every power $s^{(k)}$ with $k\geq2$ is equivalent to $\bot$, and
  \begin{align*}
    s^+
    &\equiv s;(\top\oplus s;s^*)
      && \text{(Lemma~\ref{appendix:semiring-eqns}(12))}\\
    &\equiv s\oplus (s;s);s^*
      && \text{(Lemma~\ref{appendix:semiring-eqns}(5),(7))}\\
    &\equiv s
      && \text{(Lemma~\ref{appendix:semiring-eqns}(3),(6)).}
  \end{align*}

  If $v=u$, the selected entry is a self-loop. The same calculation, followed
  by the same $\omega$-continuity argument used in Pass~I, gives
  \[
    s^+
    \equiv
    (\EQ{f}{u})\SEQN(\ASSN{f}{u})\SEQN
    \norm{\sigma\wsppmul\Star(\sigma)},
  \]
  using Lemma~\ref{appendix:packet-alg}(8),(10),(15), the premise for
  $\Star(\sigma)$, and \Cref{theorem:wspp-add-mul-soundness}.
  Thus both cases are captured by the algorithm's definition
  \[
    \eta \triangleq
    \begin{cases}
      \sigma\wsppmul\Star(\sigma) & \text{if } v=u,\\
      \sigma & \text{if } v\neq u,
    \end{cases}
    \qquad
    s^+ \equiv
    (\EQ{f}{u})\SEQN(\ASSN{f}{v})\SEQN\norm{\eta}.
    \label{eq:pass-2-selected-closure}
  \]
  Consequently the factor emitted at the end of the iteration,
  \[
    \psi=\hat{\mathbb{1}}\wsppadd
      \strhat{(f{=}u)}\wsppmul\strhat{(f{\leftarrow}v)}\wsppmul\eta,
  \]
  satisfies $\norm{\psi}\equiv\top\oplus s^+\equiv s^*$ by
  \Cref{theorem:wspp-add-mul-soundness} and
  Lemma~\ref{appendix:semiring-eqns}(12).

  We now compute the second compensation, $R;s^+$, directly from the expanded
  forms of its two components. For the explicit branches,
{\small
\begin{align}
  B^{-};s^+
  &\equiv
    \left(
    \bigoplus_{x\in\keys(b)}(\EQ{f}{x})\SEQN
      \bigoplus_{y\mapsto\tau\in b^{-}(x)}
        (\ASSN{f}{y})\SEQN\norm{\tau}
    \right)\SEQN s^+
    \notag\\
  &\equiv
    \bigoplus_{x\in\keys(b)}
      \bigoplus_{y\mapsto\tau\in b^{-}(x)}
      (\EQ{f}{x})\SEQN(\ASSN{f}{y})\SEQN\norm{\tau}\SEQN s^+
      \notag\\
  &\equiv
    \bigoplus_{x\in\keys(b)}
      \bigoplus_{y\mapsto\tau\in b^{-}(x)}
      (\EQ{f}{x})\SEQN(\ASSN{f}{y})\SEQN\norm{\tau}\SEQN
      (\EQ{f}{u})\SEQN(\ASSN{f}{v})\SEQN\norm{\eta}
      \notag\\[-0.2em]
  &\hspace{20em}
      \text{(\eqref{eq:pass-2-selected-closure})}
      \notag\\
  &\equiv
    \bigoplus_{x\in\keys(b)}
      \bigoplus_{y\mapsto\tau\in b^{-}(x)}
      (\EQ{f}{x})\SEQN(\ASSN{f}{y})\SEQN(\EQ{f}{u})\SEQN
      \norm{\tau}\SEQN(\ASSN{f}{v})\SEQN\norm{\eta}
      \notag\\[-0.2em]
  &\hspace{20em}
      \text{(Lemma~\ref{appendix:packet-alg}(15))}
      \notag\\
  &\equiv
    \bigoplus_{x\in\keys(b)}
      \bigoplus_{y\mapsto\tau\in b^{-}(x)}
      \begin{cases}
        (\EQ{f}{x})\SEQN(\ASSN{f}{u})\SEQN
        \norm{\tau}\SEQN(\ASSN{f}{v})\SEQN\norm{\eta}
          & \text{if } y=u,\\
        \bot & \text{if } y\neq u
      \end{cases}
      \notag\\[-0.2em]
  &\hspace{20em}
      \text{(Lemma~\ref{appendix:packet-alg}(8),
        Lemma~\ref{appendix:semiring-eqns}(6))}
      \notag\\
  &\equiv
    \bigoplus_{x\in\keys(b)}
      \bigoplus_{y\mapsto\tau\in b^{-}(x)}
      \begin{cases}
        (\EQ{f}{x})\SEQN(\ASSN{f}{v})\SEQN\norm{\tau\wsppmul\eta}
          & \text{if } y=u,\\
        \bot & \text{if } y\neq u
      \end{cases}
      \notag\\[-0.2em]
  &\hspace{20em}
      \text{(Lemma~\ref{appendix:packet-alg}(7),(15);
        \Cref{theorem:wspp-add-mul-soundness})}
      \notag\\
  &\equiv
    \bigoplus_{\mathclap{\substack{x\in\keys(b)\\ u\in\keys(b^{-}(x))}}}
      (\EQ{f}{x})\SEQN(\ASSN{f}{v})\SEQN
      \norm{b^{-}(x)(u)\wsppmul\eta}.
      \label{eq:pass-2-branch-comp}
\end{align}
}
  The default-assignment component is analogous, but there is at most one
  default column that writes $u$:
{\small
\begin{align}
  M_{\rho'};s^+
  &\equiv
    \left(
      G_b\SEQN
      \bigoplus_{y\mapsto\tau\in m}(\ASSN{f}{y})\SEQN\norm{\tau}
    \right)\SEQN s^+
    \notag\\
  &\equiv
    \bigoplus_{y\mapsto\tau\in m}
      G_b\SEQN(\ASSN{f}{y})\SEQN\norm{\tau}\SEQN s^+
    \notag\\
  &\equiv
    \bigoplus_{y\mapsto\tau\in m}
      G_b\SEQN(\ASSN{f}{y})\SEQN\norm{\tau}\SEQN
      (\EQ{f}{u})\SEQN(\ASSN{f}{v})\SEQN\norm{\eta}
    \notag\\[-0.2em]
  &\hspace{20em}
      \text{(\eqref{eq:pass-2-selected-closure})}
    \notag\\
  &\equiv
    \bigoplus_{y\mapsto\tau\in m}
      G_b\SEQN(\ASSN{f}{y})\SEQN(\EQ{f}{u})\SEQN
      \norm{\tau}\SEQN(\ASSN{f}{v})\SEQN\norm{\eta}
    \notag\\[-0.2em]
  &\hspace{20em}
      \text{(Lemma~\ref{appendix:packet-alg}(15))}
    \notag\\
  &\equiv
    \bigoplus_{y\mapsto\tau\in m}
    \begin{cases}
      G_b\SEQN(\ASSN{f}{u})\SEQN
      \norm{\tau}\SEQN(\ASSN{f}{v})\SEQN\norm{\eta}
        & \text{if } y=u,\\
      \bot & \text{if } y\neq u
    \end{cases}
    \notag\\[-0.2em]
  &\hspace{20em}
      \text{(Lemma~\ref{appendix:packet-alg}(8),
        Lemma~\ref{appendix:semiring-eqns}(6))}
    \notag\\
  &\equiv
    \bigoplus_{y\mapsto\tau\in m}
    \begin{cases}
      G_b\SEQN(\ASSN{f}{v})\SEQN\norm{\tau\wsppmul\eta}
        & \text{if } y=u,\\
      \bot & \text{if } y\neq u
    \end{cases}
    \notag\\[-0.2em]
  &\hspace{20em}
      \text{(Lemma~\ref{appendix:packet-alg}(7),(15);
        \Cref{theorem:wspp-add-mul-soundness})}
    \notag\\
  &\equiv
    \begin{cases}
      G_b\SEQN(\ASSN{f}{v})\SEQN\norm{m(u)\wsppmul\eta}
        & \text{if } u\in\keys(m),\\
      \bot & \text{otherwise.}
    \end{cases}
    \label{eq:pass-2-default-comp}
\end{align}
}

Let the residual wSPP at
  \AlgPoint{wsppbranch}{II-e} be
  $\rho''=\wspp(f,b'',m'',\hat{\mathbb{0}})$. We now need to show that the equations above 
  precisely correspond to the contribution of the component of $b''$ and $m''$ (i.e. $B_{\rho''} and M_{\rho''}$)
  to $\norm{\rho''}$. 
  Unfolding the program update,
  the maps already present in this residual satisfy
  \[
    b''(x)(y) =
    \begin{cases}
      \dflt{b^{-}(x)(y)} \wsppadd b^{-}(x)(u)\wsppmul\eta
        & \text{if } y=v \text{ and } u\in\keys(b^{-}(x)),\\
      \dflt{b^{-}(x)(y)}
        & \text{otherwise,}
    \end{cases}
  \]
  and
  \[
    m''(y) =
    \begin{cases}
      \dflt{m(y)} \wsppadd m(u)\wsppmul\eta
        & \text{if } y=v \text{ and } u\in\keys(m),\\
      \dflt{m(y)}
        & \text{otherwise.}
    \end{cases}
  \]
  These equations are precisely the program's case analysis. It first deletes
  $(u,v)$, producing $b^{-}$. Then the branch-update loop visits exactly those
  rows $x$ with $u\in\keys(b^{-}(x))$ and changes only the output column $v$;
  all other entries are the defaulted old entries. The default-update line
  changes only $m(v)$, and only in the case $u\in\keys(m)$. Finally, the
  algorithm's branch $v=u$ versus $v\neq u$ is exactly the definition of
  $\eta$ in \eqref{eq:pass-2-selected-closure}. Therefore,
  \[
    B^{-}\oplus B^{-};s^+ \equiv B_{\rho''},
    \qquad
    M_{\rho'}\oplus M_{\rho'};s^+ \equiv M_{\rho''}.
  \]
  Since Pass~II keeps the default identity equal to $\hat{\mathbb{0}}$, the
  default-identity component of $\rho''$ is still $\bot$. Thus
  \[
    R\oplus R;s^+
    \equiv
    (B^{-}\oplus B^{-};s^+)\oplus(M_{\rho'}\oplus M_{\rho'};s^+)
    \equiv \norm{\rho''},
  \]
  by Lemma~\ref{appendix:semiring-eqns}(1),(2),(3) and
  \Cref{def:wspp-sem}. Substituting this and
  $\norm{\psi}\equiv s^*$ into \eqref{eq:pass-2-factorization} proves
  \[
    \ITER{\norm{\rho'}}
    \equiv \SEQ{\norm{\psi}}{\ITER{\norm{\rho''}}},
  \]
  which is the Pass~II preservation claim.

  It remains to prove the Pass~III item. At \AlgPoint{wsppmissing}{III-s},
  all explicit branches have already been removed, but the first-level rows of
  $b$ are still retained. Thus write the current residual as
  $\rho'=\wspp(f,b,m,\hat{\mathbb{0}})$ with $b(x)=\varnothing$ for every
  $x\in\keys(b)$. Consider an arbitrary iteration that selects $z\in\keys(m)$
  and $\sigma=m(z)$, and let $m^{-}$ be $m$ with only the entry $z$ deleted.
  Abbreviate
  \[
    G_b \triangleq \prod_{\mathclap{x\in\keys(b)}}\NOTEQ{f}{x},
    \qquad
    s \triangleq G_b\SEQN(\ASSN{f}{z})\SEQN\norm{\sigma},
    \qquad
    s^+ \triangleq s;s^*.
  \]
  Let $M^{-}$ be the default-assignment component computed from $m^{-}$. Since
  $B_{\rho'}\equiv D_{\rho'}\equiv\bot$,
  \[
    \norm{\rho'} \equiv M^{-}\oplus s.
  \]
  Pass~III emits its factor on the right, so we use the right-denesting
  identity:
  \begin{align}
    \ITER{\norm{\rho'}}
    &\equiv \ITER{(s\oplus M^{-})}
      && \text{(Definition~\ref*{def:wspp-sem}, Lemma~\ref{appendix:semiring-eqns}(1),(2),(3))}
      \notag\\
    &\equiv (s^*;M^{-})^*;s^*
      && \text{(Lemma~\ref{appendix:semiring-eqns}(14))}
      \notag\\
    &\equiv ((\top\oplus s^+);M^{-})^*;s^*
      && \text{(Lemma~\ref{appendix:semiring-eqns}(12))}
      \notag\\
    &\equiv (M^{-}\oplus s^+;M^{-})^*;s^*.
      \label{eq:pass-3-factorization}
  \end{align}
  The last line uses Lemma~\ref{appendix:semiring-eqns}(5),(7). Hence the
  residual must denote $M^{-}\oplus s^+;M^{-}$, and the emitted factor must
  denote $s^*$.

  We first compute $s^+$. Since $\sigma$ operates only on fields greater than $f$,
  $\norm{\sigma}$ commutes with tests and assignments on $f$ by
  Lemma~\ref{appendix:packet-alg}(15). If $z\in\keys(b)$, the next copy of
  $s$ is blocked by the retained explicit row:
{\small
\begin{align*}
  s;s
  &\equiv
    G_b\SEQN(\ASSN{f}{z})\SEQN\norm{\sigma}\SEQN
    G_b\SEQN(\ASSN{f}{z})\SEQN\norm{\sigma}
    \\
  &\equiv
    G_b\SEQN(\ASSN{f}{z})\SEQN G_b\SEQN
    \norm{\sigma}\SEQN(\ASSN{f}{z})\SEQN\norm{\sigma}
    && \text{(Lemma~\ref{appendix:packet-alg}(15))}\\
  &\equiv \bot
    && \text{($z\in\keys(b)$; Lemma~\ref{appendix:packet-alg}(9),
      Lemma~\ref{appendix:semiring-eqns}(6)).}
\end{align*}
}
  Therefore, as in Pass~II,
  \begin{align*}
    s^+
    &\equiv s;(\top\oplus s;s^*)
      && \text{(Lemma~\ref{appendix:semiring-eqns}(12))}\\
    &\equiv s\oplus (s;s);s^*
      && \text{(Lemma~\ref{appendix:semiring-eqns}(5),(7))}\\
    &\equiv s
      && \text{(Lemma~\ref{appendix:semiring-eqns}(3),(6)).}
  \end{align*}

  If $z\notin\keys(b)$, the assignment to $z$ passes every test in $G_b$, so
  the selected default assignment can loop. By similar reasoning to the
  self-loop case in Pass~II, again using $\omega$-continuity and the premise
  for $\Star(\sigma)$,
  \[
    s^+
    \equiv
    G_b\SEQN(\ASSN{f}{z})\SEQN\norm{\sigma\wsppmul\Star(\sigma)}.
  \]
  Combining the two cases, the algorithm's emitted factor $\phi$ satisfies
  $\norm{\phi}\equiv\top\oplus s^+\equiv s^*$: if $z\in\keys(b)$, then
  $\phi=\hat{\mathbb{1}}\wsppadd
  \guardprod{x\in\keys(b)}{(f\ne x)}\wsppmul
  \strhat{(f{\leftarrow}z)}\wsppmul \sigma$; otherwise the same expression uses
  $\sigma\wsppmul\Star(\sigma)$ in place of $\sigma$.

  It remains to identify the residual $M^{-}\oplus s^+;M^{-}$. If
  $z\in\keys(b)$, then $s^+=s$ and
{\small
\begin{align*}
  s^+;M^{-}
  &\equiv
    G_b\SEQN(\ASSN{f}{z})\SEQN\norm{\sigma}\SEQN
    G_b\SEQN
    \Bigl(\bigoplus_{\mathclap{(y\mapsto\tau)\in m^{-}}}
      (\ASSN{f}{y})\SEQN\norm{\tau}\Bigr)
    \\
  \shortintertext{\hfill{\footnotesize(Definition~\ref*{def:wspp-sem}; $s^+=s$)}}
  &\equiv \bot
    && \text{($z\in\keys(b)$; Lemma~\ref{appendix:packet-alg}(15),(9);
      Lemma~\ref{appendix:semiring-eqns}(6)).}
\end{align*}
}
  Thus the residual is obtained simply by deleting $z$ from $m$, which is
  exactly the first branch of Pass~III.

  If $z\notin\keys(b)$, then $s^+=G_b;(\ASSN{f}{z});\norm{\sigma\wsppmul\Star(\sigma)}$.
  The contribution through the selected column and then a remaining column is:
{\small
\begin{align*}
  s^+;M^{-}
  &\equiv
    G_b\SEQN(\ASSN{f}{z})\SEQN\norm{\sigma\wsppmul\Star(\sigma)}\SEQN
    G_b\SEQN
    \Bigl(\bigoplus_{\mathclap{(y\mapsto\tau)\in m^{-}}}
      (\ASSN{f}{y})\SEQN\norm{\tau}\Bigr)
    \\
  \shortintertext{\hfill{\footnotesize(Definition~\ref*{def:wspp-sem}; closed form of $s^+$)}}
  &\equiv
    G_b\SEQN
    \Bigl(\bigoplus_{\mathclap{(y\mapsto\tau)\in m^{-}}}
      (\ASSN{f}{y})\SEQN\norm{\sigma\wsppmul\Star(\sigma)}\SEQN\norm{\tau}\Bigr)
    \\
  \shortintertext{\hfill{\footnotesize($z\notin\keys(b)$; Lemma~\ref{appendix:packet-alg}(7),(15),(9))}}
  M^{-}\oplus s^+;M^{-}
  &\equiv
    G_b\SEQN
    \Bigl(\bigoplus_{\mathclap{(y\mapsto\tau)\in m^{-}}}
      (\ASSN{f}{y})\SEQN
      \bigl(\norm{\tau}\oplus
        \norm{\sigma\wsppmul\Star(\sigma)}\SEQN\norm{\tau}\bigr)\Bigr)
    \\
  \shortintertext{\hfill{\footnotesize(Lemma~\ref{appendix:semiring-eqns}(1),(2),(7))}}
  &\equiv
    G_b\SEQN
    \Bigl(\bigoplus_{\mathclap{(y\mapsto\tau)\in m^{-}}}
      (\ASSN{f}{y})\SEQN\norm{\Star(\sigma)\wsppmul\tau}\Bigr).
    \quad
    \text{(Lemma~\ref{appendix:semiring-eqns}(12),(7);
      \Cref{theorem:wspp-add-mul-soundness})}
\end{align*}
}
   This reflects exactly the second branch of Pass~III, which deletes $z$ and then
  replaces every remaining entry $\tau=m^{-}(y)$ by
  $\Star(\sigma)\wsppmul\tau$. Let $m''$ denote the map produced by the appropriate
  branch of the program, and set
  $\rho''=\wspp(f,b,m'',\hat{\mathbb{0}})$. In both cases,
  \[
    M^{-}\oplus s^+;M^{-}\equiv M_{\rho''} \equiv \norm{\rho''}.
  \]
  Substituting this equality and $\norm{\phi}\equiv s^*$ into
  \eqref{eq:pass-3-factorization} proves
  \[
    \ITER{\norm{\rho'}}
    \equiv \SEQ{\ITER{\norm{\rho''}}}{\norm{\phi}},
  \]
  which is the Pass~III preservation claim.
\end{proof}

%% file: embedding.tex

\Cref{theorem:wspp-main} computes the semantics of policies relative to an
embedding $\emb \colon \csr \hookrightarrow \semi$
(\Cref{def:embedding}). This appendix collects embeddings for the
$\omega$-continuous semirings most common in network verification
(\Cref{fig:embedding-portfolio}). Producing an embedding means finding a
closed form for the semantic star on the representable weights---the
star-coherence obligation (iii) of \Cref{def:embedding} is its entire
mathematical content. Three arguments of increasing difficulty cover the
portfolio.

\paragraph{Bounded semirings}
For the \emph{bounded} semirings (\Cref{def:bounded})---those whose unit
$\semione$ is the greatest element---the
star is constantly $\semione$ by \Cref{lemma:bounded-star}, so obligation
(iii) is discharged by the definition $\cstar{a} \triangleq \semione$, and
any computable subsemiring containing $\semione$ embeds: the entire
semiring when it is computable to begin with (bottleneck, security), its
rational part otherwise (Viterbi, tropical).

\begin{example}[geometric series]
  \label{example:embedding-geometric}
  For the probability semiring $\Rext$ of
  \Cref{example:embedding-tension} no such shortcut exists: its star
  genuinely sums. Over $\Qext$, define $\cstar{a} \triangleq (1-a)^{-1}$
  if $a < 1$ and $\cstar{a} \triangleq \infty$ otherwise. Condition (iii)
  is the geometric series: for $a < 1$ the partial sums
  $\sum_{i \leq k} a^i = \tfrac{1-a^{k+1}}{1-a}$ converge to $(1-a)^{-1}$,
  and for $a \geq 1$ they exceed every bound. Thus
  $\Qext \hookrightarrow \Rext$---even though $\Qext$ is not
  $\omega$-continuous. \hfill $\triangle$
\end{example}

\begin{example}[failure rates]
  \label{example:embedding-probunion}
  The probabilistic-union semiring
  $([0,1] \cup \{-\infty\}, \max, \uplus, -\infty, 0)$ with
  $\wta \uplus \wtb = \wta + \wtb - \wta \cdot \wtb$ models failure rates.
  It is not bounded---its unit $0$ is not greatest---and its star is not a
  rational function either. The $\uplus$-powers of $a \in [0,1]$ satisfy
  $a^{\uplus k} = 1 - (1-a)^k$ and form an ascending chain above
  $a^{\uplus 0} = 0$, so the star, i.e., the supremum of the $\max$-partial
  sums, is $\sup_k \bigl(1-(1-a)^k\bigr)$. Over the computable carrier
  $\csrset = (\mathbb{Q} \cap [0,1]) \cup \{-\infty\}$ we accordingly define
  $\cstar{a} \triangleq 0$ for $a \leq 0$ and $\cstar{a} \triangleq 1$
  otherwise, and (iii) reduces to the observation that $(1-a)^k \to 0$ for
  $a > 0$. \hfill $\triangle$
\end{example}

\begin{figure}[t]
\centering
\footnotesize
\setlength{\tabcolsep}{4.5pt}
\resizebox{\textwidth}{!}{%
\begin{tabular}{@{}lllll@{}}
\toprule
Quantity & Semantic $\semi$ & Computable $\csr$ & $\cstar{a}$ & Obligation (iii)\\
\midrule
bandwidth (bottleneck) & $(\N\cup\{\pm\infty\},\max,\min,-\infty,\infty)$ & $\semi$ itself & $\infty$ & Lemma~\ref{lemma:bounded-star}\\
security levels & $(\{0,L,M,H\},\max,\min,0,H)$ & $\semi$ itself & $H$ & Lemma~\ref{lemma:bounded-star}\\
reliability (Viterbi) & $([0,1],\max,\cdot\,,0,1)$ & $\mathbb{Q}\cap[0,1]$ & $1$ & Lemma~\ref{lemma:bounded-star}\\
latency (tropical) & $([0,\infty],\min,+,\infty,0)$ & $\mathbb{Q}_{\geq0}\cup\{\infty\}$ & $0$ & Lemma~\ref{lemma:bounded-star}\\
failure rate (prob.\ union) & $([0,1]\cup\{-\infty\},\max,\uplus,-\infty,0)$ & $(\mathbb{Q}\cap[0,1])\cup\{-\infty\}$ & $0$ if $a \leq 0$, else $1$ & Example~\ref{example:embedding-probunion}\\
expected cost & $([0,\infty],+,\cdot\,,0,1)$ & $\mathbb{Q}_{\geq0}\cup\{\infty\}$ & $(1-a)^{-1}$ if $a<1$, else $\infty$ & Example~\ref{example:embedding-geometric}\\
\midrule
multi-objective, traces & downward-closed sets & finite (trace) frontiers & $\semione$, resp.\ $\semione + a$ & \Cref{sec:pareto}\\
\bottomrule
\end{tabular}%
}
\caption{\footnotesize Embeddings $\emb : \csr \hookrightarrow \semi$ of
computable star semirings into the $\omega$-continuous semirings of common
network quantities (cf.~\cite{wNetKAT}). In the first six rows $\emb$ is the
evident inclusion. The latency semiring is ordered by $\geq$, under which its
unit $0$ is the greatest element.}
\Description{A table listing, for each network quantity, the omega-continuous
semantic semiring, a computable star semiring embedding into it, the closed
form of the computable star, and the argument discharging the star-coherence
obligation.}
\label{fig:embedding-portfolio}
\end{figure}

In all rows of \Cref{fig:embedding-portfolio} except the last, $\csr$ is a
subsemiring of $\semi$ and $\emb$ an inclusion: codes denote themselves.
The multi-objective and trace-carrying weights of \Cref{sec:pareto} exploit
the full generality of \Cref{def:embedding}: there the codes (finite Pareto
frontiers) are genuinely different objects from the semantic weights
(infinite downward-closed sets), related by the decoding map
$\dcl{(-)}$.

%% file: synthetic.tex
\Cref{fig:f10-vs-jellyfish} compares the running time of \wspptool on F10 and
Jellyfish topologies in both Pareto and trace-carrying Pareto modes.
The running time of trace-carrying Pareto is only marginally higher than 
Pareto at around $5\%$ on average. This is due to the fact that
these networks, while having a large number of paths (400,000+ on larger networks), 
do not have wide Pareto frontiers (1.7 on average).

\begin{figure}[h]
    \centering
    \includegraphics[width=0.5\textwidth]{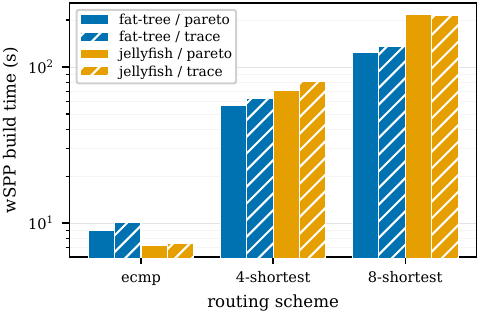}
    \caption{Overhead of trace-carrying Pareto on F10 and Jellyfish topologies.}
    \label{fig:f10-vs-jellyfish}
\end{figure}

To investigate the scaling behavior of trace-carrying Pareto, we generate 
synthetic adversarial benchmarks which produce wide Pareto frontiers.
\Cref{fig:trace-overhead} shows the running time and peak memory usage of 
Pareto and trace-carrying Pareto on these benchmarks.
Here we can see that the overhead of trace-carrying scales is more significantly
influenced by the width of the Pareto frontier than the raw number of switches 
in the network.

\begin{figure}[h]
    \centering
    \par\smallskip
    \begin{subfigure}{0.49\textwidth}
        \centering
        \includegraphics[width=\textwidth]{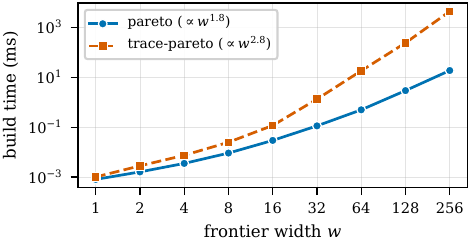}
        \caption{Running time.}
        \label{fig:trace-time}
    \end{subfigure}
    \hfill
    \begin{subfigure}{0.49\textwidth}
        \centering
        \includegraphics[width=\textwidth]{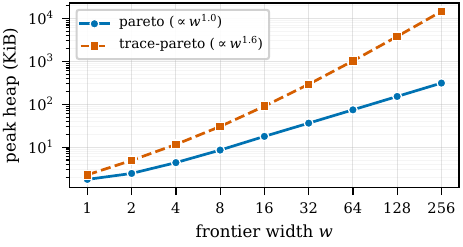}
        \caption{Peak memory usage.}
        \label{fig:trace-mem}
    \end{subfigure}
    \par\smallskip
    \begin{subfigure}{0.49\textwidth}
        \centering
        \includegraphics[width=\textwidth]{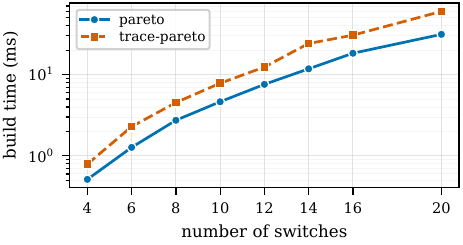}
        \caption{Running time.}
        \label{fig:trace-time}
    \end{subfigure}
    \hfill
    \begin{subfigure}{0.49\textwidth}
        \centering
        \includegraphics[width=\textwidth]{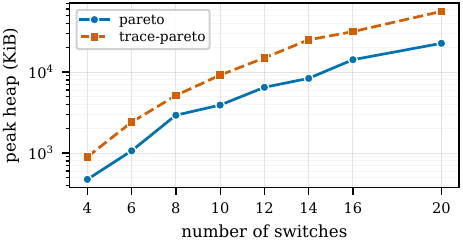}
        \caption{Peak memory usage.}
        \label{fig:trace-mem}
    \end{subfigure}
    \caption{Trace scaling: running time (\subref{fig:trace-time})
    and peak memory usage (\subref{fig:trace-mem}) of Pareto and trace-carrying Pareto}
    \label{fig:trace-overhead}
\end{figure}